\documentclass[10pt,romanappendices]{IEEEtran}

\usepackage[english]{babel} 
\usepackage{amsmath}                
\usepackage{amsfonts}               
\usepackage{amssymb}                
\usepackage{amsopn}                 
\allowdisplaybreaks[3] 
\usepackage{bbm}                    
\usepackage{mathrsfs}               
\usepackage{calc}                   
\usepackage[dvips]{graphicx}        
\usepackage{epsfig}                
\usepackage{psfrag}                 
\usepackage[dvips]{color}           
\usepackage{fancyhdr}               
\usepackage{verbatim}               
\usepackage{exscale}                

\usepackage{macros}

\newtheorem{theorem}{Theorem}
\newtheorem{lemma}{Lemma}
\newtheorem{corollary}{Corollary}
\newtheorem{proposition}{Proposition}

\newtheorem{definition}{Definition}

\renewcommand{\const}[1]{\mathcal{#1}}
\renewcommand{\WW}{W}

\title{At Low SNR Asymmetric Quantizers Are Better}

\author{Tobias~Koch,~\IEEEmembership{Member,~IEEE,}~and~Amos Lapidoth,~\IEEEmembership{Fellow,~IEEE}%
\thanks{T.~Koch has received funding from the European Community's Seventh Framework Programme (FP7/2007-2013) under grant agreement No. 252663 and from the Ministerio de
  Econom\'ia of Spain (projects DEIPRO, id.~TEC2009-14504-C02-01,
  and COMONSENS, id.~CSD2008-00010). The material in this paper was presented in part at the IEEE International Symposium on Information Theory (ISIT), St. Petersburg, Russia, July 31 -- August 5, 2011 and at the International Zurich Seminar on Communications (IZS), Zurich, Switzerland, February 29 -- March 2, 2012.}
 \thanks{T.~Koch was with the Department of Engineering, University of Cambridge, Cambridge CB2 1PZ, UK. He is now with the Signal Theory and Communications Department, Universidad Carlos III de Madrid, 28911 Legan\'es, Spain (e-mail: koch@tsc.uc3m.es).}
 \thanks{A.~Lapidoth is with the Department of Information Technology and Electrical Engineering, ETH Zurich, 8092 Zurich, Switzerland (e-mail: lapidoth@isi.ee.ethz.ch).}}

\begin{document}

\maketitle

\begin{abstract}
  We study the capacity of the discrete-time Gaussian channel when its
  output is quantized with a one-bit quantizer. We focus on the low
  signal-to-noise ratio (SNR) regime, where communication at very low
  spectral efficiencies takes place. In this regime a symmetric
  threshold quantizer is known to reduce channel capacity by a factor of $2/\pi$,
  i.e., to cause an asymptotic power loss of approximately two
  decibels. Here it is shown that this power loss can be
  avoided by using asymmetric threshold quantizers and asymmetric
  signaling constellations.  To avoid this power loss,
  flash-signaling input distributions are essential. Consequently,
  one-bit output quantization of the Gaussian channel reduces
  spectral efficiency.

  Threshold quantizers are not only asymptotically optimal: 
  at every fixed SNR a threshold quantizer maximizes capacity
  among all one-bit output quantizers.

  The picture changes on the Rayleigh-fading channel. In the
  noncoherent case a one-bit output quantizer causes an
  unavoidable low-SNR asymptotic power loss. In the coherent case,
  however, this power loss is avoidable provided that we allow the
  quantizer to depend on the fading level.
\end{abstract}


\section{Introduction}
\label{sec:intro}
\IEEEPARstart{W}{e} study the effect on channel capacity of quantizing
the output of the discrete-time average-power-limited Gaussian channel
using a one-bit quantizer. This problem arises in communication
systems where the receiver uses digital signal processing techniques,
which require that the analog received signal be quantized using an
analog-to-digital converter (ADC). For ADCs with high resolution, the
effects of quantization are negligible. However, high-resolution ADCs
may not be practical when the bandwidth of the communication system is
large and the sampling rate high \cite{walden99}. In such scenarios, low-resolution ADCs must be used. The capacity of the discrete-time Gaussian channel with one-bit output quantization indicates what communication rates can be achieved when the receiver employs a low-resolution ADC.

We focus on the low signal-to-noise ratio (SNR) regime, where communication at low spectral efficiencies takes place, as in Spread-Spectrum and Ultra-Wideband communications. In this regime, a symmetric threshold quantizer\footnote{A threshold quantizer produces $1$ if its input is above a threshold, and it produces $0$ if it is not. A symmetric threshold quantizer is a threshold quantizer whose threshold is zero.} reduces the capacity by a factor of $2/\pi$, corresponding to a 2dB power loss \cite{viterbiomura79}. Hence the rule of thumb that ``hard decisions cause a 2dB power loss." Here we demonstrate that if we allow for \emph{asymmetric threshold quantizers} with corresponding \emph{asymmetric signal constellations}, then the two decibels can be fully recovered.

This result shows that a threshold (but not necessarily symmetric) quantizer is asymptotically optimal as the SNR tends to zero. We further show that this is not only true asymptotically: for any fixed SNR a threshold quantizer is optimal among all one-bit output quantizers.

While quantizing the output of the Gaussian channel with a
one-bit quantizer does not cause a loss with respect to the low-SNR
asymptotic capacity, it does cause a significant loss with respect to
the spectral efficiency. Indeed, as we show, the low-SNR asymptotic capacity of the quantized Gaussian channel can only be achieved by
\emph{flash-signaling} input distributions \cite[Def.~2]{verdu02}. For the
Gaussian channel (even without output quantization), such input
distributions result in poor spectral efficiency
\cite[Th.~16]{verdu02}: Gaussian inputs or (at low SNR) binary antipodal inputs yield much higher spectral efficiencies \cite[Th.~11]{verdu02}.  Since output quantization cannot increase the spectral efficiency, it follows that flash signaling results in poor spectral efficiency also on the quantized Gaussian channel.
Thus, at low SNR, the Gaussian channel with optimal one-bit output quantization has poor
spectral efficiency.

It should be noted that the discrete-time channel model that we consider implicitly assumes that the channel output is sampled at Nyquist rate. While sampling the output at Nyquist rate incurs no loss in capacity for the additive white Gaussian noise (AWGN) channel \cite{shannon48,gallager68}, it is not necessarily optimal (with respect to capacity) when the channel output is first quantized using a one-bit quantizer. In fact, when a symmetric threshold quantizer is employed, sampling the output above the Nyquist rate increases the low-SNR asymptotic capacity \cite{kochlapidoth10_2}, \cite{kochlapidoth10_1_arxiv} and it increases the capacity in the noiseless case \cite{gilbert93,shamai94}.

The rest of the paper is organized as follows. Section~\ref{sec:channel} introduces the channel model and defines the capacity as well as the capacity per unit-energy. Section~\ref{sec:results} presents the paper's main results. Section~\ref{sec:ppm} demonstrates that the capacity per unit-energy can be achieved by pulse-position modulation (PPM). Section~\ref{sec:EbNo} discusses the implications of our results on the spectral efficiency. Section~\ref{sec:noncoherent} studies the effect on the capacity per unit-energy of quantizing the output of the Rayleigh-fading channel using a one-bit quantizer. Sections~\ref{sec:dobi_thmamos} through \ref{sec:noncoherent_proofs} contain the proofs of our results: Section~\ref{sec:dobi_thmamos} contains the proofs concerning channel capacity, Section~\ref{sec:CUE_proofs} contains the proofs concerning the capacity per unit-energy, Section~\ref{sec:dobi_PP} contains the proofs concerning peak-power-limited channels, and Section~\ref{sec:noncoherent_proofs} contains the proofs concerning Rayleigh-fading channels. Section~\ref{sec:conclusion} concludes the paper with a summary and a discussion.

\section{Channel Model and Capacity}
\label{sec:channel}
\begin{figure}[t]
\centering
\psfrag{X}[cb][cb]{\small $X_k$}
\psfrag{M}[rc][rc]{\small $M$}
 \psfrag{Z}[cc][cc]{\small $Z_k$}
\psfrag{encoder}[cc][cc]{\small encoder}
\psfrag{decoder}[cc][cc]{\small decoder}
\psfrag{quantizer}[cc][cc]{\small quantizer}
\psfrag{Y}[cb][cb]{\small $Y_k$}
\psfrag{N}[lc][lc]{\small $\hat{M}$}
\psfrag{U}[cb][cb]{\small $\tilde{Y}_k$}
\epsfig{file=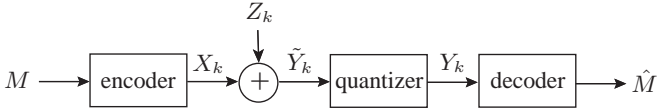, width=0.48\textwidth}
 \caption{System model.}
 \label{fig1}
\end{figure}
We consider the discrete-time communication system depicted in
Figure~\ref{fig1}.  A message~$M$, which is uniformly distributed over
the set $\{1,2,\ldots,\const{M}\}$, is mapped by an encoder to the
length-$n$ real sequence $X_1,X_2,\ldots,X_n \in \Reals$ of channel
inputs. (Here $\Reals$ denotes the set of real numbers.) The channel corrupts this sequence by adding white Gaussian noise to produce the unquantized output sequence
\begin{equation}
\label{eq:channel}
\tilde{Y}_k = X_k + Z_k, \quad k\in\Integers
\end{equation}
where $\{Z_k,\,k\in\Integers\}$ is a sequence of independent and
identically distributed (i.i.d.) Gaussian random variables of zero
mean and variance $\sigma^2$. (Here $\Integers$ denotes the set of
integers.) The unquantized output sequence is then quantized using a
quantizer that is specified by a Borel subset $\set{D}$ of the reals:
it produces $1$ if $\tilde{Y}_k$ is in $\set{D}$ and
produces $0$ if it is not. Denoting the time-$k$ quantizer output by
$Y_{k}$,
\begin{equation*}
  Y_{k} = \left\{\begin{array}{ll} 
    1 \quad & \text{if $\tilde{Y}_k \in \set{D}$,} \\
    0 \quad & \text{if $\tilde{Y}_k \notin \set{D}$.}\end{array} \right.
\end{equation*}
While we only consider deterministic quantizers, it should be noted that our results continue to hold if we allow for randomized quantization rules, i.e., if the quantizer produces $Y_k$ according to some probability distribution $P_{Y|\tilde{Y}}$ with binary $Y$. 

In view of the direct relationship between the set $\set{D}$ and the
quantizer it defines, we shall sometimes abuse notation and refer to
$\set{D}$ as the quantizer.  An example of a one-bit quantizer is the
\emph{threshold quantizer}, which corresponds to the set
\begin{equation}
\label{eq:dobi_thresholdD}
\set{D} = \{\tilde{y}\in\Reals\colon \tilde{y}\geq \Upsilon\}, \quad \Upsilon\in\Reals.
\end{equation}
The decoder observes the quantizer's outputs $Y_1,Y_2,\ldots,Y_n$ and
guesses which message was transmitted.

We impose an average-power constraint on the transmitted sequence: 
for every
realization of the message $M$, the sequence $x_1,x_2,\ldots,x_n$ must satisfy
\begin{equation}
\label{eq:power}
\frac{1}{n} \sum_{k=1}^n x^2_k \leq \const{P}
\end{equation}
for some positive constant $\const{P}$, which we call the \emph{maximal-allowed
  average-power.}

For a fixed quantizer $\set{D}$ and maximal-allowed average-power $\const{P}$,
the capacity $C(\const{P},\set{D})$ is \cite{gallager68,coverthomas91}
\begin{equation}
  \label{eq:def_C_P_D}
  C(\const{P},\set{D}) = \sup_{\E{X^{2}} \leq \const{P}} I(X;Y)
\end{equation}
where the supremum is over all distributions of $X$ under which the
second moment of $X$ does not exceed $\const{P}$. Here and
throughout the paper we omit the time indices where they are
immaterial.

We say that a \emph{rate} $R$ (in nats per channel use) is
\emph{achievable} using power $\const{P}$ and one-bit quantization if for every $\eps>0$
there exists an encoder satisfying \eqref{eq:power} and
\begin{equation}
\frac{\log\const{M}}{n} > R-\eps
\end{equation}
as well as a one-bit quantizer and a decoder such that the probability of
error $\Prob(\hat{M}\neq M)$ tends to zero as $n$ tends to
infinity. Here $\log(\cdot)$ denotes the natural logarithm
function. The \emph{capacity} $C(\const{P})$ is the supremum of all
achievable rates and is given by
\begin{align}
C(\const{P}) &= \sup_{\set{D}} C(\const{P},
\set{D}) \label{eq:am_capacity} \\
& = \sup_{\set{D},\E{X^2}\leq\const{P}} I(X;Y) \label{eq:capacity} 
\end{align}
where the first supremum is over all quantization regions $\set{D}$, and
the second supremum is over all quantization regions $\set{D}$ and over
all distributions of $X$ satisfying $\E{X^2}\leq \const{P}$.

Following \cite{verdu90}, we define the capacity per unit-energy of the
quantizer $\set{D}$ as follows: We say that a \emph{rate per unit-energy}
$\dot{R}(0,\set{D})$ (in nats per energy) is \emph{achievable} with
the quantizer~$\set{D}$ if for every
$\eps>0$ there exists an encoder satisfying
\begin{equation}
\label{eq:energy}
\sum_{k=1}^n x_k^2 \leq \const{E}, \quad \textnormal{for every realization of $M$}
\end{equation}
and
\begin{equation}
\frac{\log\const{M}}{\const{E}} > \dot{R}(0, \set{D}) - \eps
\end{equation}
together with a decoder such that the probability of error
$\Prob(\hat{M}\neq M)$ tends to zero as $\const{E}$ tends to
infinity. The \emph{capacity per unit-energy} $\dot{C}(0,\set{D})$ is
the supremum of all achievable rates per unit-energy with the quantizer $\set{D}$ and is given by \cite[Th.~2]{verdu90}
\begin{align}
\dot{C}(0,\set{D}) & = \sup_{\const{P}>0} \frac{C(\const{P},
  \set{D})}{\const{P}} \label{eq:amosCUElowSNR}\\
& = \lim_{\const{P} \downarrow 0} \frac{C(\const{P},
  \set{D})}{\const{P}} \label{eq:amosCUElowSNR2}
\end{align}
where the second equation follows because, for every $\set{D}$, the capacity $C(\const{P},\set{D})$  is a concave function of $\const{P}$.

The definition of capacity per unit-energy using a one-bit quantizer
is analogous: We say that a \emph{rate per unit-energy} $\dot{R}(0)$
(in nats per energy) is \emph{achievable} using a one-bit quantizer if
for every $\eps>0$ there exists an encoder satisfying \eqref{eq:energy} and
\begin{equation}
\frac{\log\const{M}}{\const{E}} > \dot{R}(0) - \eps
\end{equation}
as well as a one-bit quantizer and a decoder such that the probability of
error $\Prob(\hat{M}\neq M)$ tends to zero as $\const{E}$ tends to
infinity. The \emph{capacity per unit-energy} $\dot{C}(0)$ is the supremum of all achievable rates per
unit-energy.

Extending the proof of Theorem~2 in \cite{verdu90} to account for the
additional maximization over all possible quantizers, we obtain 
\begin{equation}
\label{eq:CUElowSNR}
\dot{C}(0) = \sup_{\const{P}>0} \frac{C(\const{P})}{\const{P}}
\end{equation}
which, by \eqref{eq:am_capacity}, can be expressed as 
\begin{equation}
\dot{C}(0) = \sup_{\const{P}>0} \sup_{\set{D}} \frac{C(\const{P},\set{D})}{\const{P}}.
\end{equation}
Exchanging the order of the suprema and applying \eqref{eq:amosCUElowSNR} yields
\begin{IEEEeqnarray}{lCl}
  \dot{C}(0)  & = & \sup_{\set{D}} \dot{C}(0,\set{D}) \label{eq:amos_French}\\
  & = & \sup_{\xi\neq 0, \set{D}} \frac{D\bigl(P_{Y | X=\xi} \bigm\| P_{Y | X=0}\bigr)}{\xi^2} \label{eq:CUEKL}
\end{IEEEeqnarray}
where the last step follows from \cite[Th.~3]{verdu02}. Here $D(\cdot\|\cdot)$ denotes relative entropy
\begin{equation}
D(P \| Q) \triangleq \left\{\begin{array}{ll}\displaystyle \int \log\left(\frac{\d P}{\d Q}\right)\d P, \quad & \textnormal{if }P\ll Q\\ \infty, \quad &\textnormal{otherwise}\end{array}\right.
\end{equation}
(where $P\ll Q$ indicates that $P$ is absolutely continuous with
respect to $Q$), and $P_{Y|X=x}$ denotes the output distribution
corresponding to the input $x$. In our case, since the output of the quantizer
is binary,
\begin{IEEEeqnarray}{lCl}
  \IEEEeqnarraymulticol{3}{l}{D\bigl(P_{Y | X=\xi} \bigm\| P_{Y | X=0}\bigr)}\nonumber\\
  \quad & = & \Prob\bigl(\tilde{Y}\in \set{D}\bigm| X = \xi\bigr) \log \frac{\Prob\bigl(\tilde{Y}
    \in \set{D}\bigm| X = \xi\bigr)}{\Prob\bigl(\tilde{Y}
    \in \set{D}\bigm| X = 0\bigr)}\nonumber\\
    & & {} + \Prob\bigl(\tilde{Y}
    \notin \set{D}\bigm| X = \xi\bigr) \log \frac{\Prob\bigl(\tilde{Y}
    \notin \set{D}\bigm| X = \xi\bigr)}{\Prob\bigl(\tilde{Y}
    \notin \set{D}\bigm| X = 0\bigr)}. \IEEEeqnarraynumspace
\end{IEEEeqnarray}
It follows from \eqref{eq:am_capacity} and \eqref{eq:amosCUElowSNR2} that
\begin{IEEEeqnarray}{lCl}
\lim_{\const{P}\downarrow 0} \frac{C(\const{P})}{\const{P}} & = &
\lim_{\const{P}\downarrow 0} \sup_{\set{D}}
\frac{C(\const{P},\set{D})}{\const{P}} \nonumber \\
& \geq & \sup_{\set{D}} \dot{C}(0,\set{D})
\end{IEEEeqnarray}
which, together with \eqref{eq:CUElowSNR} and \eqref{eq:amos_French}, yields
\begin{equation}
\dot{C}(0) = \lim_{\const{P}\downarrow 0} \frac{C(\const{P})}{\const{P}}.
\end{equation}
Thus, the capacity per unit-energy is equal to the slope at zero of
the capacity-vs-power curve.

By the Data Processing Inequality \cite[Th.~2.8.1]{coverthomas91}, $C(\const{P},\set{D})$ is upper-bounded by the capacity of the unquantized channel \cite{shannon48}
\begin{equation}
\label{eq:capacityGaussian}
C(\const{P}, \set{D}) \leq \frac{1}{2} \log\biggl(1+\frac{\const{P}}{\sigma^2}\biggr).
\end{equation}
Consequently, by \eqref{eq:amosCUElowSNR2} and \eqref{eq:amos_French},
\begin{equation}
\label{eq:CUCUB}
  \dot{C}(0,\set{D}) \leq \frac{1}{2\sigma^2} \quad \textnormal{and} \quad \dot{C}(0) \leq \frac{1}{2\sigma^2}.
\end{equation}

A ubiquitous quantizer is the \emph{symmetric threshold quantizer}, for which
$\set{D}=\{\tilde{y}\in\Reals\colon \tilde{y}\geq 0\}$. For this
quantizer the capacity $C_{\text{sym}}(\const{P})$ is given by
\cite[Th.~2]{singhdabeermadhow09_2}, \cite[Eq.~(3.4.18)]{viterbiomura79}
\begin{equation}
\label{eq:Csym}
C_{\text{sym}}(\const{P}) = \log 2 - H_b\Biggl(Q\Biggl(\sqrt{\frac{\const{P}}{\sigma^2}}\Biggr)\Biggr)
\end{equation}
where $H_b(\cdot)$ denotes the binary entropy function
\begin{equation}
\label{eq:amos_def_Hb}
H_b(p) \triangleq  -p\log p - (1-p)\log(1-p), \quad 0\leq p \leq 1
\end{equation}
(where we define $0\log 0\triangleq 0$) and $Q(\cdot)$ denotes the $Q$-function
\begin{equation}
Q(x) \triangleq \frac{1}{\sqrt{2\pi}}\int_x^{\infty} e^{-\frac{t^2}{2}}\d t, \quad x\in\Reals.
\end{equation}
The capacity $C_{\text{sym}}(\const{P})$ can be achieved by
transmitting $\sqrt{\const{P}}$ and $-\sqrt{\const{P}}$ equiprobably. 

From \eqref{eq:Csym}, the capacity per unit-energy
$\dot{C}_{\text{sym}}(0)$ for a symmetric threshold quantizer is \cite[Eq.~(3.4.20)]{viterbiomura79}
\begin{equation}
\label{eq:CUEsym}
\dot{C}_{\text{sym}}(0) = \lim_{\const{P}\downarrow 0}\frac{C_{\text{sym}}(\const{P})}{\const{P}} = \frac{1}{\pi\sigma^2}.
\end{equation}
This is a factor of $2/\pi$ smaller than the capacity per unit-energy $1/(2\sigma^2)$ of the Gaussian channel without output quantization. Thus, quantizing the channel output using a symmetric threshold quantizer causes a loss of roughly 2dB.

It is tempting to attribute this loss to the fact that the quantizer
forces the decoder to perform only hard-decision decoding. However, as we shall see, the loss of 2dB is not a consequence of the
hard-decision decoder but of the suboptimal quantizer. In fact, with
an asymmetric threshold quantizer the loss vanishes (Theorem~\ref{thm:No2dB}).

\section{Main Results}
\label{sec:results}
Our main results are presented in the following two
subsections. Section~\ref{sub:capacity} presents the results on
channel capacity. We show that the capacity-achieving
input distribution is discrete with at most three mass points and that
threshold quantizers achieve capacity
(Theorem~\ref{thm:amos}). Furthermore, we provide an expression for
the capacity when the average-power constraint \eqref{eq:power} is
replaced by a peak-power constraint (Proposition~\ref{note:PP}).

Section~\ref{sub:CUC} presents the results on capacity per
unit-energy. We show that asymmetric threshold quantizers
and asymmetric signal constellations can achieve the capacity per unit-energy of
the Gaussian channel (Theorem~\ref{thm:No2dB}), thus
demonstrating that quantizing the output of the Gaussian channel with
a one-bit quantizer does not cause an asymptotic power loss. We
further demonstrate that, in order to achieve this capacity per unit-energy, \emph{flash-signaling} input distributions \cite[Def.~2]{verdu02}
are required (Theorem~\ref{thm:flash}). Finally, we show that if the
average-power constraint \eqref{eq:power} is replaced by a peak-power
constraint, then quantizing the output of the Gaussian channel with a
one-bit quantizer necessarily causes a 2dB power loss (Proposition~\ref{note:PPCUC}).

\subsection{Channel Capacity}
\label{sub:capacity}

\begin{theorem}[Optimal Input Distribution and Quantizer]
\label{thm:amos}
\mbox{}
\begin{enumerate}
\item For any given maximal-allowed average-power $\const{P}$ and any
  Borel set $\set{D}$, the supremum in~\eqref{eq:def_C_P_D}
  defining~$C(\const{P},\set{D})$ is achieved by some input
  distribution that is concentrated on at most three points.
\item For any given maximal-allowed average-power $\const{P}$ the
  supremum in~\eqref{eq:capacity} is achieved by some threshold
  quantizer
  \begin{equation*}
    \set{D}^{\star} = \{\tilde{y} \in \Reals: \tilde{y} \geq \Upsilon\}
  \end{equation*}
  (where $\Upsilon \geq 0$ depends on $\const{P}$ and $\sigma^2$) and by a
  zero-mean, variance-$\const{P}$, input distribution that is
  concentrated on at most three points.
\end{enumerate}
\end{theorem}
\begin{proof}
See Section~\ref{sec:dobi_thmamos}.
\end{proof}
The result that the capacity-achieving input distribution is concentrated on at most three mass points is consistent with Theorem~1 in \cite{singhdabeermadhow09_2}, which shows that if the quantization regions of a $K$-bit quantizer partition the real line into $2^{K}$ intervals, then the capacity-achieving input distribution is concentrated on at most $2^{K}+1$ points.

\begin{proposition}
\label{note:PP}
If the average-power constraint \eqref{eq:power} is replaced by the
peak-power constraint
\begin{equation}
\label{eq:Ppower}
X_k^2 \leq \const{P}, \quad k\in\Integers, \quad \textnormal{with
  probability one}
\end{equation}
then the capacity of the channel presented in Section~\ref{sec:channel} is given by
\begin{IEEEeqnarray}{lCl}
C_{\textnormal{PP}}(\const{P}) & = & \max_{\Upsilon\geq 0} \Biggl\{ \log\Bigl(1+e^{-\Theta(\const{P},\Upsilon)}\Bigr) \nonumber\\
\IEEEeqnarraymulticol{3}{r}{\quad {} + Q\Biggl(\frac{\sqrt{\const{P}}+\Upsilon}{\sigma}\Biggr) \Theta(\const{P},\Upsilon) - H_b\Biggl(Q\Biggl(\frac{\sqrt{\const{P}}+\Upsilon}{\sigma}\Biggr)\Biggr) \Biggr\} \IEEEeqnarraynumspace} \label{eq:propPP}
\end{IEEEeqnarray}
where
\begin{equation}
\Theta(\const{P},\Upsilon) \triangleq \frac{H_b\Bigl(Q\Bigl(\frac{\sqrt{\const{P}}-\Upsilon}{\sigma}\Bigr)\Bigr)-H_b\Bigl(Q\Bigl(\frac{\sqrt{\const{P}}+\Upsilon}{\sigma}\Bigr)\Bigr)}{1-Q\Bigl(\frac{\sqrt{\const{P}}-\Upsilon}{\sigma}\Bigr)-Q\Bigl(\frac{\sqrt{\const{P}}+\Upsilon}{\sigma}\Bigr)}.
\end{equation}
The capacity can be achieved by a binary input distribution with mass points at $\sqrt{\const{P}}$ and $-\sqrt{\const{P}}$ and by some threshold quantizer with threshold $\Upsilon\geq 0$.
\end{proposition}
\begin{proof}
See Section~\ref{sub:proof_note2}.
\end{proof}
Numerical evaluation of \eqref{eq:propPP} suggests that, for every maximal-allowed peak-power $\const{P}$, the maximum is attained for \mbox{$\Upsilon=0$}. In this case, $C_{\textnormal{PP}}(\const{P})$ would specialize to the capacity of the average-power-limited Gaussian channel with symmetric output quantization \eqref{eq:Csym}.

\subsection{Capacity Per Unit-Energy}
\label{sub:CUC}

\begin{theorem}[$\dot{C}(0)=1/(2\sigma^2)$]
\label{thm:No2dB}
The capacity per unit-energy of the channel presented in Section~\ref{sec:channel} is
\begin{equation}
\label{eq:thmCUC}
\dot{C}(0) = \frac{1}{2\sigma^2}.
\end{equation}
\end{theorem}
\begin{proof}
See Section~\ref{sub:proof_thm3}.
\end{proof}
Thus, if we allow for asymmetric threshold quantizers and asymmetric signal
constellations, then quantizing the output of the
average-power-limited Gaussian channel with an optimal one-bit quantizer does
not cause a loss with respect to the capacity per unit-energy.

Considering the symmetry of the probability density function (PDF) of the Gaussian noise, it is perhaps surprising that an asymmetric quantizer yields a larger rate per unit-energy than a symmetric one. However, the input distribution achieving \eqref{eq:thmCUC} is asymmetric (see below). Hence, the PDF of the unquantized channel output is asymmetric, so it seems plausible that the capacity per unit-energy is achieved by some asymmetric quantizer. In fact, even if the PDF of the unquantized channel output were symmetric, this would not necessarily imply that the optimal quantizer is symmetric: There are examples in the source-coding literature of symmetric PDFs for which the optimal one-bit quantizer with respect to the mean squared error is asymmetric, see, e.g., \cite[Ex.~5.2, p.~64--65]{grafluschgy00}.

Theorem~\ref{thm:No2dB} is proved by analyzing \eqref{eq:CUEKL} with a judicious choice of $\set{D}$ and $\xi$. In Section~\ref{sec:ppm} we provide an alternative proof by presenting a PPM scheme that achieves the capacity per unit-energy \eqref{eq:thmCUC}. For this scheme, the error probability can be analyzed directly using the Union Bound and an upper bound on the $Q$-function: there is no need to resort to conventional methods used to prove coding theorems such as the method of types, information-spectrum methods, or random coding exponents.

The capacity per unit-energy \eqref{eq:thmCUC} can be achieved by binary on-off keying, i.e., by binary inputs of probability mass function
\begin{equation}
\label{eq:on-off}
P(X=\xi) = 1- P(X=0) = \frac{\const{P}}{\xi^2}, \quad \xi^2\geq \const{P}
\end{equation}
where the nonzero mass point $\xi$ tends to infinity as $\const{P}$ tends to zero. The distribution of such inputs belongs to the class of \emph{flash-signaling} input distributions, which was defined by Verd\'u \cite[Def.~2]{verdu02} as follows.

\begin{definition}[Flash Signaling]
\label{def:flashsig}
A family of distributions of~$X$ parametrized by $\const{P}$ is said
to be \emph{flash signaling} if it satisfies $\E{X^2} \leq \const{P}$
and for every positive $\nu$
\begin{equation}
\lim_{\const{P}\downarrow 0} \frac{\E{X^2 \I{X^2>\nu}}}{\const{P}} = 1.
\end{equation}
Here $\I{\textnormal{statement}}$ denotes the indicator function: it
  is equal to one if the statement between the curly brackets is true
  and is equal to zero otherwise.
\end{definition}

Flash signaling is described in \cite{verdu02} as ``the mixture of a
probability distribution that asymptotically concentrates its mass at
$0$ and a probability distribution that migrates to infinity; the
weight of the latter vanishes sufficiently fast to satisfy the
vanishing power constraint." The next theorem shows that flash
signaling is necessary to achieve \eqref{eq:thmCUC}.
\begin{theorem}[Flash Signaling Is Required to Achieve $\dot{C}(0)$]
\label{thm:flash}
Every family of distributions of $X$ parametrized by $\const{P}$ that satisfies
$\E{X^2}\leq\const{P}$ and
\begin{equation}
\lim_{\const{P}\downarrow 0} \frac{I(X;Y)}{\const{P}} =
\frac{1}{2\sigma^2} 
\label{eq:1st_ord}
\end{equation}
must be flash signaling.
\end{theorem}
\begin{proof}
See Section~\ref{sub:proof_thm4}.
\end{proof}

It is easy to show that for flash-signaling input distributions, threshold quantizers with a bounded threshold give rise to zero rate per unit-energy. We thus have the following corollary.
\begin{corollary}[The Thresholds Must Be Unbounded]
\label{cor:unboundthres}
If \eqref{eq:1st_ord} holds for some family of threshold quantizers (parametrized by the average power), then the
thresholds must be unbounded in the average power.
\end{corollary}
\begin{proof}
See Section~\ref{sub:proof_cor5}.
\end{proof}

Intuitively, the power loss in quantizing the output of the Gaussian channel with a one-bit quantizer can be avoided by using flash-signaling input distributions and asymmetric threshold quantizers because for such input distributions and quantizers the probability that the quantizer causes an error vanishes as the SNR tends to zero. Indeed, by using binary on-off keying \eqref{eq:on-off} and threshold quantizers \eqref{eq:dobi_thresholdD}, and by cleverly choosing the rate at which $\xi$ and $\Upsilon$ grow as $\const{P}$ decreases, we can make the probabilities $\Prob(Y=1|X=0)$ and \mbox{$\Prob(Y=0|X=\xi)$} vanish as $\const{P}$ tends to zero. This suggests that the loss caused by the quantizer disappears with decreasing $\const{P}$. Note, however, that the same argument would also apply to the averaged-power-limited, noncoherent, Rayleigh-fading channel (see Section~\ref{sec:noncoherent}), but for this channel quantizing the output with a one-bit quantizer does cause a loss with respect to the capacity per unit-energy (Theorem~\ref{prop:noncoherent}).

As mentioned in Section~\ref{sec:channel}, the capacity per
unit-energy is equal to the slope at zero of the capacity-vs-power curve. Thus, Theorem~\ref{thm:No2dB} demonstrates that the first
derivative of $C(\const{P})$ at $\const{P}=0$ is equal to
$1/(2\sigma^2)$. Theorem~\ref{thm:flash} implies that the second
derivative of $C(\const{P})$ at $\const{P}=0$ is $-\infty$. 
\begin{corollary}[$\ddot{C}(0) = -\infty$]
  \begin{equation}
    \ddot{C}(0) = 2\lim_{\const{P}\downarrow 0}
    \frac{C(\const{P}) - \const{P}\, \dot{C}(0)}{\const{P}^2} =
    -\infty. \label{eq:amos_ddotC}
  \end{equation}
\end{corollary}
\begin{proof}
  By the Data Processing Inequality, for every family of
  distributions of $X$ parametrized by $\const{P}$
  \begin{IEEEeqnarray}{lCl}
    \lim_{\const{P}\downarrow 0} \frac{I(X;Y)-\frac{\const{P}}{2\sigma^2}}{\const{P}^2}
    & \leq & \lim_{\const{P}\downarrow 0} \frac{I(X;\tilde{Y})-\frac{\const{P}}{2\sigma^2}}{\const{P}^2}. \label{eq:dobi_Cddot}
  \end{IEEEeqnarray}
  To achieve $\dot{C}(0)$ it is necessary to use flash signaling (Theorem~\ref{thm:flash}). And for all flash-signaling input distributions the right-hand side (RHS) of \eqref{eq:dobi_Cddot} is $-\infty$ (\cite[Th.~16]{verdu02}). Consequently, so is its left-hand side (LHS). 
  \end{proof}

Note that, for the Gaussian channel, the first and second derivative of the capacity are \cite{shannon48}
\begin{equation}
\label{eq:dotC_ddotC_G}
\dot{C}(0)=\frac{1}{2\sigma^2} \quad \textnormal{and} \quad \ddot{C}_{\textnormal{G}}(0) = -\frac{1}{2\sigma^4}
\end{equation}
(where ``G'' stands for ``Gaussian"). Thus, while quantizing the
output of the Gaussian channel with a one-bit quantizer does not cause
a loss with respect to the first derivative of the capacity-vs-power
curve, it causes a substantial loss in terms of the second
derivative. The implications on the spectral efficiency are discussed
in Section~\ref{sec:EbNo}.

\begin{proposition}
\label{note:PPCUC}
If the average-power constraint \eqref{eq:power} is replaced by the peak-power constraint
\begin{equation}
X_k^2 \leq \const{P}, \quad k\in\Integers, \quad \textnormal{with probability one}
\end{equation}
then the slope at zero of the capacity-vs-power curve is
\begin{equation}
\lim_{\const{P}\downarrow 0}\frac{C_{\textnormal{PP}}(\const{P})}{\const{P}} = \frac{1}{\pi\sigma^2}.
\end{equation}
\end{proposition}
\begin{proof}
See Section~\ref{sub:proof_note3}.
\end{proof}
As was shown by Shannon \cite{shannon48}, the capacity of the peak-power-limited unquantized Gaussian channel satisfies
\begin{equation}
\lim_{\const{P}\downarrow 0} \frac{C_{\textnormal{G,PP}}(\const{P})}{\const{P}} = \frac{1}{2\sigma^2}.
\end{equation}
Thus, in contrast to the average-power-limited case, quantizing the output of the peak-power-limited Gaussian channel with a one-bit quantizer does cause a 2dB power loss.

\section{Pulse-Position Modulation}
\label{sec:ppm}
We next demonstrate that the capacity per unit-energy
\eqref{eq:thmCUC} can be achieved using a PPM scheme---no
random-coding arguments are needed. For such a scheme the encoder produces the $\const{M}$
channel inputs $x_1(m),x_2(m),\ldots,x_{\const{M}}(m)$ for each message
$m$ in $\{1,2,\ldots,\const{M}\}$, where
\begin{equation}
  x_k(m) = \begin{cases}
    \xi & \text{if $k=m$}, \\
    0 & \text{if $k \neq m$}, 
  \end{cases}
  \quad k = 1, 2,\ldots, \const{M}
\end{equation}
and where $\xi^2=\const{E}$. For a fixed rate per
unit-energy \[\dot{R}(0) = \frac{\log\const{M}}{\const{E}}\] we have
\begin{equation}
\label{eq:xi}
\xi^2 = \const{E} = \frac{\log\const{M}}{\dot{R}(0)}.
\end{equation}
Note that, while the \emph{rate per unit-energy} is fixed, the
\emph{rate} of this scheme is $\frac{\log\const{M}}{\const{M}}$ and
tends to zero as $\const{M}$ tends to infinity. 

We employ a threshold quantizer \eqref{eq:dobi_thresholdD} with the threshold $\Upsilon$ chosen so that for an arbitrary $0<\eps< 1$ the probability that the quantizer produces~$0$ given that $X=\xi$ is equal to $\eps$. Thus,
\begin{equation}
\label{eq:doby_41}
\Upsilon = \xi -\sigma Q^{-1}(\eps)
\end{equation}
which yields
\begin{subequations}
\label{eq:amos_rain}
\begin{IEEEeqnarray}{lCl}
P\bigl(Y_k=0 \bigm| X_k=\xi\bigr) & = & \eps \label{eq:amos_rainA}\\
P\bigl(Y_k=1 \bigm| X_k=0\bigr) & = & Q\left(\frac{\xi-\sigma
    Q^{-1}(\eps)}{\sigma}\right). \label{eq:amos_rainB}
\end{IEEEeqnarray}
\end{subequations}
In \eqref{eq:doby_41}, $Q^{-1}(\cdot)$ denotes the inverse $Q$-function. 

The decoder guesses ``$\hat{M}=m$'' provided that $Y_m=1$ and that
$Y_k=0$ for all $k\neq m$. If $Y_k=1$ for more than one $k$, or
if $Y_k=0$ for all $k=1,2,\ldots,\const{M}$, then the decoder
declares an error.

Suppose that message $M=m$ was transmitted. Then the probability of an error is upper-bounded by
\begin{IEEEeqnarray}{lCl}
\IEEEeqnarraymulticol{3}{l}{\Prob\bigl(\hat{M}\neq M \bigm| M=m\bigr)} \nonumber\\
\quad & = & \Prob\left(\left.\bigcup_{k\neq m}(Y_k=1)\cup (Y_m=0)\right| M=m\right)\nonumber\\
& \leq & \sum_{k\neq m} P\bigl(Y_k=1\bigm| X_k=0\bigr) + P\bigl(Y_m=0\bigm|X_m=\xi\bigr) \nonumber\\
& = & \sum_{k\neq m} P\bigl(Y_k=1\bigm|X_k=0\bigr) + \eps \nonumber\\
& = & (\const{M}-1)\,  P\bigl(Y_1=1\bigm|X_1=0\bigr) + \eps \label{eq:union}
\end{IEEEeqnarray}
where the second step follows from the Union Bound; the third step follows from \eqref{eq:amos_rainA}; and the fourth step follows because the channel is
memoryless which implies that \mbox{$\Prob(Y_k=1|X_k=0)$}
does not depend on $k$. Since the RHS of \eqref{eq:union} does not
depend on $m$, it follows that also the probability of error
\begin{equation*}
\Prob(\hat{M}\neq M) = \frac{1}{\const{M}}\sum_{m=1}^{\const{M}} \Prob\bigl(\hat{M}\neq M\bigm| M=m \bigr)
\end{equation*}
is upper-bounded by \eqref{eq:union}.

The first term on the RHS of \eqref{eq:union} can be evaluated using
\eqref{eq:amos_rainB} and \eqref{eq:xi}:
\begin{IEEEeqnarray}{lCl}
\IEEEeqnarraymulticol{3}{l}{(\const{M}-1)\, P\bigl(Y_1=1\bigm|X_1=0\bigr)} \nonumber\\
\quad & = & (\const{M}-1) \,Q\left(\frac{\xi-\sigma Q^{-1}(\eps)}{\sigma}\right) \nonumber\\
& = & (\const{M}-1) \,Q\left(\frac{\sqrt{\log\const{M}}-\sigma Q^{-1}(\eps)\sqrt{\dot{R}(0)}}{\sigma\sqrt{\dot{R}(0)}}\right). \IEEEeqnarraynumspace \label{eq:P10}
\end{IEEEeqnarray}
We continue by showing that if
\begin{equation*}
\dot{R}(0) < \frac{1}{2\sigma^2}
\end{equation*}
then, for every fixed $0<\eps<1$, the RHS of \eqref{eq:P10} tends to zero as $\const{M}$ tends to infinity.  Indeed, 
\begin{IEEEeqnarray}{lCl}
\IEEEeqnarraymulticol{3}{l}{\lim_{\const{M}\to\infty}(\const{M}-1) \,Q\left(\frac{\sqrt{\log\const{M}}-\sigma Q^{-1}(\eps)\sqrt{\dot{R}(0)}}{\sigma\sqrt{\dot{R}(0)}}\right)}\nonumber\\
\,\,\,\,\, & \leq & \lim_{\alpha\to\infty} \exp\left(\sigma^2\dot{R}(0)\left(\alpha+Q^{-1}(\eps)\right)^2\right) Q(\alpha) \nonumber\\
& \leq & \lim_{\alpha\to\infty}\frac{1}{\sqrt{2\pi}\alpha}\exp\left(\sigma^2\dot{R}(0)\left(\alpha+Q^{-1}(\eps)\right)^2-\frac{1}{2}\alpha^2\right)\IEEEeqnarraynumspace\label{eq:Mlarge}
\end{IEEEeqnarray}
where the first step follows by upper-bounding $\const{M}-1 < \const{M}$ and by substituting
\begin{equation*}
\alpha=\frac{\sqrt{\log\const{M}}-\sigma Q^{-1}(\eps)\sqrt{\dot{R}(0)}}{\sigma \sqrt{\dot{R}(0)}};
\end{equation*}
and the second step follows from the inequality \cite[Prop.~19.4.2]{lapidoth09}
\begin{equation}
\label{eq:QUB}
Q(\alpha) < \frac{1}{\sqrt{2\pi}\alpha} e^{-\alpha^2/2}, \quad \alpha>0.
\end{equation}
The RHS of \eqref{eq:Mlarge} is zero for $\dot{R}(0)<\frac{1}{2\sigma^2}$.

Combining \eqref{eq:Mlarge} with \eqref{eq:union}, we obtain that if
$\dot{R}(0)<\frac{1}{2\sigma^2}$, then the probability of error tends
to $\eps$ as $\const{E}$---and hence, by \eqref{eq:xi}, also
$\const{M}$---tends to
infinity. Since $\eps$ can be chosen arbitrarily small, the
probability of error can be made arbitrarily small,
thus proving that the capacity per unit-energy \eqref{eq:thmCUC} is
achievable with the above PPM scheme.

The fact that PPM achieves the capacity per unit-energy of the Gaussian channel with a threshold quantizer follows also from the analysis of the probability of error for block orthogonal signals shown in \cite[p.~342--346]{wozencraftjacobs65}. The threshold $a\geq 0$ introduced to bound the RHS of (5.97d) in \cite{wozencraftjacobs65} can be identified as the threshold~$\Upsilon$ of the quantizer.

\section{Spectral Efficiency}
\label{sec:EbNo}
The discrete-time channel presented in Section~\ref{sec:channel} is closely related to the
continuous-time AWGN channel with one-bit output
quantization. Indeed, suppose that the input to the latter
channel is bandlimited to $\WW$ Hz and that its average-power is limited by
$\const{P}$, and suppose that the Gaussian noise is of double-sided
power spectral density $\Nzero/2$. Then, the discrete-time channel
\eqref{eq:channel} with noise-variance 
\begin{equation}
\label{eq:amos_sigma}
\sigma^2=\WW\Nzero
\end{equation}
results from sampling the AWGN channel's output at the Nyquist rate
$2\WW$. The capacity (in bits per second) of the AWGN channel with
Nyquist sampling and one-bit output quantization is given by
\begin{equation}
\label{eq:AWGN}
C_{\textnormal{AWGN}}^{(2\WW)}(\const{P}) = \frac{2\WW}{\log 2} C(\const{P})
\end{equation}
where $C(\const{P})$ is the capacity \eqref{eq:capacity} of the discrete-time channel in nats per channel use. Note, however, that when the
channel output is quantized, sampling at the Nyquist rate need not be
optimal with respect to capacity: see, e.g.,
\cite{kochlapidoth10_2}--\cite{shamai94} for scenarios where sampling
the quantizer's output above the Nyquist rate provides capacity
gains. Consequently, $C_{\textnormal{AWGN}}^{(2\WW)}(\const{P})$ is,
in general, a lower bound on the capacity of the
AWGN channel with one-bit output quantization.

The energy per
information-bit when communicating with power $\const{P}$ at rate
$C_{\textnormal{AWGN}}^{(2\WW)}(\const{P})$ is defined as
\begin{equation}
\frac{\const{E}_{\textnormal{b}}}{\Nzero} \triangleq \frac{\const{P}}{C_{\textnormal{AWGN}}^{(2\WW)}(\const{P})} \frac{1}{\Nzero} \label{eq:EB_doby}
\end{equation}
which, by \eqref{eq:amos_sigma} and \eqref{eq:AWGN}, is equal to
\begin{equation}
 \label{eq:def_Ebno}
 \frac{\const{E}_{\textnormal{b}}}{\Nzero} = \frac{\log 2}{2\sigma^2} \frac{\const{P}}{C(\const{P})}.
\end{equation}
The spectral efficiency $\bar{C}(\cdot)$ (in bits per second per Hz)
is defined as
\begin{equation}
\label{eq:cbar1}
\bar{C}\biggl(\frac{\const{E}_{\textnormal{b}}}{\Nzero}\biggr) \triangleq \frac{C_{\textnormal{AWGN}}^{(2\WW)}(\const{P})}{\WW}
\end{equation}
which, by \eqref{eq:AWGN}, is
\begin{equation}
\label{eq:cbar2}
\bar{C}\biggl(\frac{\const{E}_{\textnormal{b}}}{\Nzero}\biggr) = \frac{2}{\log 2} C(\const{P}).
\end{equation}
In \eqref{eq:cbar1} and \eqref{eq:cbar2}, $\const{P}$ is the solution to \eqref{eq:EB_doby}, namely,
\begin{equation}
\frac{\const{E}_{\textnormal{b}}}{\Nzero} =\frac{\const{P}}{C_{\textnormal{AWGN}}^{(2\WW)}(\const{P})}\frac{1}{\Nzero}.
\end{equation}
See \cite{verdu02} for a more thorough discussion of spectral
efficiency. (Note that, in contrast to \eqref{eq:channel}, the channel
considered in \cite{verdu02} is complex-valued. Therefore, the
expressions for $\const{E}_{\textnormal{b}}/\Nzero$ and
$\bar{C}\bigl(\const{E}_{\textnormal{b}}/\Nzero\bigr)$ differ by a
factor of two.)

The \emph{minimum} $\const{E}_{\textnormal{b}}/\Nzero$ 
required for reliable communication is determined by taking the
infimum over $\const{P}$ of the RHS of \eqref{eq:def_Ebno}. By \eqref{eq:CUElowSNR} this yields
\cite[Eq.~(35)]{verdu02}
\begin{equation}
\biggl(\frac{\const{E}_{\textnormal{b}}}{\Nzero}\biggr)_{\textnormal{min}} = \frac{\log 2}{2\sigma^2}\frac{1}{\dot{C}(0)}.
\end{equation}
Furthermore, the slope of $\const{E}_{\textnormal{b}}/\Nzero\mapsto\bar{C}\bigl(\const{E}_{\textnormal{b}}/\Nzero\bigr)$ at $(\const{E}_{\textnormal{b}}/\Nzero)_{\textnormal{min}}$ in bits per second per Hz per 3dB is given by \cite[Th.~9]{verdu02}\footnote{Again, the channel considered in \cite{verdu02} is complex-valued and the expressions for $\bigl(\const{E}_{\textnormal{b}}/\Nzero\bigr)_{\textnormal{min}}$ and $\const{S}_0$ therefore differ by a factor of two. Nevertheless, since the capacity of the complex-valued channel is twice the capacity of the real-valued channel, it follows that the numerical values of $\bigl(\const{E}_{\textnormal{b}}/\Nzero\bigr)_{\textnormal{min}}$ and $\const{S}_0$ are the same as in \cite{verdu02}.}
\begin{equation}
\const{S}_0 = \frac{4\bigl[\dot{C}(0)\bigr]^2}{-\ddot{C}(0)}.
\end{equation}
\begin{figure}[t]
\centering
\psfrag{C}[c][c]{$\bar{C}(\const{E}_{\textnormal{b}}/\Nzero)$ [bps/Hz]}
\psfrag{Eb}[ct][cc]{$\frac{\const{E}_{\textnormal{b}}}{\Nzero}$ [dB]}
 \psfrag{Gaussian channel}[ll][ll]{\footnotesize Gaussian channel}
\psfrag{Quantized output}[ll][ll]{\footnotesize Quantized output}
\psfrag{Symmetric quantizer}[ll][ll]{\footnotesize Symmetric quantizer}
\psfrag{Optimal quantizer}[ll][ll]{\footnotesize Optimal quantizer}
\begin{center}
\epsfig{file=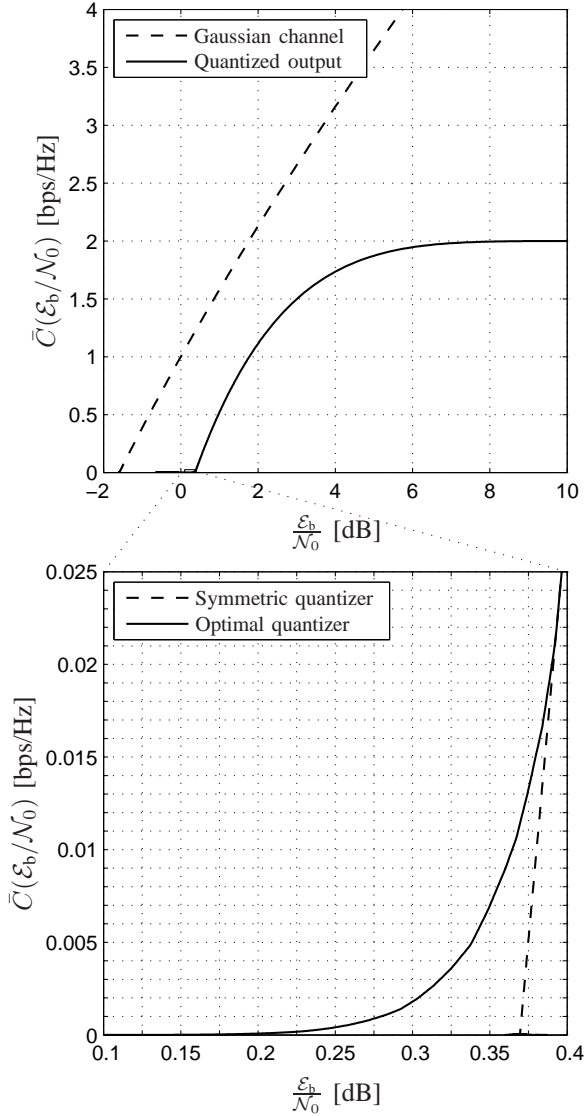,width=0.44\textwidth}
\end{center}
 \caption{Spectral efficiency versus energy per information-bit. The top subfigure shows the spectral efficiencies of the Gaussian channel with and without one-bit output quantization. The bottom subfigure shows the spectral efficiencies for the optimal one-bit quantizer and for the symmetric threshold quantizer.}
 \label{fig:spectraleff}
\end{figure}

By \eqref{eq:thmCUC} and \eqref{eq:amos_ddotC}, we have for the average-power-limited Gaussian channel with one-bit output quantization
\begin{equation}
\dot{C}(0) = \frac{1}{2\sigma^2} \quad \textnormal{and} \quad \ddot{C}(0) =-\infty
\end{equation}
which yields
\begin{subequations}
\begin{IEEEeqnarray}{rCl}
\biggl(\frac{\const{E}_{\textnormal{b}}}{\Nzero}\biggr)_{\textnormal{min}} & = & \log 2 = -1.59\textnormal{ dB} \\
\const{S}_0 & = & 0 \textnormal{ }\frac{\textnormal{bps}/\textnormal{Hz}}{3\textnormal{dB}}.
\end{IEEEeqnarray}
\end{subequations}
In comparison, for the unquantized Gaussian channel \eqref{eq:dotC_ddotC_G}
\begin{equation}
\dot{C}_{\textnormal{G}}(0) = \frac{1}{2\sigma^2} \quad \textnormal{and} \quad \ddot{C}_{\textnormal{G}}(0) = -\frac{1}{2\sigma^4}
\end{equation}
and for the Gaussian channel with \emph{symmetric} one-bit output quantization \eqref{eq:Csym}
\begin{equation}
\dot{C}_{\textnormal{sym}}(0) = \frac{1}{\pi\sigma^2} \quad \textnormal{and} \quad \ddot{C}_{\textnormal{sym}}(0) = \frac{2}{3\pi\sigma^4} \biggl(\frac{1}{\pi}-1\biggr).
\end{equation}
This yields
\begin{subequations}
\begin{IEEEeqnarray}{rCl}
\biggl(\frac{\const{E}_{\textnormal{b}}}{\Nzero}\biggr)_{\textnormal{min,G}} & = & \log 2 = -1.59 \textnormal{ dB} \label{eq:Ebmin_G}\\
\const{S}_{0,\textnormal{G}} & = & 2 \textnormal{ }\frac{\textnormal{bps}/\textnormal{Hz}}{3\textnormal{dB}} \label{eq:S0_G}
\end{IEEEeqnarray}
\end{subequations}
and
\begin{subequations}
\begin{IEEEeqnarray}{rCl}
\biggl(\frac{\const{E}_{\textnormal{b}}}{\Nzero}\biggr)_{\textnormal{min,sym}} & = & \frac{\pi}{2} \log 2 = 0.37 \textnormal{ dB} \label{eq:Ebmin_sym}\\
 \const{S}_{0,\textnormal{sym}} & = & \frac{6}{\pi-1} = 2.8\textnormal{ }\frac{\textnormal{bps}/\textnormal{Hz}}{3\textnormal{dB}}.
\end{IEEEeqnarray}
\end{subequations}
Comparing \eqref{eq:Ebmin_sym} with \eqref{eq:Ebmin_G}, we see once more that quantizing the output of the Gaussian channel with a symmetric threshold quantizer causes a power loss of roughly 2dB. We further see that with an asymmetric threshold quantizer we can recover the loss in terms of $\bigl(\const{E}_{\textnormal{b}}/\Nzero\bigr)_{\text{min}}$, but there is still a substantial loss in terms of spectral efficiency. Indeed, for the Gaussian channel with one-bit output quantization, the wideband slope $\const{S}_0$ is zero, whereas for the unquantized Gaussian channel it is $2$ bits per second per Hz per $3$dB.

The above spectral efficiencies are shown in Figure~\ref{fig:spectraleff}. The top subfigure shows the spectral efficiencies of the Gaussian channel with and without one-bit output quantization. The bottom subfigure compares the spectral efficiency $\bar{C}(\cdot)$ for the optimal one-bit quantizer with the spectral efficiency $\bar{C}_{\textnormal{sym}}(\cdot)$ for the symmetric threshold quantizer. We observe that, even though the minimum energy per information-bit is the same with and without one-bit output quantization,\footnote{For numerical reasons, the spectral efficiency of the Gaussian channel with one-bit output quantization can only be shown for $\const{E}_{\textnormal{b}}/\Nzero$ above $-0.5$dB.} the corresponding spectral efficiencies differ substantially for all $\const{E}_{\textnormal{b}}/\Nzero$. We further observe that for spectral efficiencies above $0.02$~bits per second per Hz a symmetric threshold quantizer is nearly optimal.

We conclude that, for communication systems that operate at very low spectral efficiencies, asymmetric quantizers are beneficial, although for most practical scenarios the potential power gain is significantly smaller than $2$dB. For example, at a spectral efficiency of $0.001$ bits per second per Hz, allowing for asymmetric quantizers with corresponding asymmetric signal constellations provides a power gain of roughly $0.1$dB.

\section{One-Bit Quantizers for Fading Channels}
\label{sec:noncoherent}
For the average-power-limited (real-valued) Gaussian channel, we have
demonstrated that by allowing for asymmetric threshold quantizers with
corresponding asymmetric signal constellations, one can achieve the
capacity per unit-energy of the unquantized channel. The same holds
for the average-power-limited \emph{complex-valued} Gaussian channel
\cite{zhangwillemshuang11}: using binary on-off keying
\eqref{eq:on-off} and a \emph{radial quantizer} (which produces $1$ if the magnitude
of the channel output is above some threshold and produces $0$
otherwise), one can achieve the capacity per unit-energy of the
unquantized channel by judiciously choosing the threshold and the
nonzero mass point as functions of the SNR.

In this section we briefly discuss the effect of one-bit quantization
on the capacity per unit-energy of the discrete-time,
average-power-limited, \emph{Rayleigh-fading channel}. This channel's
unquantized output $\tilde{Y}_{k}$ is given by
\begin{equation}
\label{eq:complex_channel}
\tilde{Y}_k = H_k X_k + Z_k, \quad k\in\Integers
\end{equation}
where $\{H_k,\,k\in\Integers\}$ and $\{Z_k,\,k\in\Integers\}$ are
independent sequences of i.i.d., zero-mean, circularly-symmetric,
complex Gaussian random variables, the former with unit-variance and
the latter with variance $\sigma^2$. We say that the channel is
\emph{coherent} if the receiver is cognizant of the realization of
$\{H_k,\,k\in\Integers\}$ and that it is \emph{noncoherent} if the receiver
is only cognizant of the statistics of $\{H_k,\,k\in\Integers\}$.
The unquantized output $\tilde{Y}_{k}$ is quantized using a one-bit
quantizer that is specified by a Borel subset $\set{D}$ of the complex
field $\Complex$: it produces $1$ if $\tilde{Y}_k$ is in $\set{D}$,
and it produces~$0$ if it is not. 

The capacities $C(\const{P}, \set{D})$ and $C(\const{P})$ are defined as in Section~\ref{sec:channel} but with  the average-power constraint \eqref{eq:power} replaced by
\begin{equation}
\frac{1}{n}\sum_{k=1}^n |x_k|^2\leq \const{P}.
\end{equation}
Likewise, the capacities per unit-energy $\dot{C}(0,\set{D})$ and $\dot{C}(0)$
are defined as in Section~\ref{sec:channel} but with the energy constraint \eqref{eq:energy} replaced by
\begin{equation}
\label{eq:fading_energy}
\sum_{k=1}^n |x_k|^2 \leq \const{E}.
\end{equation}

\subsection{Coherent Fading Channels}
Using the same arguments as in Section~\ref{sec:channel}, it can be
shown that, for a fixed quantizer $\set{D}$, we have for the coherent channel \cite[Th.~3]{verdu90}, \cite{verdu02}
\begin{equation}
  \label{eq:CUC_coherent_D}
  \dot{C}(0,\set{D}) = \sup_{\xi\neq 0} 
  \frac{D\bigl(P_{Y|H,X=\xi}\bigm\| P_{Y|H,X=0}\bigm| P_H\bigr)}{|\xi|^2}
\end{equation}
where $D(\cdot\|\cdot|\cdot)$ denotes conditional relative entropy
\begin{IEEEeqnarray}{lCl}
  \IEEEeqnarraymulticol{3}{l}{D\bigl(P_{Y|H,X=\xi}\bigm\| P_{Y|H,X=0}\bigm| P_H\bigr)} \nonumber\\
  \quad & = &
  \int D\bigl(P_{Y|H=h,X=\xi}\bigm\| P_{Y|H=h,X=0}\bigr) \d P_H(h); \IEEEeqnarraynumspace
\end{IEEEeqnarray}
$P_H$ denotes the distribution of the fading $H$; and $P_{Y|H=h,X=x}$
denotes the distribution of $Y$ conditioned on $(H,X)=(h,x)$.\footnote{This can be shown along the lines of the proof of Theorem~3 in \cite{verdu90} but with
the mutual information $I(X;Y)$ replaced by the conditional mutual
information $I(X;Y|H)$. That the RHS of \eqref{eq:CUC_coherent_D} is
an upper bound on $\dot{C}(0,\set{D})$ follows then immediately from
\cite[Eq.~(15)]{verdu90}. Showing that this holds with equality
requires swapping the order of taking the limit as $\const{P}$
tends to zero and of computing the expectation over the fading.}
It can be further shown that
\begin{equation}
\label{eq:CUC_coherent}
\dot{C}(0) = \sup_{\xi\neq 0, \set{D}} \frac{D\bigl(P_{Y|H,X=\xi}\bigm\| P_{Y|H,X=0}\bigm| P_H\bigr)}{|\xi|^2}.
\end{equation}

By the Data Processing Inequality, the capacity per unit-energy is
upper-bounded by that of the unquantized channel \cite{lapidothshamai02,verdu02}
\begin{equation}
\label{eq:CUC_unquant}
\dot{C}(0) \leq \frac{1}{\sigma^2}.
\end{equation}
We next show that, by choosing the one-bit quantizer as a function of $H$ and the SNR, this upper bound can be achieved.

\begin{theorem}[Coherent Case]
\label{prop:coherent}
The capacity per unit-energy of the coherent Rayleigh-fading channel 
is given by
\begin{equation}
\label{eq:prop7}
\dot{C}(0) = \frac{1}{\sigma^2}.
\end{equation}
It is achieved by a family of radial quantizers parametrized by
$\const{P}$ with thresholds that are
proportional to $|H|$.
\end{theorem}
\begin{proof}
See Section~\ref{sub:proof_prop7}.
\end{proof}
The assumption that the fading $H$ is Gaussian is not essential. In fact, Theorem~\ref{prop:coherent} holds for every fading distribution having unit variance.

\subsection{Noncoherent Fading Channels}

Using the same arguments as in Section~\ref{sec:channel}, it can be
shown that in the noncoherent case
\begin{equation}
\dot{C}(0,\set{D}) = \sup_{\xi\neq 0} \frac{D\bigl(P_{Y|X=\xi}\bigm\| P_{Y|X=0} \bigr)}{|\xi|^2}
\end{equation}
and
\begin{equation}
\label{eq:CUC_noncoherent}
\dot{C}(0) = \sup_{\xi\neq 0, \set{D}} \frac{D\bigl(P_{Y|X=\xi}\bigm\| P_{Y|X=0} \bigr)}{|\xi|^2}.
\end{equation}
Since the capacity per unit-energy of the unquantized Rayleigh-fading channel equals $1/\sigma^2$ irrespective of whether the channel is coherent or not \cite{lapidothshamai02,verdu02}, it follows from the Data Processing Inequality that \eqref{eq:CUC_unquant} holds also in the noncoherent case.

The capacity per unit-energy \eqref{eq:prop7} of the coherent channel with one-bit output quantization is achieved using binary on-off keying where the nonzero mass point tends to infinity as the SNR tends to zero. This result
might mislead one to think that \eqref{eq:prop7} also holds in the noncoherent case. Indeed, in the absence of a quantizer,
binary on-off keying with diverging nonzero mass point achieves the capacity per
unit-energy $1/\sigma^2$ irrespective of whether the receiver is cognizant of the fading realization or not \cite{verdu02,lapidothshamai02}. It
might therefore seem plausible that also in the noncoherent case
quantizing the channel output with a one-bit quantizer would cause no loss
in the capacity per unit-energy. But this is not the
case:


\begin{theorem}[Noncoherent Case]
\label{prop:noncoherent}
For the noncoherent Rayleigh-fading channel with one-bit output quantization
\begin{equation}
\dot{C}(0) < \frac{1}{\sigma^2}.
\end{equation}
\end{theorem}
\begin{proof}
See Section~\ref{sub:proof_prop8}.
\end{proof}

The case where the real and imaginary parts of the fading channel's output are quantized separately using a one-bit quantizer for each was studied, e.g., in  \cite{mezghaninossek07}--\nocite{mezghaninossek08}\nocite{mezghaninossek09}\nocite{kronefettweis10}\cite{kochlapidoth12_1}. However, in \cite{mezghaninossek07}--\nocite{mezghaninossek08}\nocite{mezghaninossek09}\cite{kronefettweis10} only symmetric threshold quantizers are considered.

\section{Proof of Theorem~\ref{thm:amos}}
\label{sec:dobi_thmamos}
We prove Theorem~\ref{thm:amos} in five steps:
\begin{enumerate}
\item We first show that for any given maximal-allowed average-power $\const{P}$ and any Borel set $\set{D}$, the supremum in~\eqref{eq:def_C_P_D} defining $C(\const{P},\set{D})$ is achieved by some input distribution that is concentrated on at most three points (Section~\ref{sub:dobi_3MP}).
\item We next show that for every three-mass-points input distribution, the supremum over all quantizers can be replaced with the supremum over all threshold quantizers and all quantizers whose quantization region consists of a finite interval (Section~\ref{sec:KreinMilman}).
\item We continue by showing that the supremum in \eqref{eq:capacity} defining $C(\const{P})$ is achieved (Section~\ref{sub:dobi_sup=max}).
\item We then show that threshold quantizers are optimal by demonstrating that quantization regions consisting of a finite interval are suboptimal (Section~\ref{sub:dobi_threshold}).
\item We finally show that the capacity-achieving input distribution must be centered and must satisfy the average-power constraint with equality (Section~\ref{sub:dobi_input}).
\end{enumerate}

\subsection{Input Distributions Consisting of Three Mass Points}
\label{sub:dobi_3MP}
Generalizing the proof of Theorem~1 in \cite{singhdabeermadhow09_2} to arbitrary quantizers, we prove that for every fixed quantizer $\set{D}$ and maximal-allowed
average-power $\const{P}$, the capacity $C(\const{P},\set{D})$ is
achieved by an input distribution consisting of three (or fewer) mass
points. To this end, we first argue that we can introduce an additional
peak-power constraint without reducing capacity, provided that we
allow the maximal-allowed peak-power to tend to infinity. Thus, we
show that $C(\const{P}, \set{D})$, which is defined in
\eqref{eq:def_C_P_D} without a peak-power constraint, can also be
expressed as
\begin{IEEEeqnarray}{lCl}
C(\const{P}, \set{D}) = \lim_{\const{A} \to \infty} 
\sup_{\substack{\E{X^2}\leq\const{P},\\ |X|\leq\const{A}}} I\bigl(P_X,W_{\set{D}}\bigr) \label{eq:am_proof_prop1_compact}
\end{IEEEeqnarray}
where $W_{\set{D}}$ denotes the channel law corresponding to
the quantization region $\set{D}$, and where
$I\bigl(P_X,W_{\set{D}}\bigr)$ denotes the mutual information of a
channel with law $W_{\set{D}}$ when its input is distributed according
to $P_X$. Clearly, the RHS of \eqref{eq:am_proof_prop1_compact} cannot
exceed its LHS, because imposing an additional peak-power constraint
cannot increase capacity. It remains to prove that the LHS cannot
exceed the RHS. 

By Fano's Inequality \cite[Th.~2.11.1]{coverthomas91} and the Data
Processing Inequality, we have that,
for every blocklength $n$, every encoder $m \mapsto
\bigl(x_1(m),\ldots,x_n(m)\bigr)$ of rate $R=\frac{\log\const{M}}{n}$
that satisfies the average-power constraint, and every quantization
region $\set{D}$, the probability of error is lower-bounded by
\cite[Sec.~8.9]{coverthomas91}
\begin{equation}
\label{eq:proof_prop1_converse}
\Prob(\hat{M}\neq M) \geq 1 - \frac{1}{nR} \sum_{k=1}^n I\bigl(X_k(M);Y_k\bigr)-\frac{1}{n R}.
\end{equation}
Let $\const{A}_{n}$ be the largest magnitude of the symbols that the encoder
can produce
\begin{equation}
\label{eq:am_fog130}
\const{A}_{n} \triangleq \max_{\substack{1\leq k \leq n,\\1\leq m
    \leq\const{M}}} |x_{k}(m)|
\end{equation}
so
\begin{equation}
\label{eq:am_fog30}
|x_k(m)| \leq \const{A}_{n}, \quad \bigl(k=1,2,\ldots,n,\,m=1,2,\ldots,\const{M}\bigr).
\end{equation}
With this notation, we have for every blocklength~$n$ and every quantizer~$\set{D}$,
\begin{align}
\frac{1}{n}\sum_{k=1}^n I\bigl(X_k(M);Y_k\bigr) & \leq 
\sup_{\substack{ \E{X^2}\leq\const{P},\\
    |X|\leq\const{A}_n}} I\bigl(P_X,W_{\set{D}}\bigr) \nonumber\\
& \leq \sup_{\const{A}>0} \sup_{\substack{
    \E{X^2}\leq\const{P},\\|X|\leq\const{A}}}
I\bigl(P_X,W_{\set{D}}\bigr) \label{eq:am_fog40}
\end{align}
where the first inequality follows from \eqref{eq:am_fog30} and by
the concavity of
\begin{equation*}
  \const{P} \mapsto \sup_{\substack{ \E{X^2}\leq\const{P},\\
    |X|\leq\const{A}_n}} I\bigl(P_X,W_{\set{D}}\bigr).
\end{equation*}
Thus, the RHS of \eqref{eq:proof_prop1_converse} is
bounded away from zero whenever $R$ exceeds the RHS of
\eqref{eq:am_fog40}, and the inequality 
\begin{equation}
  C(\const{P}, \set{D}) \leq \sup_{\const{A}>0} \sup_{\substack{
    \E{X^2}\leq\const{P},\\ |X|\leq\const{A}}}
I\bigl(P_X,W_{\set{D}}\bigr) \label{eq:dobi_2012}
\end{equation}
is established. Since the inner supremum on the RHS of \eqref{eq:dobi_2012} is monotonically nondecreasing in $\const{A}$, we can replace the outer supremum by a limit and thus establish \eqref{eq:am_proof_prop1_compact}. 


Introducing a peak-power constraint in
\eqref{eq:am_proof_prop1_compact} allows us next to establish the
existence of a capacity-achieving input distribution of three mass
points using Dubins's Theorem as follows. Recall that
by~\eqref{eq:am_proof_prop1_compact}
\begin{equation}
  C(\const{P}, \set{D})= \lim_{\const{A} \to \infty} 
  C_{\set{D},\const{A}}(\const{P})
\end{equation}
where $C_{\set{D},\const{A}}(\const{P})$ denotes the capacity of the
memoryless channel $\Prob\bigl(\tilde{Y}\in\set{D}\bigm|X=x\bigr)$
with the input $X$ taking values in the interval
$[-\const{A},\const{A}]$ and with the binary output~$Y$:
\begin{equation}
  \label{eq:amos_def_CDAP}
C_{\set{D},\const{A}}(\const{P}) \triangleq \sup_{\substack{
    \E{X^2}\leq\const{P},\\ |X|\leq\const{A}}}
I\bigl(P_X,W_{\set{D}}\bigr).
\end{equation}
Proceeding along the lines of \cite[Sec.~II-C]{witsenhausen80} but
accounting for the additional average-power constraint, it can be
shown that $C_{\set{D},\const{A}}(\const{P})$ is achieved by an input
distribution consisting of three mass points.
Indeed, since $\const{P}\mapsto C_{\set{D},\const{A}}(\const{P})$ is concave it is continuous, so there exists some
$\const{P}' \leq \const{P}$ such that 
\begin{equation}
  \label{eq:pad10}
  C_{\set{D},\const{A}}(\const{P}) = \sup_{\substack{
    \E{X^2} = \const{P}',\\|X|\leq\const{A}}}
I\bigl(P_X,W_{\set{D}}\bigr).
\end{equation}
The input distribution achieving $C_{\set{D},\const{A}}(\const{P})$
must be concentrated on the interval $[-\const{A}, \const{A}]$
and additionally satisfy
\begin{equation}
\label{eq:a_hyperplane2}
\int x^2 \d P_X(x) = \const{P}'.
\end{equation}
The arguments in \cite[Sec.~II-C]{witsenhausen80} thus go through with
the set~$A$ in \cite[Sec.~II-C]{witsenhausen80} replaced by the set of
input distributions that induce the given output distribution and that
additionally lie on the hyperplane \eqref{eq:a_hyperplane2}. 


Having established that under an additional peak-power constraint
capacity is achieved by a three-mass-points input distribution, we now
study what happens to these three mass points as the allowed
peak-power tends to infinity.  We thus study how the 
three mass points at locations
\begin{equation*}
  \bfxi = (\xi_{\text{L}}, \xi_{\text{M}}, \xi_{\text{R}})
\end{equation*}
with corresponding masses
\begin{equation*}
  \mathbf{p} = (p_{\text{L}}, p_{\text{M}}, p_{\text{R}})
\end{equation*}
behave as $\const{A}$ tends to infinity.

By possibly considering a subsequence of peak powers, we can assume
that, as $\const{A}$ tends to infinity, $\bfxi$ converges to some
$\bfxi^{\star} = (\xi_{\text{L}}^{\star}, \xi_{\text{M}}^{\star},
\xi_{\text{R}}^{\star})$ whose components are on the extended real line
$\Reals \cup \{\pm \infty\}$. Likewise we can assume that $\mathbf{p}$
converges to some probability vector $\mathbf{p}^{\star}$.  Since the
input distributions must satisfy the average-power constraint, if any
of the components of $\bfxi^{\star}$ is $\pm \infty$, then the
corresponding component of $\mathbf{p}^{\star}$ must be zero.
By Lemma~\ref{lem:amos_cont} (Appendix~\ref{app:amos}), $\Prob(\tilde{Y}\in\set{D}|X=\xi_{\ell})$ converges to $\Prob(\tilde{Y}\in\set{D}|X=\xi_{\ell}^{\star})$ whenever $\xi_{\ell}^{\star}\in\Reals$, and  the continuity of
\begin{IEEEeqnarray}{lCl}
C_{\set{D},\const{A}}(\const{P}) & = & H_b\Biggl(\sum_{\ell \in
  \{\text{L}, \text{M}, \text{R}\}} p_{\ell} \, \Prob\bigl(\tilde{Y}\in\set{D}\bigm|X=\xi_{\ell}\bigr)\Biggr) \nonumber\\
& & {} - \sum_{\ell \in
  \{\text{L}, \text{M}, \text{R}\}} p_{\ell} \, H_b\Bigl(\Prob\bigl(\tilde{Y}\in\set{D}\bigm|X=\xi_{\ell}\bigr)\Bigr)\nonumber
\end{IEEEeqnarray}
demonstrates that $\lim_{\const{A} \to \infty} C_{\set{D},
  \const{A}}(\const{P})$ (which equals $C(\const{P}, \set{D})$
by~\eqref{eq:am_proof_prop1_compact}) equals the
mutual information corresponding to $(\mathbf{p}^{\star}, \bfxi^{\star})$
provided that in computing the latter the mass points of zero mass are
ignored. Since the mass points at $\pm \infty$ are of zero mass (by
the average-power constraint), those are ignored, and we conclude that
$C(\const{P}, \set{D})$ is achieved by (at most) three
\emph{finite} mass point. For sufficiently large $\const{A}$
(exceeding the largest of these mass points) the peak-power constraint
is inactive.

\subsection{Quantizers for Three-Mass-Points Input Distributions}
\label{sec:KreinMilman}
Having established that for any quantizer $\set{D}$ the capacity
$C(\const{P}, \set{D})$ is achieved by a three-mass-points
input distribution, we now fix some arbitrary three-mass-points
input distribution\footnote{Every two-mass-points distribution can be viewed as a three-mass-points distribution with one of the masses being zero.} $P_{X}$ concentrated at $(\xi_{1}, \xi_{2},
\xi_{3})$ and study the quantizer that maximizes the mutual
information $I\bigl(P_X,W_{\set{D}}\bigr)$ corresponding to it. (Without loss of generality, we assume that $\xi_1\neq\xi_2$, $\xi_1\neq\xi_3$, and $\xi_2\neq\xi_3$.) We
will show that when $P_{X}$ is a three-mass-points input distribution, we have
\begin{equation}
  \label{eq:dobi_proof_thm2_1}
  \sup_{\set{D}} I\bigl(P_X,W_{\set{D}}\bigr) = \sup_{\Upsilon_1\leq\Upsilon_2} I\bigl(P_X,W_{\set{D}(\Upsilon_1,\Upsilon_2)}\bigr)
\end{equation}
where the quantizer $\set{D}(\Upsilon_1,\Upsilon_2)$ is defined as
\begin{equation}
\label{eq:three_regions}
\set{D}(\Upsilon_1,\Upsilon_2) \triangleq \{\tilde{y}\in\Reals\colon \Upsilon_1\leq \tilde{y}\leq\Upsilon_2\}, \quad \Upsilon_1\leq\Upsilon_2
\end{equation}
with
\begin{subequations}
\label{eq:amos_rays}
\begin{IEEEeqnarray}{lCl}
\set{D}(-\infty,\Upsilon_2) & \triangleq & \{\tilde{y}\in\Reals\colon \tilde{y}\leq\Upsilon_2\}, \quad \Upsilon_2\in\Reals \label{eq:D1}\\
\set{D}(\Upsilon_1,\infty)  & \triangleq & \{\tilde{y}\in\Reals\colon \tilde{y}\geq\Upsilon_1\}, \quad \Upsilon_1\in\Reals \label{eq:D2}\\
\set{D}(-\infty,\infty) & \triangleq & \Reals \label{eq:D3} \\
\set{D}(-\infty,-\infty) & = & \set{D}(\infty,\infty) \triangleq \varnothing.\label{eq:D4}
\end{IEEEeqnarray}
\end{subequations}
(Here $\varnothing$ denotes the empty set.) Needless to say, the case $\Upsilon_1=\Upsilon_2$ and the forms \eqref{eq:D3} and \eqref{eq:D4} yield zero
mutual information and are thus uninteresting.

Define
\begin{IEEEeqnarray}{ll}
\set{W} \triangleq \Bigl\{ & (\omega_1,\omega_2,\omega_3)\in [0,1]^3\colon \nonumber\\
& \,\omega_{\ell}=\Prob\bigl(\tilde{Y}\in\set{D}\bigm| X=\xi_{\ell}\bigr), \set{D}\subset\Reals\Bigr\}
\end{IEEEeqnarray}
as the set of possible channel laws that different quantizers
can induce for the inputs $(\xi_{1}, \xi_{2}, \xi_{3})$, and let
$\overline{\set{W}}$ denote the closure of the convex hull of
$\set{W}$. With this notation
\begin{IEEEeqnarray}{lCl}
  \sup_{\set{D}} I(P_{X}, W_{\set{D}}) & = & 
  \sup_{W\in\set{W}} I\bigl(P_X,W\bigr)\nonumber \\
  & \leq & \sup_{W\in\overline{\set{W}}}
  I\bigl(P_X,W\bigr) \label{eq:dobi_proof_thm2_2}
\end{IEEEeqnarray}
where the second step follows because $\set{W}\subseteq
\overline{\set{W}}$. Recall that an extreme point of
$\overline{\set{W}}$ is a channel in $\overline{\set{W}}$ that cannot
be written as a convex combination of two different
channels in~$\overline{\set{W}}$. By the Krein-Milman Theorem
\cite[Cor.~18.5.1]{rockafellar70}, every channel law
$W\in\overline{\set{W}}$ can be written as a convex combination of
extreme points of $\overline{\set{W}}$. Since mutual information is
convex in the channel law (when the input distribution is held fixed) \cite[Th.~2.7.4]{coverthomas91},
it follows that on the RHS of~\eqref{eq:dobi_proof_thm2_2} we can
replace the supremum over the set $\overline{\set{W}}$ with the
supremum over its extreme points.

We next show that the extreme points of $\overline{\set{W}}$
correspond to quantizers of the
form~\eqref{eq:three_regions}. Once we show this, it will follow
that~\eqref{eq:dobi_proof_thm2_2} holds with equality, because these
extreme points of $\overline{\set{W}}$ are in fact in $\set{W}$. This will prove \eqref{eq:dobi_proof_thm2_1}.\footnote{Note that $\overline{\set{W}}$ is the set of possible channel laws that different quantizers can induce for the inputs $(\xi_1,\xi_2,\xi_3)$, provided that we allow for randomized quantization rules. It thus follows that \eqref{eq:dobi_proof_thm2_1} continues to hold if on the LHS, instead of maximizing over all deterministic quantizers $\set{D}$, we maximize over all probability distributions $P_{Y|\tilde{Y}}$ with $Y$ binary.}

To prove that the extreme points of $\overline{\set{W}}$ are indeed the channel laws corresponding to quantizers of the form \eqref{eq:three_regions}, we consider the
\emph{support function} of $\overline{\set{W}}$ \cite[Sec.~13]{rockafellar70}
\begin{equation}
f(\bflambda) \triangleq \sup_{(\omega_1,\omega_2,\omega_3)\in\overline{\set{W}}} \{\lambda_1 \,\omega_1 + \lambda_2 \,\omega_2 + \lambda_3\,\omega_3\}
\end{equation}
for $\bflambda=(\lambda_1,\lambda_2,\lambda_3)\in\Reals^3$. Since $\overline{\set{W}}$ is the closure of all convex combinations of the elements of $\set{W}$ \cite[Th.~2.3]{rockafellar70}, the support function of $\overline{\set{W}}$ is the same as that of $\set{W}$ and
\begin{IEEEeqnarray}{lCl}
f(\bflambda) 
& = & \sup_{\set{D}} \bigl\{\lambda_1 \,\omega_1(\set{D}) + \lambda_2 \,\omega_2(\set{D}) + \lambda_3\,\omega_3(\set{D})\bigr\}\label{eq:supp}
\end{IEEEeqnarray}
where
\begin{equation}
\omega_{\ell}(\set{D}) \triangleq \Prob\bigl(\tilde{Y}\in\set{D}\bigm|X=\xi_{\ell}\bigr), \quad \ell=1,2,3.
\end{equation}
We rewrite \eqref{eq:supp} as
\begin{equation}
\label{eq:supp_int}
f(\bflambda) = \sup_{\set{D}} \frac{1}{\sqrt{2\pi\sigma^2}} \int_{\set{D}} g_{\bflambda}(\tilde{y})\d\tilde{y}
\end{equation}
where
\begin{IEEEeqnarray}{lCl}
g_{\bflambda}(\tilde{y}) & \triangleq & \lambda_1 e^{-\frac{(\tilde{y}-\xi_1)^2}{2\sigma^2}} \nonumber\\
& & {} + \lambda_2 e^{-\frac{(\tilde{y}-\xi_2)^2}{2\sigma^2}} + \lambda_3 e^{-\frac{(\tilde{y}-\xi_3)^2}{2\sigma^2}}, \quad \tilde{y}\in\Reals.\IEEEeqnarraynumspace
\end{IEEEeqnarray}
The integral on the RHS of \eqref{eq:supp_int} is maximized when $\set{D}$ is the set 
\begin{equation}
\set{D}^{\star}(\bflambda)=\bigl\{\tilde{y}\in\Reals\colon
g_{\bflambda}(\tilde{y})\geq 0\bigr\}.
\end{equation} 
The structure of $\set{D}^{\star}(\bflambda)$ depends on the
zeros of $g_{\bflambda}(\cdot)$, which we study next.

Our study of the zeros of $g_{\bflambda}(\cdot)$ depends on the signs
of $\lambda_1, \lambda_2,\lambda_3$ and on how many of them are zero. The case where $\lambda_1,
\lambda_2,\lambda_3$ are all zero is trivial, because in this case
$f(\bflambda)$ is zero irrespective of $\set{D}$. We will see that in
all other cases the set $\set{D}$ that achieves $f(\bflambda)$ is
unique up to Lebesgue measure zero. If exactly two $\lambda$'s, say
$\lambda_1$ and $\lambda_2$, are zero, then the set $\set{D}$ that
achieves $f(\bflambda)$ is either $\Reals$ or $\varnothing$, depending
on whether $\lambda_3$ is positive or negative.
We next consider the case where exactly one of the $\lambda$'s, say
$\lambda_3$, is zero. In this case
\begin{equation}
\label{eq:glambda_lambda30}
g_{\bflambda}(\tilde{y}) = \lambda_1 e^{-\frac{(\tilde{y}-\xi_1)^2}{2\sigma^2}} + \lambda_2 e^{-\frac{(\tilde{y}-\xi_2)^2}{2\sigma^2}}, \quad \tilde{y}\in\Reals
\end{equation}
which is either positive (if $\lambda_1>0$ and $\lambda_2>0$), negative (if $\lambda_1<0$ and $\lambda_2<0$), or has a zero at
\begin{equation}
\label{eq:zeros_lambda30}
\tilde{y} = \frac{\xi_1+\xi_2}{2} + \frac{\sigma^2}{\xi_2-\xi_1} \log\biggl|\frac{\lambda_1}{\lambda_2}\biggr|
\end{equation}
(if $\lambda_1$ and $\lambda_2$ have opposite signs). Consequently, if exactly one of the $\lambda$'s is zero, then
the set $\set{D}$ that achieves $f(\bflambda)$ is either the entire real line, the empty set, or a ray, i.e., of the form
$(-\infty, \Upsilon)$ or $(\Upsilon, \infty)$, where $\Upsilon$ is the
RHS of~\eqref{eq:zeros_lambda30}.

We finally turn to the case where all the $\lambda$'s are nonzero. If
they are all of equal sign, then $f(\bflambda)$ has no zeros and the
set $\set{D}$ that maximizes $f(\bflambda)$ is either the entire real
line $\Reals$ or the empty set, depending on whether the $\lambda$'s
are all positive or all negative. It remains to study the case where
the $\lambda$'s are nonzero but not of equal sign. Changing the sign
of all the $\lambda$'s is tantamount to multiplying
$g_{\bflambda}(\cdot)$ by $-1$ and therefore does not change the
locations of the zeros, so we can assume without loss of generality that
one of the $\lambda$'s, say $\lambda_{1}$, is positive and that the
remaining two $\lambda_{2}, \lambda_{3}$ are negative. In this case
\begin{IEEEeqnarray}{lCl}
g_{\bflambda}(\tilde{y}) 
& = & \lambda_1 e^{-\frac{(\tilde{y}-\xi_1)^2}{2\sigma^2}} h_{\bflambda}(\tilde{y}), \quad \tilde{y}\in\Reals
\end{IEEEeqnarray}
where
\begin{IEEEeqnarray}{lCl}
h_{\bflambda}(\tilde{y}) & \triangleq & 1-\biggl|\frac{\lambda_2}{\lambda_1}\biggr| e^{\frac{\xi_1^2-\xi_2^2}{2\sigma^2}}e^{\tilde{y}\frac{\xi_2-\xi_1}{\sigma^2}} \nonumber\\
& & {} - \biggl|\frac{\lambda_3}{\lambda_1}\biggr| e^{\frac{\xi_1^2-\xi_3^2}{2\sigma^2}}e^{\tilde{y}\frac{\xi_3-\xi_1}{\sigma^2}}, \quad \tilde{y}\in\Reals.
\end{IEEEeqnarray}
Note that the zeros of $g_{\bflambda}(\cdot)$
are the same as the zeros of $h_{\bflambda}(\cdot)$. Further note that
$h_{\bflambda}(\cdot)$ is a nonzero analytic function whose second
derivative
\begin{IEEEeqnarray}{lCl}
\frac{\partial^2}{\partial \tilde{y}^2} h_{\bflambda}(\tilde{y}) & = &
-\frac{(\xi_2-\xi_1)^2}{\sigma^4}\biggl|\frac{\lambda_2}{\lambda_1}\biggr|
e^{\frac{\xi_1^2-\xi_2^2}{2\sigma^2}}e^{\tilde{y}\frac{\xi_2-\xi_1}{\sigma^2}}\nonumber\\
&& {} -
\frac{(\xi_3-\xi_1)^2}{\sigma^4}\biggl|\frac{\lambda_3}{\lambda_1}\biggr|
e^{\frac{\xi_1^2-\xi_3^2}{2\sigma^2}}e^{\tilde{y}\frac{\xi_3-\xi_1}{\sigma^2}}, \quad\!\! \tilde{y}\in\Reals\IEEEeqnarraynumspace
\label{eq:amosRolle}
\end{IEEEeqnarray}
 is strictly negative. Consequently, $h_{\bflambda}(\cdot)$---and hence also
$g_{\bflambda}(\cdot)$---can have at most two zeros. (If it had three
or more, then by Rolle's Theorem its derivative would have at least
two zeros, and its second derivative would therefore have a zero in
contradiction to~\eqref{eq:amosRolle}.) If $h_{\bflambda}(\cdot)$ has at most one zero, then the set $\set{D}$ achieving $f(\bflambda)$ is either the entire real line, the empty set, or a ray. If it has two zeros, then $\set{D}$ comprises two disjoint rays or
else a finite interval---either way, $\set{D}$ or its complement is a
finite interval.

We next show that for every $\bflambda\neq \bfzero$ the quantization
region achieving $f(\bflambda)$ is unique up to sets of Lebesgue
measure zero. Let $\set{D}^{\star}(\bflambda)$ be the quantization
region that achieves $f(\bflambda)$, and let $\set{D}_1$
be any other quantization region. Then
\begin{IEEEeqnarray}{lCl}
\IEEEeqnarraymulticol{3}{l}{\int_{\set{D}^{\star}(\bflambda)}g_{\bflambda}(\tilde{y})\d\tilde{y}-\int_{\set{D}_1}g_{\bflambda}(\tilde{y})\d\tilde{y}} \nonumber\\
\quad & = & \int_{\set{D}^{\star}(\bflambda)\cap\set{D}_1^c} g_{\bflambda}(\tilde{y})\d\tilde{y} - \int_{\set{D}_1\cap\set{D}^{\star}(\bflambda)^c}g_{\bflambda}(\tilde{y})\d\tilde{y} \nonumber\\
& \geq & \int_{\set{D}^{\star}(\bflambda)\cap\set{D}_1^c} g_{\bflambda}(\tilde{y})\d\tilde{y} \nonumber\\
& \geq & 0
\end{IEEEeqnarray}
where the second step follows because for every
$\tilde{y}\in\set{D}^{\star}(\bflambda)^c$ we have
$g_{\bflambda}(\tilde{y})<0$; and the last step follows because for
every $\tilde{y}\in\set{D}^{\star}(\bflambda)$ we have
$g_{\bflambda}(\tilde{y})\geq 0$. (Here $\set{A}^c$ denotes the complement of the set $\set{A}$.) Furthermore, since the zeros of
$g_{\bflambda}(\cdot)$ are isolated, it is nonzero almost everywhere,
so the inequalities hold with equality if, and only if,
$\set{D}^{\star}(\bflambda)\cap\set{D}_1^c$ and
$\set{D}_1\cap\set{D}^{\star}(\bflambda)^c$ have both Lebesgue measure
zero.

Because quantizers that differ on a set of Lebesgue measure zero induce identical channel laws,
the uniqueness (up to sets of Lebesgue measure zero) of the set
$\set{D}$ achieving $f(\bflambda)$ (for $\bflambda \neq \bfzero$)
implies that for every $\bflambda\neq \bfzero$ the tuple
$(\omega_1^{\star},\omega_2^{\star},\omega_3^{\star})$ that achieves
$f(\bflambda)$ is unique. 

We next  note that, by \cite[Th.~13.1]{rockafellar70}, every $(\omega_1,\omega_2,\omega_3)\in\overline{\set{W}}$ satisfying
\begin{equation*}
\lambda_1 \omega_1+\lambda_2 \omega_2 + \lambda_3 \omega_3 < f(\bflambda), \quad \textnormal{for every $\bflambda\neq 0$}
\end{equation*}
must be an interior point of $\overline{\set{W}}$. Since an interior point cannot be an extreme point, it follows that every extreme point of a compact convex set achieves the
supremum defining $f(\bflambda)$ at some $\bflambda\neq 0$. Furthermore, since for a
given $\bflambda\neq 0$ the support function $f(\bflambda)$ is achieved
uniquely by a channel law that is induced by a quantizer of the
form~\eqref{eq:three_regions} or their complement, it follows that the
extreme points of $\overline{\set{W}}$ are all achieved by quantizers
of this form or their complement.  Recalling that mutual information is
maximized over $\overline{\set{W}}$ (for a given input distribution)
at an extreme point, and noting that the mutual information corresponding to
the quantizer $\set{D}$ is the same as that corresponding to its
complement, we conclude that---for any fixed three-mass-points input distribution---the supremum over all quantizers can be replaced with the supremum over all quantizers of the form \eqref{eq:three_regions}, thus proving \eqref{eq:dobi_proof_thm2_1}.

\subsection{The Supremum Defining $C(\const{P})$ Is Achieved}
\label{sub:dobi_sup=max}
Having established that to each quantizer the optimal
input distribution is of three mass points, and having established
that to each three-mass-points input distribution the optimal
quantizer is of the form~\eqref{eq:three_regions}, we conclude that we
can express $C(\const{P})$ of \eqref{eq:capacity} as
\begin{equation}
\label{eq:app_supmax}
C(\const{P}) =
\sup_{\substack{(\vect{p},\bfxi)\colon\E{X^2}\leq\const{P},\\\Upsilon_1\leq\Upsilon_2}}
I\bigl(\vect{p},\vect{W}(\Upsilon_1,\Upsilon_2|\bfxi)\bigr)
\end{equation}
where $(\vect{p},\bfxi)$ denotes the three-mass-points distribution of
masses \[\vect{p} = (p_{1}, p_{2}, p_{3}) \in [0,1]^{3}\] and locations
\[\bfxi = (\xi_{1}, \xi_{2}, \xi_{3}) \in \Reals^{3}\] and where $\vect{W}(\Upsilon_1,\Upsilon_2|\bfxi)$ denotes the channel law corresponding to the quantizer
$\set{D}(\Upsilon_{1}, \Upsilon_{2})$ and to the mass points $\xi_{\ell}$, $\ell=1,2,3$:
\begin{equation}
W\bigl(\Upsilon_1,\Upsilon_2 \bigm| \xi_{\ell}\bigr) \triangleq \Prob\bigl(\tilde{Y}\in\set{D}(\Upsilon_1,\Upsilon_2)\bigm| X= \xi_{\ell}\bigr).
\end{equation}
We next show that this supremum is achieved.

By the definition of the supremum, there exists a sequence
$\bigl\{(\vect{p}_i,\bfxi_i,\Upsilon_{1,i},\Upsilon_{2,i}),\,i\in\Naturals\bigr\}$
(where $\Naturals$ denotes the set of positive integers) such that
\begin{equation}
  \label{eq:erasure10}
\lim_{i\to\infty} I\bigl(\vect{p}_i,\vect{W}(\Upsilon_{1,i},\Upsilon_{2,i}|\bfxi_i)\bigr) = C(\const{P}).
\end{equation}
By taking a subsequence (if needed), we may assume without loss of
generality that $\vect{p}_{i}$ converges to some $\vect{p}^{\star}$,
that $\bfxi_{i}$ converges to some $\bfxi^{\star}$ (whose components
may be $\pm \infty$) and that $\Upsilon_{1,i}$ and
$\Upsilon_{2,i}$ converge to $\Upsilon_{1}^{\star}$ and
$\Upsilon_{2}^{\star}$, both of which may be $\pm \infty$.  
From the continuity of the cumulative distribution
function of the Normal distribution, it follows that, whenever
$\xi^{\star}_{\ell}$ is finite,
\begin{IEEEeqnarray}{lCl}
  \IEEEeqnarraymulticol{3}{l}{\lim_{i \to \infty} \Prv{\Upsilon_{1,i} \leq \xi_{\ell,i} + Z
    \leq \Upsilon_{2,i}}} \nonumber\\
    \qquad\qquad\quad & = & \Prv{\Upsilon_{1}^{\star} \leq \xi_{\ell}^{\star} + Z
    \leq \Upsilon_{2}^{\star}}\label{eq:amos_cont_sup}
\end{IEEEeqnarray}
where we recall that $Z$ is a centered Gaussian random variable of positive variance
$\sigma^{2}$.

Since the mass $p_{\ell}^{\star}$ corresponding to nonfinite locations
$\xi_{\ell}^{\star}$ is zero (by the average-power constraint), and
since $p_{\ell,i}$ converges to $p_{\ell}^{\star}$,
\eqref{eq:amos_cont_sup} and the continuity of the binary entropy function allow us to infer that
\begin{IEEEeqnarray}{lCl}
  \IEEEeqnarraymulticol{3}{l}{\lim_{i \to \infty} I\bigl(\vect{p}_i,\vect{W}(\Upsilon_{1,i},\Upsilon_{2,i}|\bfxi_i)\bigr)}\nonumber\\
  \quad & = & \lim_{i \to \infty} \Biggl\{H_b\Biggl(\sum_{\ell=1}^3 p_{\ell,i} \,W(\Upsilon_{1,i},\Upsilon_{2,i} | \xi_{\ell,i})\Biggr) \nonumber\\
  & & \qquad\quad {} - \sum_{\ell=1}^3 p_{\ell,i}\, H_b\bigl(W(\Upsilon_{1,i},\Upsilon_{2,i} | \xi_{\ell,i})\bigr) \Biggr\} \nonumber\\
  & = &  I\bigl(\vect{p}^{\star},\vect{W}(\Upsilon_{1}^{\star},\Upsilon_{2}^{\star}|\bfxi^{\star})\bigr)
\end{IEEEeqnarray}
provided that in computing the mutual information on the LHS of \eqref{eq:sup=max} the mass points of zero mass are ignored. This combines with~\eqref{eq:erasure10} to imply that
\begin{equation}
\label{eq:sup=max}
 I\bigl(\vect{p}^\star,\vect{W}(\Upsilon_{1}^{\star},\Upsilon_{2}^{\star}|\bfxi^\star)\bigr)
  = C(\const{P}).
\end{equation}
Noting that the mass points at $\pm \infty$ are of zero mass and therefore ignored, we conclude that $C(\const{P})$ is achieved by an input distribution of (at most) three \emph{finite} mass points and by a quantizer of the form \eqref{eq:three_regions}.

\subsection{A Threshold Quantizer Is Optimal}
\label{sub:dobi_threshold}
Having established that $C(\const{P})$ is achieved by a
three-mass-points input distribution and a quantizer of the
form~\eqref{eq:three_regions}, we now prove that $C(\const{P})$ is in fact achieved by a three-mass-points input distribution and a threshold quantizer, i.e., a quantizer of the form \eqref{eq:D2}.
Clearly $\Upsilon_{1}$ and $\Upsilon_{2}$ cannot be both nonfinite, as this
would result in zero mutual information, whereas $C(\const{P})$ is
strictly positive whenever $\const{P}$ is positive\footnote{This can be verified by noting that a symmetric threshold quantizer
and an equiprobable $\pm \sqrt{\const{P}}$ input distribution yield
positive mutual information for every positive $\const{P}$, cf.~\eqref{eq:Csym}.}
\begin{equation}
  \label{eq:dobi_proof_threshold_pos}
  C(\const{P}) > 0, \quad \const{P} > 0.
\end{equation}
For the same reason we can assume, without loss of optimality, that $\Upsilon_1\neq\Upsilon_2$. Since \eqref{eq:D1} is the complement of a set of the form \eqref{eq:D2}---which gives rise to the same mutual information---it remains to rule out the case where $\Upsilon_1$ and $\Upsilon_2$ are both finite.

We shall prove this by
contradiction. We shall assume that the quantization region
$\set{D}(\Upsilon_1,\Upsilon_2)$ for some finite
$\Upsilon_1 < \Upsilon_2$ is optimal and derive a contradiction to
optimality. Assume then that $\Upsilon_{1}$ and $\Upsilon_{2}$ are
both finite with $\Upsilon_{1} < \Upsilon_{2}$. Define
\begin{equation}
\label{eq:theta}
\theta \triangleq \frac{\Upsilon_1+\Upsilon_2}{2}.
\end{equation}
Let $\bfxi$ be the mass points of the capacity-achieving
input distribution, and let $\vect{p}$ be the corresponding
probabilities. Note that there is no loss in optimality in
assuming that $\theta$ is nonnegative
\begin{equation}
  \theta \geq 0
\end{equation}
because if $\theta$ is negative, then we can consider the input
$(\vect{p},-\bfxi)$ (whose second moment is identical to that of
$(\vect{p}, \bfxi)$) and the quantizer
$\set{D}(-\Upsilon_2,-\Upsilon_1)$ (whose midpoint is of opposite sign
to that of $\set{D}(\Upsilon_1,\Upsilon_2)$) which give rise to the
same mutual information as the input $(\vect{p}, \bfxi)$ and the
quantizer $\set{D}(\Upsilon_1,\Upsilon_2)$.

Assume that the mass points are ordered, i.e., $\xi_1<\xi_2<\xi_3$. Since the locations of mass points of zero mass have no effect on the mutual information, there is no loss in optimality in assuming that the probability of the largest mass point satisfies $p_3>0$. Furthermore, $p_3<1$ since $p_3=1$ would imply that $C(\const{P})=0$, $\const{P}>0$ in contradiction to \eqref{eq:dobi_proof_threshold_pos}.

We continue by noting that the symmetry of the Normal distribution
implies that 
\begin{equation}
\label{eq:symmetry_theta}
W\bigl(\Upsilon_1,\Upsilon_2\bigm| \theta-\delta\bigr) = W\bigl(\Upsilon_1,\Upsilon_2\bigm|\theta+\delta\bigr), \quad \delta\geq 0.
\end{equation}
Indeed, defining $\Delta\triangleq(\Upsilon_2-\Upsilon_1)/2$ (so $\Upsilon_1=\theta-\Delta$ and $\Upsilon_2=\theta+\Delta$), we have
\begin{IEEEeqnarray}{lCl}
W\bigl(\Upsilon_1,\Upsilon_2\bigm|\theta-\delta\bigr) & = & \int_{\theta-\Delta}^{\theta+\Delta} \frac{1}{\sqrt{2\pi\sigma^2}} e^{-\frac{(\tilde{y}-\theta+\delta)^2}{2\sigma^2}} \d\tilde{y} \nonumber\\
& = & \int_{\theta-\Delta}^{\theta+\Delta} \frac{1}{\sqrt{2\pi\sigma^2}} e^{-\frac{(-\tau+\theta+\delta)^2}{2\sigma^2}} \d\tau \IEEEeqnarraynumspace\nonumber\\
& = & W\bigl(\Upsilon_1,\Upsilon_2\bigm| \theta+\delta\bigr)
\end{IEEEeqnarray}
where we made the substitution $\tau=-\tilde{y}+2\theta$. Furthermore, since
$\theta\geq 0$, 
\begin{equation}
\label{eq:below_theta}
(\theta-\delta)^2 
\leq (\theta+\delta)^2, \quad \delta\geq 0.
\end{equation} 
As we next argue, \eqref{eq:symmetry_theta} and \eqref{eq:below_theta}
imply that there is no loss in optimality in assuming that
\begin{equation}
\label{eq:theta_ordered}
\xi_1 < \xi_2 < \xi_3 \leq \theta.
\end{equation}
Indeed, suppose $\xi_3>\theta$. Then $\xi_3$ can be written as
$\theta+\delta$, for some $\delta>0$. However,
$\tilde\xi_3=\theta-\delta$ gives rise to the same channel law
\eqref{eq:symmetry_theta} but has a smaller cost
\eqref{eq:below_theta}. Thus, for every $\xi_3>\theta$ we can find a
$\tilde\xi_3<\theta$ satisfying the power constraint that achieves the
same rate.

We next show that \eqref{eq:theta_ordered} leads to a contradiction
by considering a perturbation of the quantizer. For every 
$\Gamma>\Upsilon_2$ define the perturbed quantization region
\begin{equation}
\tilde{\set{D}} \triangleq (\Upsilon_1, \Upsilon_2) \cup [\Gamma, +\infty)
\end{equation}
and denote the channel law corresponding to $\tilde{\set{D}}$ and
$\bfxi$ by $\vect{W}(\tilde{\set{D}}|\bfxi)$:
\begin{align}
W\bigl(\tilde{\set{D}}\bigm| \xi_{\ell}\bigr) & 
\triangleq \Prob\bigl(\tilde{Y}\in\tilde{\set{D}}\bigm|
X=\xi_{\ell}\bigr) \nonumber\\
& = W(\Upsilon_1,\Upsilon_2 | \xi_{\ell}) +
Q\biggl(\frac{\Gamma-\xi_{\ell}}{\sigma}\biggr)\label{eq:WD}
\end{align}
for $\ell=1,2,3$. We will contradict the optimality of the input $(\vect{p},\bfxi)$
and the quantizer $\set{D}(\Upsilon_1,\Upsilon_2)$
by showing that for $(\vect{p},\bfxi)$
satisfying \eqref{eq:theta_ordered}, we can find a sufficiently large 
$\Gamma$ exceeding $\Upsilon_{2}$ such that
\begin{equation}
\label{eq:contradiction}
I\bigl(\vect{p},\vect{W}(\tilde{\set{D}}|\bfxi)\bigr) > I\bigl(\vect{p},\vect{W}(\Upsilon_1,\Upsilon_2|\bfxi)\bigr).
\end{equation}
To show this we use \eqref{eq:WD} to express the mutual information on the LHS of \eqref{eq:contradiction} as
\begin{IEEEeqnarray}{lCl}
I\bigl(\vect{p},\vect{W}(\tilde{\set{D}}|\bfxi)\bigr) & = & H_b\bigl(P(\Upsilon_1,\Upsilon_2) + P(\Gamma)\bigr) \nonumber\\
\IEEEeqnarraymulticol{3}{r}{\qquad {} -  \sum_{\ell=1}^3 p_{\ell} H_b\Biggl(W\bigl(\Upsilon_1,\Upsilon_2\bigm| \xi_{\ell}\bigr)+Q\biggl(\frac{\Gamma-\xi_{\ell}}{\sigma}\biggr)\Biggr) \IEEEeqnarraynumspace}
\end{IEEEeqnarray}
where
\begin{subequations}
\begin{IEEEeqnarray}{rCl}
P(\Upsilon_1,\Upsilon_2) & \triangleq & \sum_{\ell=1}^3 p_{\ell}\, W(\Upsilon_1,\Upsilon_2 | \xi_{\ell}) \\
P(\Gamma) & \triangleq & \sum_{\ell=1}^3 p_{\ell}\, Q\biggl(\frac{\Gamma-\xi_{\ell}}{\sigma}\biggr). \label{eq:tobi_PGamma}
\end{IEEEeqnarray}
\end{subequations}
A Taylor series expansion of $H_b(p+\eps)$ around $p$ yields
\begin{equation}
H_b(p+\eps) = H_b(p) + \eps \log\frac{1-p}{p} + \mathsf{R}(p,\eps)
\end{equation}
for $0<p<1-\eps$ and some remainder $\mathsf{R}(p,\eps)$ satisfying
\begin{equation}
\label{eq:remainder}
|\mathsf{R}(p,\eps)| \leq \frac{\eps^2}{2} \frac{1}{p(1-p-\eps)}.
\end{equation}
With this, we obtain
\begin{IEEEeqnarray}{lCl}
\IEEEeqnarraymulticol{3}{l}{I\bigl(\vect{p},\vect{W}({\tilde{\set{D}}|\bfxi})\bigr)} \nonumber\\
\quad & = & H_b\bigl(P(\Upsilon_1,\Upsilon_2)\bigr) + P(\Gamma)\log\frac{1-P(\Upsilon_1,\Upsilon_2)}{P(\Upsilon_1,\Upsilon_2)} \nonumber\\
& & {} -  \sum_{\ell=1}^3 p_{\ell} H_b\Bigl(W\bigl(\Upsilon_1,\Upsilon_2\bigm| \xi_{\ell}\bigr)\Bigr) \nonumber\\
& & {} -\sum_{\ell=1}^3 p_{\ell} Q\biggl(\frac{\Gamma-\xi_{\ell}}{\sigma}\biggr) \log\frac{1-W\bigl(\Upsilon_1,\Upsilon_2\bigm| \xi_{\ell}\bigr)}{W\bigl(\Upsilon_1,\Upsilon_2\bigm| \xi_{\ell}\bigr)} \nonumber\\
& &{} + \mathsf{K}(\vect{p},\bfxi,\Gamma) \nonumber\\
& = & I\bigl(\vect{p},\vect{W}(\Upsilon_1,\Upsilon_2|\bfxi)\bigr) + P(\Gamma)\log\frac{1-P(\Upsilon_1,\Upsilon_2)}{P(\Upsilon_1,\Upsilon_2)} \nonumber\\
& & {} - \sum_{\ell=1}^3 p_{\ell} Q\biggl(\frac{\Gamma-\xi_{\ell}}{\sigma}\biggr) \log\frac{1-W\bigl(\Upsilon_1,\Upsilon_2\bigm| \xi_{\ell}\bigr)}{W\bigl(\Upsilon_1,\Upsilon_2\bigm| \xi_{\ell}\bigr)} \nonumber\\
& & {}  + \mathsf{K}(\vect{p},\bfxi,\Gamma) \label{eq:88}
\end{IEEEeqnarray}
where
\begin{IEEEeqnarray}{lCl}
\mathsf{K}(\vect{p},\bfxi,\Gamma) & \triangleq & \mathsf{R}\bigl(P(\Upsilon_1,\Upsilon_2),P(\Gamma)\bigr) \nonumber\\
\IEEEeqnarraymulticol{3}{r}{\qquad{}  - \sum_{\ell=1}^3 p_{\ell}
\mathsf{R}\Biggl(W\bigl(\Upsilon_1,\Upsilon_2\bigm|
\xi_{\ell}\bigr),Q\biggl(\frac{\Gamma-\xi_{\ell}}{\sigma}\biggr)\Biggr).\IEEEeqnarraynumspace}
\label{eq:amos_zik}
\end{IEEEeqnarray}

Since the LHS of \eqref{eq:WD} is strictly smaller than~$1$ so is its
RHS and it follows upon averaging over $\vect{p}$ that for every
\mbox{$\const{P}>0$} and every $\Upsilon_1\leq\Upsilon_2 < \Gamma$
\begin{equation}
P(\Upsilon_1,\Upsilon_2) + P(\Gamma)< 1.
\end{equation}
Furthermore, $P(\Upsilon_1,\Upsilon_2)$ is strictly positive since $W(\Upsilon_1,\Upsilon_2|\xi_{\ell})>0$ for $\ell=1,2,3$.
Using \eqref{eq:remainder}, it thus follows that
\begin{IEEEeqnarray}{lCl}
\IEEEeqnarraymulticol{3}{l}{\lim_{\Gamma\to\infty} \frac{\bigl|\mathsf{R}\bigl(P(\Upsilon_1,\Upsilon_2),P(\Gamma)\bigr)\bigr|}{Q\Bigl(\frac{\Gamma-\xi_3}{\sigma}\Bigr)}} \nonumber\\
\quad & \leq & \lim_{\Gamma\to\infty} \frac{[P(\Gamma)]^2}{Q\Bigl(\frac{\Gamma-\xi_3}{\sigma}\Bigr)} \frac{1}{2\, P(\Upsilon_1,\Upsilon_2)\bigl(1-P(\Upsilon_1,\Upsilon_2)-P(\Gamma)\bigr)} \nonumber\\
& \leq & \lim_{\Gamma\to\infty}  \frac{Q\Bigl(\frac{\Gamma-\xi_3}{\sigma}\Bigr)}{2\, P(\Upsilon_1,\Upsilon_2)\bigl(1-P(\Upsilon_1,\Upsilon_2)-P(\Gamma)\bigr)} \nonumber\\
& = & 0 \label{eq:dobi_remainder1}
\end{IEEEeqnarray}
where the second step follows because $\xi_1<\xi_2<\xi_3$, which implies that
\begin{IEEEeqnarray}{rCll}
P(\Gamma) & \leq & Q\biggl(\frac{\Gamma-\xi_3}{\sigma}\biggr)\nonumber
\end{IEEEeqnarray}
and where the last step follows because $P(\Gamma)$ and \mbox{$Q\bigl((\Gamma-\xi_{3})/\sigma\bigr)$} both tend to zero as $\Gamma$ tends to infinity. Along the same lines, it can be shown that for $\ell=1,2,3$
\begin{IEEEeqnarray}{lCl}
\lim_{\Gamma\to\infty} \frac{\biggl|\mathsf{R}\biggl(W(\Upsilon_1,\Upsilon_2|\xi_{\ell}),Q\Bigl(\frac{\Gamma-\xi_{\ell}}{\sigma}\Bigr)\biggr)\biggr|}{Q\Bigl(\frac{\Gamma-\xi_3}{\sigma}\Bigr)} & = & 0.\label{eq:dobi_remainder2}
\end{IEEEeqnarray}
It thus follows from \eqref{eq:amos_zik}, \eqref{eq:dobi_remainder1}, \eqref{eq:dobi_remainder2}, and the Triangle Inequality that
\begin{IEEEeqnarray}{lCl}
\IEEEeqnarraymulticol{3}{l}{\lim_{\Gamma\to\infty} \frac{\bigl|\mathsf{K}(\vect{q},\bfxi,\Gamma)\bigr|}{Q\Bigl(\frac{\Gamma-\xi_3}{\sigma}\Bigr)}} \nonumber\\
\quad & \leq & \lim_{\Gamma\to\infty} \frac{\bigl|\mathsf{R}\bigl(P(\Upsilon_1,\Upsilon_2),P(\Gamma)\bigr)\bigr|}{Q\Bigl(\frac{\Gamma-\xi_3}{\sigma}\Bigr)}\nonumber\\
& & {} + \lim_{\Gamma\to\infty} \sum_{\ell=1}^3 p_{\ell} \frac{\biggl|\mathsf{R}\biggl(W(\Upsilon_1,\Upsilon_2|\xi_{\ell}),Q\Bigl(\frac{\Gamma-\xi_{\ell}}{\sigma}\Bigr)\biggr)\biggr|}{Q\Bigl(\frac{\Gamma-\xi_3}{\sigma}\Bigr)} \nonumber\\
& = & 0. \label{eq:eps2}
\end{IEEEeqnarray}
We further have by \cite[Prop.~19.4.2]{lapidoth09} that for $\ell=1,2$
\begin{IEEEeqnarray}{lCl}
\lim_{\Gamma\to\infty} \frac{Q\Bigl(\frac{\Gamma-\xi_{\ell}}{\sigma}\Bigr)}{Q\Bigl(\frac{\Gamma-\xi_3}{\sigma}\Bigr)} & \leq & 
\lim_{\Gamma\to\infty} \frac{\Gamma-\xi_3}{\Gamma-\xi_{\ell}}\frac{e^{\frac{\xi_3^2-\xi_{\ell}^2}{2\sigma^2}}}{1-\frac{\sigma^2}{(\Gamma-\xi_3)^2}} e^{-\Gamma\frac{\xi_3-\xi_{\ell}}{\sigma^2}} \nonumber\\
& = & 0. \label{eq:eps1}
\end{IEEEeqnarray}
We thus obtain from \eqref{eq:tobi_PGamma}, \eqref{eq:88}, \eqref{eq:eps2}, and \eqref{eq:eps1} that
\begin{IEEEeqnarray}{lCl}
\IEEEeqnarraymulticol{3}{l}{\lim_{\Gamma\to\infty} \frac{I\bigl(\vect{p},\vect{W}(\tilde{\set{D}}|\bfxi)\bigr)-I\bigl(\vect{p},\vect{W}(\Upsilon_1,\Upsilon_2|\bfxi)\bigr)}{Q\Bigl(\frac{\Gamma-\xi_3}{\sigma}\Bigr)}}\nonumber\\
\quad & = & p_3 \log\frac{1-P(\Upsilon_1,\Upsilon_2)}{P(\Upsilon_1,\Upsilon_2)} - p_3 \log\frac{1-W\bigl(\Upsilon_1,\Upsilon_2\bigm| \xi_3\bigr)}{W\bigl(\Upsilon_1,\Upsilon_2\bigm| \xi_3\bigr)} \nonumber\\
& = & p_3 \Biggl(\log\frac{1-P(\Upsilon_1,\Upsilon_2)}{1-W\bigl(\Upsilon_1,\Upsilon_2\bigm| \xi_3\bigr)} + \log\frac{W\bigl(\Upsilon_1,\Upsilon_2\bigm| \xi_3\bigr)}{P(\Upsilon_1,\Upsilon_2)}\Biggr) \nonumber\\
& > & 0 \label{eq:positive}
\end{IEEEeqnarray}
where the inequality follows from the assumption $p_3>0$ and by noting that \[\xi\mapsto W\bigl(\Upsilon_1,\Upsilon_2\bigm| \xi\bigr)\] is strictly increasing on $(-\infty,\theta)$ (see Appendix~\ref{app:increasing}), which together with $p_3<1$ implies that
\begin{equation}
\label{eq:positive_II}
W\bigl(\Upsilon_1,\Upsilon_2\bigm| \xi_3 \bigr) > P(\Upsilon_1,\Upsilon_2).
\end{equation}
Consequently, for a sufficiently
large $\Gamma$,
$I\bigl(\vect{p},\vect{W}(\tilde{\set{D}}|\bfxi)\bigr)$ is strictly
larger than
$I\bigl(\vect{p},\vect{W}(\Upsilon_1,\Upsilon_2|\bfxi)\bigr)$,
contradicting the assumption that $\set{D}(\Upsilon_1,\Upsilon_2)$
with finite \mbox{$\Upsilon_1\leq\Upsilon_2$} achieves $C(\const{P})$.

\subsection{Centered, Variance-$\const{P}$ Input Distribution}
\label{sub:dobi_input}
We have shown that the supremum in \eqref{eq:capacity} is achieved by some input distribution that is concentrated on at most three points and by some threshold quantizer:
\begin{equation}
\label{eq:dobi_input_1}
C(\const{P}) = I\bigl(\vect{p}^{\star},\vect{W}(\Upsilon^{\star}|\bfxi^{\star})\bigr)
\end{equation}
where $\bfxi^{\star} \in \Reals^3$ is the location of the mass points,
$\vect{p}^{\star}$ is their corresponding probabilities, $\Upsilon^{\star}$ is
the threshold of the quantizer, and $\vect{W}(\Upsilon^{\star}|\bfxi^{\star})$
is the resulting channel law.
We next show that the input distribution $(\vect{p}^{\star}, \bfxi^{\star})$
must be centered and must satisfy the average-power constraint with
equality:
\begin{subequations}
\begin{IEEEeqnarray}{rCl}
  \sum_{\ell=1}^{3} p^{\star}_{\ell} \, \xi^{\star}_{\ell} & = & 0 \\
  \sum_{\ell=1}^{3} p^{\star}_{\ell} \bigl(\xi^{\star}_{\ell}\bigr)^{2} & = & \const{P}.
\end{IEEEeqnarray}
\end{subequations}
To show this we note that, for a fixed threshold quantizer~$\Upsilon^{\star}$, the capacity as a function of the maximal-allowed
average-power is a concave nondecreasing function that is strictly
smaller than~$1$~bit per channel use, and that tends to~$1$~bit per
channel use as the maximal-allowed average-power tends to
infinity. Consequently, this capacity-cost function must be strictly
increasing and the second moment of $(\vect{p}^{\star},
\bfxi^{\star})$ must therefore be $\const{P}$. By noting that the capacity is achieved by some threshold quantizer, this argument also
proves that $C(\const{P})$ must be strictly increasing in $\const{P}$. This further implies that $(\vect{p}^{\star}, \bfxi^{\star})$ must be centered
because otherwise we could shift $\bfxi^{\star}$ and
$\Upsilon^{\star}$ by the mean and thus reduce the second moment
without changing the mutual information.

\section{Proofs: Capacity Per Unit-Energy}
\label{sec:CUE_proofs}

\subsection{Proof of Theorem~\ref{thm:No2dB}}
\label{sub:proof_thm3}
We will lower-bound the RHS of \eqref{eq:CUEKL} by restricting the
supremum to threshold quantizers \eqref{eq:dobi_thresholdD}
and thus demonstrate that
\begin{equation}
\label{eq:CUCLB}
\dot{C}(0) \geq \frac{1}{2\sigma^2}.
\end{equation}
Together with the upper bound~\eqref{eq:CUCUB}, this will prove
Theorem~\ref{thm:No2dB}. 

To prove \eqref{eq:CUCLB}, we first note that a threshold quantizer induces the channel
\begin{IEEEeqnarray}{lCl}
P\bigl(Y =1 \bigm| X =x) 
& = & Q\left(\frac{\Upsilon-x}{\sigma}\right), \quad x\in\Reals\label{eq:dobi_thres_ch}
\end{IEEEeqnarray}
and $P\bigl(Y=0\bigm|X=x\bigr)=1-P\bigl(Y=1\bigm|
X=x)$. By~\eqref{eq:CUEKL}, we thus obtain
\begin{IEEEeqnarray}{lCl}
\dot{C}(0) & \geq & \sup_{\xi\neq 0, \Upsilon\in\Reals} \left\{\frac{Q\left(\frac{\Upsilon-\xi}{\sigma}\right)\log\frac{Q\left(\frac{\Upsilon-\xi}{\sigma}\right)}{Q\left(\frac{\Upsilon}{\sigma}\right)}}{\xi^2}\right.\nonumber\\
& & \qquad\qquad\,\, {} +\left.\frac{\left[1-Q\left(\frac{\Upsilon-\xi}{\sigma}\right)\right]\log\frac{1-Q\left(\frac{\Upsilon-\xi}{\sigma}\right)}{1-Q\left(\frac{\Upsilon}{\sigma}\right)}}{\xi^2}\right\}\nonumber\\
& = & \sup_{\xi\neq 0, \Upsilon\in\Reals} \left\{\frac{Q\left(\frac{\Upsilon-\xi}{\sigma}\right)\log\frac{1}{Q\left(\frac{\Upsilon}{\sigma}\right)}}{\xi^2}\right.\nonumber\\
& & \qquad\qquad\,\, {} + \frac{\left[1-Q\left(\frac{\Upsilon-\xi}{\sigma}\right)\right]\log\frac{1}{1-Q\left(\frac{\Upsilon}{\sigma}\right)}}{\xi^2} \nonumber\\
& & \qquad\qquad\qquad\qquad\qquad {} - \left.\vphantom{\frac{Q\left(\frac{\Upsilon-\xi}{\sigma}\right)\log\frac{Q\left(\frac{\Upsilon-\xi}{\sigma}\right)}{Q\left(\frac{\Upsilon}{\sigma}\right)}}{\xi^2}} \frac{H_b\left(Q\left(\frac{\Upsilon-\xi}{\sigma}\right)\right)}{\xi^2}\right\}.\IEEEeqnarraynumspace\label{eq:1}
\end{IEEEeqnarray}
We now change variables by defining $\mu \triangleq \xi - \Upsilon$
and by replacing the supremum over $(\xi, \Upsilon)$ with the supremum
over $(\xi, \mu)$. This latter supremum we lower-bound by taking $\xi$
to infinity while holding $\mu$ fixed.
This yields for the last two terms on the RHS of \eqref{eq:1}
\begin{equation}
\label{eq:2}
\lim_{\xi\to\infty} \frac{H_b\left(Q\left(-\frac{\mu}{\sigma}\right)\right)}{\xi^2} = 0
\end{equation}
and
\begin{equation}
\label{eq:3}
\lim_{\xi\to\infty} \frac{\left[1-Q\left(-\frac{\mu}{\sigma}\right)\right]\log\frac{1}{1-Q\left(\frac{\xi-\mu}{\sigma}\right)}}{\xi^2} = 0.
\end{equation}
We use the upper bound on the $Q$-function \eqref{eq:QUB} to lower-bound the first term on the RHS of \eqref{eq:1} as
\begin{IEEEeqnarray}{lCl}
\IEEEeqnarraymulticol{3}{l}{\lim_{\xi\to\infty} \frac{Q\left(-\frac{\mu}{\sigma}\right)\log\frac{1}{Q\left(\frac{\xi-\mu}{\sigma}\right)}}{\xi^2}}\nonumber\\
\quad & \geq & Q\left(-\frac{\mu}{\sigma}\right) \lim_{\xi\to\infty} \frac{\frac{1}{2}\log(2\pi)+\log\frac{\xi-\mu}{\sigma}+\frac{(\xi-\mu)^2}{2\sigma^2}}{\xi^2}\nonumber\\
& = & Q\left(-\frac{\mu}{\sigma}\right)\frac{1}{2\sigma^2}.\label{eq:4}
\end{IEEEeqnarray}
Combining \eqref{eq:2}--\eqref{eq:4} with \eqref{eq:1} yields
\begin{equation}
\dot{C}(0) \geq Q\left(-\frac{\mu}{\sigma}\right)\frac{1}{2\sigma^2}
\end{equation}
from which we obtain \eqref{eq:CUCLB} by letting $\mu$ tend to infinity. This proves Theorem~\ref{thm:No2dB}.

Note that \eqref{eq:CUEKL} is achieved by binary on-off keying \cite{verdu90}. By showing that \eqref{eq:CUEKL} is lower-bounded by $1/(2\sigma^2)$ as we take $\xi$ to infinity, we thus implicitly show that $\dot{C}(0)$ is achieved by binary on-off keying where the nonzero mass point tends to infinity as $\const{P}$ tends to zero.

\subsection{Proof of Theorem~\ref{thm:flash}}
\label{sub:proof_thm4}
We first argue that in order  to prove Theorem~\ref{thm:flash} it
suffices to show that for every fixed $\nu > 0$ 
\begin{equation}
\label{eq:proof_thm4_suplim}
\sup_{\xi^2\leq\nu,\set{D}}\frac{D\bigl(P_{Y|X=\xi}\bigm\| P_{Y|X=0}\bigr)}{\xi^2} < \frac{1}{2\sigma^2}.
\end{equation}
Suppose then that this strict inequality holds for every \mbox{$\nu >
0$}. Consider a family of quantizers and input distributions
parametrized by $\const{P}$ with $\E{X^2} \leq \const{P}$.
By~\cite[Eq.~(15)]{verdu90}, it follows that for every $\nu>0$
\begin{IEEEeqnarray}{lCl}
\IEEEeqnarraymulticol{3}{l}{\frac{I(X;Y)}{\const{P}}} \nonumber\\
\,\,\, & \leq & \int \frac{D\bigl(P_{Y|X=x}\bigm\| P_{Y|X=0}\bigr)}{x^2} \frac{x^2}{\const{P}} \d P_X(x) \nonumber\\
& =  & \int_{x^2\leq\nu} \frac{D\bigl(P_{Y|X=x}\bigm\|
  P_{Y|X=0}\bigr)}{x^2} \frac{x^2}{\const{P}} \d P_X(x) \nonumber\\
& & {}  + \int_{x^2>\nu} \frac{D\bigl(P_{Y|X=x}\bigm\| P_{Y|X=0}\bigr)}{x^2} \frac{x^2}{\const{P}} \d P_X(x)\nonumber\\ 
& \leq & \sup_{\xi^2\leq\nu,\set{D}}
\biggl\{ \frac{D\bigl(P_{Y|X=\xi}\bigm\| P_{Y|X=0}\bigr)}{\xi^2} \biggr\} \frac{\E{X^2
    \I{X^2\leq\nu}}}{\const{P}} \nonumber\\
& & {} + \sup_{\xi^2>\nu,\set{D}} \biggl\{\frac{D\bigl(P_{Y|X=\xi}\bigm\| P_{Y|X=0}\bigr)}{\xi^2} \biggr\} \frac{\E{X^2 \I{X^2>\nu}}}{\const{P}} \nonumber\\
& = & \sup_{\xi^2\leq\nu,\set{D}}
\biggl\{ \frac{D\bigl(P_{Y|X=\xi}\bigm\| P_{Y|X=0}\bigr)}{\xi^2} \biggr\} \frac{\E{X^2
    \I{X^2\leq\nu}}}{\const{P}} \nonumber\\
& & {} + \frac{1}{2 \sigma^{2}} \frac{\E{X^2 \I{X^2>\nu}}}{\const{P}}   \label{eq:proof_thm4_1_dobi}
\end{IEEEeqnarray}
where the last step follows because the capacity per unit-energy can be achieved by binary on-off keying where the nonzero
mass point tends to infinity (see Section~\ref{sub:proof_thm3}), so
\begin{equation}
\sup_{\xi^2>\nu,\set{D}} \frac{D\bigl(P_{Y|X=\xi}\bigm\| P_{Y|X=0}\bigr)}{\xi^2} = \frac{1}{2\sigma^2}.
\end{equation}
Taking the limit as $\const{P}$ tends to zero on both sides of \eqref{eq:proof_thm4_1_dobi}  yields
\begin{IEEEeqnarray}{lCl}
\IEEEeqnarraymulticol{3}{l}{\varliminf_{\const{P}\downarrow 0} \frac{I(X;Y)}{\const{P}}} \nonumber\\
 & \leq & \varliminf_{\const{P}\downarrow 0} \Biggl(\frac{1}{2 \sigma^{2}} \frac{\E{X^2 \I{X^2>\nu}}}{\const{P}} \nonumber\\
& & {} + \sup_{\xi^2\leq\nu,\set{D}}
\biggl\{ \frac{D(P_{Y|X=\xi} \| P_{Y|X=0})}{\xi^2} \biggr\} \frac{\E{X^2
    \I{X^2\leq\nu}}}{\const{P}}\Biggr) \nonumber\\
    & \leq & \frac{1}{2\sigma^2} \label{eq:proof_thm4_1}
\end{IEEEeqnarray}
where $\varliminf$ denotes the \emph{limit inferior}. Here the last step follows from~\eqref{eq:proof_thm4_suplim} and
from the average-power constraint
\begin{equation}
  \label{eq:amos_sigh10}
  \frac{\E{X^2 \I{X^2>\nu}}}{\const{P}} + \frac{\E{X^2
    \I{X^2\leq\nu}}}{\const{P}} \leq 1.
\end{equation}
Since the inequality in~\eqref{eq:proof_thm4_suplim} is strict for
every $\nu > 0$, it follows from~\eqref{eq:amos_sigh10} that the last
line in~\eqref{eq:proof_thm4_1} can hold with equality only if for every $\nu > 0$
\begin{equation}
  \lim_{\const{P}\downarrow 0} \frac{\E{X^2 \I{X^2>\nu}}}{\const{P}} =1. \label{eq:proof_thm4_flash}
\end{equation}
Thus, if \eqref{eq:proof_thm4_suplim} holds, then every family of distributions of $X$ satisfying
$\E{X^2}\leq\const{P}$ that achieves
\begin{equation}
  \lim_{\const{P}\downarrow 0} \frac{I(X;Y)}{\const{P}} = \frac{1}{2\sigma^2}
\end{equation}
must be flash signaling, thus proving Theorem~\ref{thm:flash}.

Having established that in order to prove Theorem~\ref{thm:flash} it
suffices to show that \eqref{eq:proof_thm4_suplim} holds for
every $\nu > 0$, we now proceed to do so.
We first note that, for every $\xi\neq 0$, the supremum
in~\eqref{eq:proof_thm4_suplim} over all quantizers $\set{D}$ can be
replaced with the supremum over all threshold quantizers. Indeed, let
\begin{IEEEeqnarray}{lCl}
\set{W} & \triangleq & \Bigl\{(\omega_1,\omega_2)\in[0,1]^2\colon \nonumber\\
& & \quad {} \omega_{1}=\Prob\bigl(\tilde{Y}\in\set{D}\bigm|X=\xi\bigr), \nonumber\\
& & \quad {} \omega_{2}=\Prob\bigl(\tilde{Y}\in\set{D}\bigm|X=0\bigr),\,\set{D}\subset\Reals\Bigr\}
\end{IEEEeqnarray}
denote the set of possible conditional probability distributions
$\bigl(P_{Y|X=\xi},P_{Y|X=0}\bigr)$ that different quantizers can
induce. Applying the methods of Section~\ref{sec:KreinMilman}, it can
be shown that the extreme points of $\overline{\set{W}}$ correspond to
threshold quantizers. (Recall that $\overline{\set{W}}$ denotes the closure of the convex hull of $\set{W}$.) Indeed, for binary inputs, the support
function $f(\cdot)$ is given by \eqref{eq:supp_int} with
$\lambda_3=0$, $\xi_1=\xi$, and $\xi_2=0$. The quantization region $\set{D}^{\star}(\bflambda)$ that achieves the
supremum in \eqref{eq:supp_int} consists of the set of
$\tilde{y}\in\Reals$ for which $g_{\bflambda}(\tilde{y})$ in
\eqref{eq:glambda_lambda30} is nonnegative. Since
$g_{\bflambda}(\cdot)$ has at most one zero, it follows that $\set{D}^{\star}(\bflambda)$
consists of at most two regions, i.e., it is a threshold quantizer. Using that the relative entropy on the LHS of \eqref{eq:proof_thm4_suplim} is convex in
$\bigl(P_{Y|X=\xi},P_{Y|X=0}\bigr)$ \cite[Th.~2.7.2]{coverthomas91},
it follows by the same arguments as in Section~\ref{sec:KreinMilman}
that, for every $\xi\neq 0$, \mbox{$D\bigl(P_{Y|X=\xi}\bigm\| P_{Y|X=0}\bigr)$} is maximized by some
threshold quantizer.

We next note that we can assume, without loss of optimality, that the
threshold $\Upsilon$ of the quantizer is nonnegative. Consequently, the supremum
over $\set{D}$ on the LHS of \eqref{eq:proof_thm4_suplim} can be replaced by a supremum over threshold quantizers
of nonnegative thresholds $\Upsilon \geq 0$. Indeed, for $x\in\Reals$,
\begin{IEEEeqnarray}{lCl}
\Prob\bigl(\tilde{Y}\geq\Upsilon\bigm|X=x\bigr)& = & 1-\Prob\bigl(\tilde{Y}\geq -\Upsilon\bigm|X=-x\bigr) \IEEEeqnarraynumspace \end{IEEEeqnarray}
and consequently,
\begin{IEEEeqnarray}{lCl}
\IEEEeqnarraymulticol{3}{l}{\left.D\bigl(P_{Y|X=\xi}\bigm\| P_{Y|X=0}\bigr)\right|_{\set{D}=\{\tilde{y}\in\Reals\colon \tilde{y}\geq\Upsilon\}}}\nonumber\\
\qquad & = & \left.D\bigl(P_{Y|X=-\xi}\bigm\| P_{Y|X=0}\bigr)\right|_{\set{D}=\{\tilde{y}\in\Reals\colon \tilde{y}\geq-\Upsilon\}}.\IEEEeqnarraynumspace
\end{IEEEeqnarray}
Thus, to every pair $(\xi,\Upsilon)$ corresponds
another pair $(-\xi,-\Upsilon)$ achieving the same relative entropy.  Since $\xi$ and $-\xi$ have the same magnitude, this implies that both pairs give rise to the same value for
\begin{equation*}
\frac{D\bigl(P_{Y|X=\xi}\bigm\| P_{Y|X=0}\bigr)}{\xi^2}
\end{equation*}
hence we can assume without loss of generality that
$\Upsilon\geq 0$.

We continue by defining the random variable $U$ as
\begin{equation}
U \triangleq \tilde{Y} \I{\tilde{Y}\geq 0}.
\end{equation}
Note that, for $\Upsilon\geq 0$, the quantizer's output can be expressed as $Y=\I{U\geq\Upsilon}$. It thus follows from the Data Processing Inequality for Relative Entropy \cite[Sec.~2.9]{coverthomas91} that
\begin{IEEEeqnarray}{lCl}
\IEEEeqnarraymulticol{3}{l}{D\bigl(P_{Y|X=\xi}\bigm\| P_{Y|X=0}\bigr)} \nonumber\\
\quad & \leq & D\bigl(P_{U|X=\xi}\bigm\| P_{U|X=0}\bigr) \nonumber\\
& = & \frac{1}{\sqrt{2\pi\sigma^2}}\int^{\infty}_{0} e^{-\frac{(\tilde{y}-\xi)^2}{2\sigma^2}} \log\frac{e^{-\frac{(\tilde{y}-\xi)^2}{2\sigma^2}}}{e^{-\frac{\tilde{y}^2}{2\sigma^2}}}\d\tilde{y} \nonumber\\
& & {} + \frac{1}{\sqrt{2\pi\sigma^2}} \left(\displaystyle \int^{0}_{-\infty} e^{-\frac{(\tilde{y}-\xi)^2}{2\sigma^2}}\d\tilde{y}\right) \log\frac{\displaystyle \int^{0}_{-\infty} e^{-\frac{(\tilde{y}-\xi)^2}{2\sigma^2}}\d\tilde{y}}{\displaystyle \int^{0}_{-\infty} e^{-\frac{\tilde{y}^2}{2\sigma^2}}\d\tilde{y}} \nonumber\\
& \triangleq & \Psi(\xi) \label{eq:proof_thm4_DUB}
\end{IEEEeqnarray}
irrespective of the threshold $\Upsilon\geq 0$. Here the last equality should be viewed as the definition of $\Psi(\xi)$. By applying the Log-Sum Inequality \cite[Th.~2.7.1]{coverthomas91} to $\Psi(\xi)$, we obtain
\begin{IEEEeqnarray}{lCl}
\Psi(\xi) & \leq &
\frac{1}{\sqrt{2\pi\sigma^2}}\int^{\infty}_{-\infty}
e^{-\frac{(\tilde{y}-\xi)^2}{2\sigma^2}}
\log\frac{e^{-\frac{(\tilde{y}-\xi)^2}{2\sigma^2}}}{e^{-\frac{\tilde{y}^2}{2\sigma^2}}}\d\tilde{y} \nonumber\\
& = & \frac{\xi^2}{2\sigma^2} \label{eq:proof_thm4_3}
\end{IEEEeqnarray}
with equality if, and only if,
\begin{equation}
\label{eq:proof_thm4_condlogsum}
\frac{e^{-\frac{(\tilde{y}-\xi)^2}{2\sigma^2}}}{e^{-\frac{\tilde{y}^2}{2\sigma^2}}}
= 2\, Q\biggl(\frac{\xi}{\sigma}\biggr), \quad \textnormal{for almost
  every $\tilde{y} \leq 0$}.
\end{equation}
Since \eqref{eq:proof_thm4_condlogsum} holds only for $\xi=0$, this yields 
\begin{equation}
\label{eq:proof_thm4_strict}
\Psi(\xi) < \frac{\xi^{2}}{2\sigma^2}, \quad \xi\neq 0.
\end{equation}
Note that~\eqref{eq:proof_thm4_strict} and \eqref{eq:proof_thm4_3}
give an upper bound on the relative entropy that does not depend on
the threshold. By combining \eqref{eq:proof_thm4_DUB} and
\eqref{eq:proof_thm4_strict}, and recalling that for every $\xi\neq 0$
the relative entropy in \eqref{eq:proof_thm4_suplim} is maximized by
some threshold quantizer, we obtain
\begin{IEEEeqnarray}{lCl}
\sup_{\set{D}} \frac{D\bigl(P_{Y|X=\xi}\bigm\|
  P_{Y|X=0}\bigr)}{\xi^2} & \leq & \frac{\Psi(\xi)}{\xi^2} <
\frac{1}{2\sigma^2}, \quad \xi\neq 0. \IEEEeqnarraynumspace
\label{eq:proof_thm4_finite}
\end{IEEEeqnarray}
Since the function $\xi \mapsto \xi^{-2} \Psi(\xi)$ is continuous on $\Reals
\setminus \{0\}$ and, as shown in Appendix~\ref{app:proof_thm4}, satisfies
\begin{equation}
  \lim_{\xi\to 0} \frac{\Psi(\xi)}{\xi^2} = 
  \frac{1}{2\sigma^2}\biggl(\frac{1}{2}+\frac{1}{\pi}\biggr) <
  \frac{1}{2\sigma^2}
  \label{eq:amos_Xmas}
\end{equation}
we obtain \eqref{eq:proof_thm4_suplim} by maximizing \eqref{eq:proof_thm4_finite} over $\xi^2\leq\nu$. This proves Theorem~\ref{thm:flash}.

\subsection{Proof of Corollary~\ref{cor:unboundthres}}
\label{sub:proof_cor5}
To prove Corollary~\ref{cor:unboundthres} we need to show that for every $\nu>0$ and every threshold quantizer with threshold $0\leq\Upsilon\leq\nu$,
\begin{equation}
\sup_{\xi\neq 0, 0\leq\Upsilon\leq\nu}\frac{D\bigl(P_{Y|X=\xi}\bigm\| P_{Y|X=0}\bigr)}{\xi^2} < \frac{1}{2\sigma^2}.\end{equation}
By~\eqref{eq:proof_thm4_finite} we have that for every $\xi \neq 0$
and every $\nu > 0$
\begin{equation}
\sup_{0\leq\Upsilon\leq\nu}\frac{D\bigl(P_{Y|X=\xi}\bigm\|
  P_{Y|X=0}\bigr)}{\xi^2}  \leq \frac{\Psi(\xi)}{\xi^{2}} 
< \frac{1}{2\sigma^2}
\end{equation}
where $\xi \mapsto \xi^{-2} \Psi(\xi)$ is continuous on $\Reals
\setminus \{0\}$ and satisfies~\eqref{eq:amos_Xmas}.
To conclude the proof of the corollary it thus remains to show that for every $\nu>0$
\begin{equation}
\label{eq:proof_cor5_claim}
\varlimsup_{\xi^2\to\infty} \sup_{0\leq\Upsilon\leq\nu} \frac{D\bigl(P_{Y|X=\xi}\bigm\| P_{Y|X=0}\bigr)}{\xi^2} < \frac{1}{2\sigma^2}
\end{equation}
where $\varlimsup$ denotes the \emph{limit superior}. This can be done by noting that for $0\leq\Upsilon\leq\nu$
\begin{IEEEeqnarray}{lCl}
\IEEEeqnarraymulticol{3}{l}{D\bigl(P_{Y|X=\xi}\bigm\| P_{Y|X=0}\bigr)}\nonumber\\
\quad & = & Q\biggl(\frac{\Upsilon-\xi}{\sigma}\biggr)\log\frac{1}{Q\Bigl(\frac{\Upsilon}{\sigma}\Bigr)} - H_b\Biggl(Q\biggl(\frac{\Upsilon-\xi}{\sigma}\biggr)\Biggr) \nonumber\\
& & {} +\Biggl[1-Q\biggl(\frac{\Upsilon-\xi}{\sigma}\biggr)\Biggr]\log\frac{1}{1-Q\Bigl(\frac{\Upsilon}{\sigma}\Bigr)} \nonumber\\
& \leq & \log\frac{1}{Q\Bigl(\frac{\Upsilon}{\sigma}\Bigr)} + \log\frac{1}{1-Q\Bigl(\frac{\Upsilon}{\sigma}\Bigr)} \nonumber\\
& \leq & \log\frac{1}{Q\Bigl(\frac{\nu}{\sigma}\Bigr)} + \log 2 \label{eq:proof_cor5_1}
\end{IEEEeqnarray} 
where the second step follows because $0\leq Q(x)\leq 1$, $x\in\Reals$ and $H_b(p)\geq 0$, $0\leq p\leq 1$, and where the last step follows because $x\mapsto Q(x)$ is monotonically decreasing in $x\in\Reals$ and because $0\leq\Upsilon\leq\nu$.
Computing the limiting ratio of the RHS of \eqref{eq:proof_cor5_1} to $\xi^2$ as $\xi^2$ tends to infinity yields for every $\nu>0$
\begin{equation}
\lim_{\xi^2\to\infty} \sup_{0\leq\Upsilon\leq\nu} \frac{D\bigl(P_{Y|X=\xi}\bigm\| P_{Y|X=0}\bigr)}{\xi^2} = 0
\end{equation}
thus establishing~\eqref{eq:proof_cor5_claim}. This proves Corollary~\ref{cor:unboundthres}.

\section{Proofs: Peak-Power-Limited Channels}
\label{sec:dobi_PP}

\subsection{Proof of Proposition~\ref{note:PP}}
\label{sub:proof_note2}
The peak-power-limited Gaussian channel with one-bit output
quantization is a memoryless channel with a continuous input taking
values in $\bigl[-\sqrt{\const{P}},\sqrt{\const{P}}\bigr]$ and a
binary output. It thus follows from Dubins's Theorem that, for every
quantization region $\set{D}$, the capacity-achieving
input distribution is discrete with two mass points
\cite[Sec.~II-C]{witsenhausen80}. We shall denote these two mass
points by $\xi_1$ and $\xi_2$.

We next argue that threshold quantizers are optimal. Let $\set{W}$ denote the set of all possible channel laws, i.e.,
\begin{IEEEeqnarray}{lCl}
\set{W} & \triangleq & \Bigl\{(\omega_1,\omega_2)\in[0,1]^2\colon \nonumber\\
& & \quad {} \omega_{\ell}=\Prob\bigl(\tilde{Y}\in\set{D}\bigm|X=\xi_{\ell}\bigr), \set{D}\subset\Reals\Bigr\}.
\end{IEEEeqnarray}
Applying the methods of Section~\ref{sec:KreinMilman} to binary
channel inputs, it can be shown that the extreme points of $\overline{\set{W}}$
correspond to threshold quantizers \eqref{eq:dobi_thresholdD}
or complements thereof. (For more details, see also Section~\ref{sub:proof_thm4}.)
By the same arguments as in
Section~\ref{sec:KreinMilman}, it follows that for every binary random
variable $X$, the mutual information $I(X;Y)$ is maximized by some
threshold quantizer.

The capacity of the peak-power-limited Gaussian channel with one-bit output quantization is thus given by
\begin{equation}
\label{eq:dobi_CPP_sup}
C_{\textnormal{PP}}(\const{P}) = \sup_{(\vect{p},\bfxi),\Upsilon\in\Reals} I\bigl(\vect{p},\vect{W}(\Upsilon|\bfxi)\bigr)
\end{equation}
where $(\vect{p},\bfxi)$ denotes the two-mass-points distribution with masses \[\vect{p}=(p_1,p_2)\in[0,1]^2\] and locations \[\bfxi=(\xi_1,\xi_2)\in[-\sqrt{\const{P}},\sqrt{\const{P}}]^2\] and where $\vect{W}(\Upsilon|\bfxi)$ denotes the channel law corresponding to the threshold quantizer $\eqref{eq:dobi_thresholdD}$ and to the mass points $(\xi_1,\xi_2)$:
\begin{equation}
W(\Upsilon|\xi_{\ell}) = \Prob\bigl(\tilde{Y}\geq \Upsilon\bigm| X=\xi_{\ell}\bigr), \quad \ell=1,2.
\end{equation}
Following the steps in Section~\ref{sub:dobi_sup=max}, it can be further shown that the supremum on the RHS of \eqref{eq:dobi_CPP_sup} is achieved.


In the following, we demonstrate that there is no loss in optimality in assuming that the mass points of the capacity-achieving input distribution are located at $-\sqrt{\const{P}}$ and $\sqrt{\const{P}}$. Indeed, suppose that the optimal mass points are located at
\begin{equation}
\label{eq:proof_note1_xis}
-\sqrt{\const{P}} \leq \xi_1 < \xi_2<\sqrt{\const{P}}.
\end{equation}
Then, it follows from the strict monotonicity of the $Q$-function that 
\begin{equation}
Q\biggl(\frac{\Upsilon-\xi_1}{\sigma}\biggr) < Q\biggl(\frac{\Upsilon-\xi_2}{\sigma}\biggr) < Q\biggl(\frac{\Upsilon-\sqrt{\const{P}}}{\sigma}\biggr).
\end{equation}
Since $W(\Upsilon|\xi_1)$ does not depend on $\xi_2$, this implies that for every $\Upsilon$ and $\xi_1$, the channel law $\vect{W}(\Upsilon|\bfxi)$ can be written as a convex combination of $\vect{W}(\Upsilon|\bfpsi)$ and $\vect{W}(\Upsilon|\bfzeta)$, where $\bfpsi=(\xi_1,\xi_1)$ and $\bfzeta=\bigl(\xi_1,\sqrt{\const{P}}\bigr)$. By the convexity of mutual information in the channel law, and by noting that $I\bigl(\vect{p},\vect{W}(\Upsilon|\bfpsi)\bigr)=0$, it follows that
\begin{equation}
I\bigl(\vect{p},\vect{W}(\Upsilon|\bfxi)\bigr) \leq I\bigl(\vect{p},\vect{W}(\Upsilon|\bfzeta)\bigr)
\end{equation}
for every $\Upsilon$ and $(\vect{p},\bfxi)$ satisfying \eqref{eq:proof_note1_xis}. Thus, $\xi_2=\sqrt{\const{P}}$ achieves the capacity. By repeating the same arguments for $\xi_1$, we obtain that the mass points of the capacity-achieving input distribution are located at $-\sqrt{\const{P}}$ and $\sqrt{\const{P}}$. It follows that the capacity can be expressed as
\begin{equation}
\label{eq:proof_note1_PPBAC}
C_{\textnormal{PP}}(\const{P}) = \max_{\Upsilon\in\Reals} C_{\Upsilon}(\const{P})
\end{equation}
where $C_{\Upsilon}(\const{P})$ denotes the capacity of the binary asymmetric channel with crossover probabilities
\begin{subequations}
\begin{IEEEeqnarray}{lCl}
W(0|1) & = & Q\biggl(\frac{\sqrt{\const{P}}-\Upsilon}{\sigma}\biggr) \label{eq:proof_note1_BACpq_a}\\
W(1|0) & = & Q\biggl(\frac{\sqrt{\const{P}}+\Upsilon}{\sigma}\biggr). \label{eq:proof_note1_BACpq_b}
\end{IEEEeqnarray}
\end{subequations}
For every $\Upsilon\in\Reals$, the capacity of the binary asymmetric channel can be computed as
\begin{equation}
\label{eq:proof_note1_BACC}
C_{\Upsilon}(\const{P}) = \log\Bigl(1+e^{-\theta}\Bigr) + \theta\,W(1|0) - H_b\bigl(W(1|0)\bigr)
\end{equation}
where
\begin{equation}
\theta \triangleq \frac{H_b\bigl(W(0|1)\bigr)-H_b\bigl(W(1|0)\bigr)}{1-W(0|1)-W(1|0)}.
\end{equation}
Combining \eqref{eq:proof_note1_BACC}, \eqref{eq:proof_note1_BACpq_a}, and \eqref{eq:proof_note1_BACpq_b} with \eqref{eq:proof_note1_PPBAC} yields
\begin{IEEEeqnarray}{lCl}
\label{eq:proof_note1_almost}
C_{\textnormal{PP}}(\const{P}) & = & \max_{\Upsilon\in\Reals} \Biggl\{ \log\Bigl(1+e^{-\Theta(\const{P},\Upsilon)}\Bigr) \nonumber\\
\IEEEeqnarraymulticol{3}{r}{\quad {}  + Q\Biggl(\frac{\sqrt{\const{P}}+\Upsilon}{\sigma}\Biggr) \Theta(\const{P},\Upsilon)  - H_b\Biggl(Q\Biggl(\frac{\sqrt{\const{P}}+\Upsilon}{\sigma}\Biggr)\Biggr) \Biggr\}\IEEEeqnarraynumspace}
\end{IEEEeqnarray}
where
\begin{equation}
\Theta(\const{P},\Upsilon) \triangleq \frac{H_b\Bigl(Q\Bigl(\frac{\sqrt{\const{P}}-\Upsilon}{\sigma}\Bigr)\Bigr)-H_b\Bigl(Q\Bigl(\frac{\sqrt{\const{P}}+\Upsilon}{\sigma}\Bigr)\Bigr)}{1-Q\Bigl(\frac{\sqrt{\const{P}}-\Upsilon}{\sigma}\Bigr)-Q\Bigl(\frac{\sqrt{\const{P}}+\Upsilon}{\sigma}\Bigr)}.
\end{equation}
Proposition~\ref{note:PP} follows then by noting that the RHS of
\eqref{eq:proof_note1_almost} is symmetric in $\Upsilon\in\Reals$, so the maximization in \eqref{eq:proof_note1_almost} can be restricted to $\Upsilon\geq 0$ without reducing \eqref{eq:proof_note1_almost}.

\subsection{Proof of Proposition~\ref{note:PPCUC}}
\label{sub:proof_note3}
It was shown in the previous section that the capacity is achieved with a threshold quantizer and a binary input distribution having mass points at $\sqrt{\const{P}}$ and $-\sqrt{\const{P}}$. Thus, the capacity can be expressed as
\begin{IEEEeqnarray}{lCl}
C_{\textnormal{PP}}\bigl(\const{P}\bigr) & = & \max_{\Upsilon\geq 0} \Biggl\{ H_b\Biggl(p_+ Q\biggl(\frac{\Upsilon-\const{A}}{\sigma}\biggr)+p_- Q\biggl(\frac{\Upsilon+\const{A}}{\sigma}\biggr)\Biggr) \nonumber\\
\IEEEeqnarraymulticol{3}{r}{\,\,\,\, {} - p_+ H_b\Biggl(Q\biggl(\frac{\Upsilon-\const{A}}{\sigma}\biggr)\Biggr)  - p_- H_b\Biggl(Q\biggl(\frac{\Upsilon+\const{A}}{\sigma}\biggr)\Biggr)\Biggr\} \IEEEeqnarraynumspace}\label{eq:proof_note3_C}
\end{IEEEeqnarray}
for some probabilities $0<p_+<1$ and $0<p_-<1$ satisfying $p_++p_-=1$. To simplify notation, we have introduced \mbox{$\const{A}\triangleq\sqrt{\const{P}}$} and we have made the dependence of $p_+$ and $p_-$ on $\Upsilon$ not explicit.

Expanding $H_b(\cdot)$ as a Taylor series around $Q\bigl(\Upsilon/\sigma\bigr)$, we obtain for the first term on the RHS of \eqref{eq:proof_note3_C}
\begin{IEEEeqnarray}{lCl}
\IEEEeqnarraymulticol{3}{l}{H_b\Biggl(p_+ Q\biggl(\frac{\Upsilon-\const{A}}{\sigma}\biggr)+p_- Q\biggl(\frac{\Upsilon+\const{A}}{\sigma}\biggr)\Biggr)}\nonumber\\
 \,\,\,\, & = & H_b\Biggl(Q\biggl(\frac{\Upsilon}{\sigma}\biggr)\Biggr) + \log\frac{1-Q\Bigl(\frac{\Upsilon}{\sigma}\Bigr)}{Q\Bigl(\frac{\Upsilon}{\sigma}\Bigr)}\times\nonumber\\
 & & \qquad {} \times\biggl[p_+ Q\biggl(\frac{\Upsilon-\const{A}}{\sigma}\biggr) +p_- Q\biggl(\frac{\Upsilon+\const{A}}{\sigma}\biggr)-Q\biggl(\frac{\Upsilon}{\sigma}\biggr)\biggr] \nonumber\\
& & {} - \frac{1}{2Q\bigl(\frac{\Upsilon}{\sigma}\bigr)\bigl[1-Q\bigl(\frac{\Upsilon}{\sigma}\bigr)\bigr]}\times\nonumber\\
& & \quad {} \times \Biggl[p_+ Q\biggl(\frac{\Upsilon-\const{A}}{\sigma}\biggr) + p_- Q\biggl(\frac{\Upsilon+\const{A}}{\sigma}\biggr)-Q\biggl(\frac{\Upsilon}{\sigma}\biggr)\Biggr]^2 \nonumber\\
& & {} + \mathsf{R}_H(\const{A},\Upsilon,p_+) \IEEEeqnarraynumspace\label{eq:proof_note3_Taylor1}
\end{IEEEeqnarray}
where 
\begin{IEEEeqnarray}{lCl}
\mathsf{R}_H(\const{A},\Upsilon,p_+) & \triangleq & \frac{1-2\tilde{p}}{6 \tilde{p}\,(1-\tilde{p})} \times\nonumber\\
\IEEEeqnarraymulticol{3}{r}{\quad \times\Biggl[p_+ Q\biggl(\frac{\Upsilon-\const{A}}{\sigma}\biggr) +p_- Q\biggl(\frac{\Upsilon+\const{A}}{\sigma}\Biggr)-Q\biggl(\frac{\Upsilon}{\sigma}\biggr)\Biggr]^3 \IEEEeqnarraynumspace}\label{eq:proof_note3_H}
\end{IEEEeqnarray}
for some $\tilde{p}\in\bigl[Q\bigl((\Upsilon+\const{A})/\sigma\bigr),Q\bigl((\Upsilon-\const{A})/\sigma\bigr)\bigr]$. Expanding the $Q$-function as a Taylor series around $\Upsilon/\sigma$ yields
\begin{IEEEeqnarray}{lCl}
\IEEEeqnarraymulticol{3}{l}{p_+ Q\biggl(\frac{\Upsilon-\const{A}}{\sigma}\biggr)+p_- Q\biggl(\frac{\Upsilon+\const{A}}{\sigma}\biggr)-Q\biggl(\frac{\Upsilon}{\sigma}\biggr)} \nonumber\\
\quad & = & (p_+-p_-) \frac{\const{A}}{\sigma} \frac{1}{\sqrt{2\pi}} e^{-\frac{\Upsilon^2}{2\sigma^2}} + \mathsf{R}_Q(\const{A},\Upsilon,p_+) \label{eq:proof_note3_Q}
\end{IEEEeqnarray}
where
\begin{equation}
\label{eq:proof_note3_RQ_def}
\mathsf{R}_Q(\const{A},\Upsilon,p_+) \triangleq \frac{\const{A}^2}{2\sigma^2} \frac{\tilde{x}}{\sqrt{2\pi\sigma^2}} e^{-\frac{\tilde{x}^2}{2\sigma^2}}
\end{equation}
for some $\tilde{x}\in[\Upsilon-\const{A},\Upsilon+\const{A}]$.
Note that 
\begin{equation}
\bigl|\tilde{x}\exp\bigl(-\tilde{x}^2/(2\sigma^2)\bigr)\bigr|\leq \sigma/\sqrt{e}
\end{equation}
so $\mathsf{R}_Q(\const{A},\Upsilon,p_+)$ satisfies
\begin{equation}
\label{eq:proof_note3_RQ}
\bigl|\mathsf{R}_Q(\const{A},\Upsilon,p_+)\bigr| \leq \frac{\const{A}^2}{2\sigma^2 \sqrt{2\pi e}}, \qquad 0\leq p_+ \leq 1.
\end{equation}
Combining \eqref{eq:proof_note3_Q} with \eqref{eq:proof_note3_Taylor1}, we obtain for the first term on the RHS of \eqref{eq:proof_note3_C}
\begin{IEEEeqnarray}{lCl}
\IEEEeqnarraymulticol{3}{l}{H_b\Biggl(p_+ Q\biggl(\frac{\Upsilon-\const{A}}{\sigma}\biggr)+p_- Q\biggl(\frac{\Upsilon+\const{A}}{\sigma}\biggr)\Biggr)}\nonumber\\
\quad & = & H_b\Biggl(Q\biggl(\frac{\Upsilon}{\sigma}\biggr)\Biggr)+\log\frac{1-Q\Bigl(\frac{\Upsilon}{\sigma}\Bigr)}{Q\Bigl(\frac{\Upsilon}{\sigma}\Bigr)}\times\nonumber\\
& & \qquad {} \times\biggl[p_+ Q\biggl(\frac{\Upsilon-\const{A}}{\sigma}\biggr)+p_- Q\biggl(\frac{\Upsilon+\const{A}}{\sigma}\biggr)-Q\biggl(\frac{\Upsilon}{\sigma}\biggr)\biggr] \nonumber\\
& & {} - \frac{1}{2Q\bigl(\frac{\Upsilon}{\sigma}\bigr)\bigl[1-Q\bigl(\frac{\Upsilon}{\sigma}\bigr)\bigr]}\times\nonumber\\
& & \quad {} \times\biggl[ (p_+-p_-) \frac{\const{A}}{\sigma}\frac{1}{\sqrt{2\pi}}e^{-\frac{\Upsilon^2}{2\sigma^2}} + \mathsf{R}_Q(\const{A},\Upsilon,p_+)\biggr]^2 \nonumber\\
& & {} + \mathsf{R}_H(\const{A},\Upsilon,p_+) \nonumber\\
& = & H_b\Biggl(Q\biggl(\frac{\Upsilon}{\sigma}\biggr)\Biggr)+\log\frac{1-Q\Bigl(\frac{\Upsilon}{\sigma}\Bigr)}{Q\Bigl(\frac{\Upsilon}{\sigma}\Bigr)}\times\nonumber\\
& & \qquad {} \times\biggl[p_+ Q\biggl(\frac{\Upsilon-\const{A}}{\sigma}\biggr)+p_- Q\biggl(\frac{\Upsilon+\const{A}}{\sigma}\biggr)-Q\biggl(\frac{\Upsilon}{\sigma}\biggr)\biggr] \nonumber\\
& & {} - \frac{\const{A}^2}{\sigma^2}\frac{e^{-\frac{\Upsilon^2}{\sigma^2}}}{4\pi Q\bigl(\frac{\Upsilon}{\sigma}\bigr)\bigl[1-Q\bigl(\frac{\Upsilon}{\sigma}\bigr)\bigr]}  (p_+-p_-)^2 \nonumber\\
& & {} + \mathsf{K}(\const{A},\Upsilon,p_+) + \mathsf{R}_H(\const{A},\Upsilon,p_+) \label{eq:proof_note3_first}
\end{IEEEeqnarray}
where
\begin{IEEEeqnarray}{lCl}
\mathsf{K}(\const{A},\Upsilon,p_+) & \triangleq & - \frac{2(p_+-p_-)\frac{1}{\sqrt{2\pi}}e^{-\frac{\Upsilon^2}{2\sigma^2}} \frac{\const{A}}{\sigma}\mathsf{R}_Q(\const{A},\Upsilon,p_+)}{2Q\bigl(\frac{\Upsilon}{\sigma}\bigr)\bigl[1-Q\bigl(\frac{\Upsilon}{\sigma}\bigr)\bigr]}\nonumber\\
& & {} - \frac{\bigl|\mathsf{R}_Q(\const{A},\Upsilon,p_+)\bigr|^2}{2Q\bigl(\frac{\Upsilon}{\sigma}\bigr)\bigl[1-Q\bigl(\frac{\Upsilon}{\sigma}\bigr)\bigr]}.
\end{IEEEeqnarray}

Taylor-series expansions for the last two terms on the RHS of \eqref{eq:proof_note3_C} follow directly from \eqref{eq:proof_note3_first} by setting $p_+$ to $1$ and to $0$. Thus, by applying \eqref{eq:proof_note3_first} to \eqref{eq:proof_note3_C}, and by using that $p_++p_-=1$, we obtain
\begin{IEEEeqnarray}{lCl}
C_{\textnormal{PP}}(\const{P}) & = & \max_{\Upsilon\geq 0} \left\{\frac{\const{A}^2}{\sigma^2}\frac{e^{-\frac{\Upsilon^2}{\sigma^2}}}{4\pi Q\bigl(\frac{\Upsilon}{\sigma}\bigr)\bigl[1-Q\bigl(\frac{\Upsilon}{\sigma}\bigr)\bigr]} \bigl[1 - (p_+-p_-)^2\bigr]\right.\nonumber\\
& & \quad\qquad {}+ \mathsf{K}(\const{A},\Upsilon,p_+) + \mathsf{R}_H(\const{A},\Upsilon,p_+)\nonumber\\
& & \quad\qquad {} - p_+ \bigl[\mathsf{K}(\const{A},\Upsilon,1)+\mathsf{R}_H(\const{A},\Upsilon,1)\bigr] \nonumber\\
& & \quad\qquad\! {} \left.\vphantom{\frac{e^{-\frac{\Upsilon^2}{\sigma^2}}}{4\pi Q\bigl(\frac{\Upsilon}{\sigma}\bigr)\bigl[1-Q\bigl(\frac{\Upsilon}{\sigma}\bigr)\bigr]}} - p_- \bigl[\mathsf{K}(\const{A},\Upsilon,0)+\mathsf{R}_H(\const{A},\Upsilon,0)\bigr] \right\}. \label{eq:proof_note3_sec} \IEEEeqnarraynumspace
\end{IEEEeqnarray}
As shown in Appendix~\ref{app:proof_note3}, we have
\begin{subequations}
\begin{IEEEeqnarray}{rCll}
\lim_{\const{A}\downarrow 0} \sup_{\Upsilon\geq 0} \frac{|\mathsf{R}_H(\const{A},\Upsilon,p_+)|}{\const{A}^2} & = & 0, \quad & 0\leq p_+ \leq 1\label{eq:proof_note3_claim1} \\
\lim_{\const{A} \downarrow 0} \sup_{\Upsilon\geq 0} \frac{|\mathsf{K}(\const{A},\Upsilon,p_+)|}{\const{A}^2} & = & 0, \quad & 0\leq p_+\leq 1.\IEEEeqnarraynumspace \label{eq:proof_note3_claim2}
\end{IEEEeqnarray}
\end{subequations}
Using \eqref{eq:proof_note3_claim1}, \eqref{eq:proof_note3_claim2}, and the Triangle Inequality, \eqref{eq:proof_note3_sec} can thus be upper-bounded by
\begin{IEEEeqnarray}{lCl}
C_{\textnormal{PP}}(\const{P}) & \leq & \sup_{\Upsilon\geq 0} \frac{\const{A}^2}{\sigma^2}\frac{e^{-\frac{\Upsilon^2}{\sigma^2}}\bigl[1 - (p_+-p_-)^2\bigr]}{4\pi Q\bigl(\frac{\Upsilon}{\sigma}\bigr)\bigl[1-Q\bigl(\frac{\Upsilon}{\sigma}\bigr)\bigr]} + o\bigl(\const{A}^2\bigr) \IEEEeqnarraynumspace\label{eq:proof_note3_sec_2}
\end{IEEEeqnarray}
where $\lim_{\const{A}\downarrow 0} o\bigl(\const{A}^2\bigr)/\const{A}^2=0$. Consequently, dividing \eqref{eq:proof_note3_sec_2} by $\const{P}=\const{A}^2$ and computing the limit as $\const{P}$ tends to zero, yields
\begin{IEEEeqnarray}{lCl}
\lim_{\const{P}\downarrow 0} \frac{C_{\textnormal{PP}}(\const{P})}{\const{P}}  & \leq & \sup_{\Upsilon\geq 0} \frac{1}{\sigma^2}\frac{e^{-\frac{\Upsilon^2}{\sigma^2}}\bigl[1 - (p_+-p_-)^2\bigr]}{4\pi Q\bigl(\frac{\Upsilon}{\sigma}\bigr)\bigl[1-Q\bigl(\frac{\Upsilon}{\sigma}\bigr)\bigr]} \nonumber\\
& \leq & \sup_{\Upsilon\geq 0} \frac{e^{-\frac{\Upsilon^2}{\sigma^2}}}{4\pi Q\bigl(\frac{\Upsilon}{\sigma}\bigr)\bigl[1-Q\bigl(\frac{\Upsilon}{\sigma}\bigr)\bigr]} \frac{1}{\sigma^2} \label{eq:proof_note3_CUC}
\end{IEEEeqnarray}
where the second inequality holds with equality for \mbox{$p_+=p_-=1/2$}.

\begin{figure}[t!]
\centering
\psfrag{U}[c][c]{$u$}
\psfrag{g}[b][b]{$g(u)$}
\begin{center}
\epsfig{file=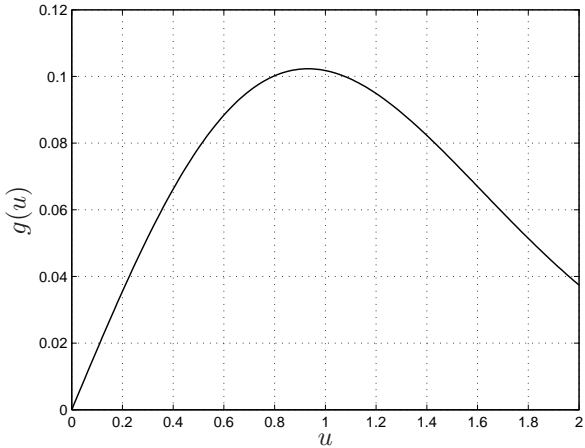,width=0.48\textwidth}
\end{center}
 \caption{The function $u \mapsto g(u)$ for $0\leq u \leq 2$.}
 \label{fig:g_plot}
\end{figure}

It remains to show that the maximum on the RHS of \eqref{eq:proof_note3_CUC} is attained for $\Upsilon=0$. To this end, we argue that the function
\begin{equation}
f(\Upsilon) \triangleq \frac{e^{-\frac{\Upsilon^2}{\sigma^2}}}{Q\bigl(\frac{\Upsilon}{\sigma}\bigr)\bigl[1-Q\bigl(\frac{\Upsilon}{\sigma}\bigr)\bigr]}, \quad \Upsilon\geq 0
\end{equation} 
is monotonically decreasing in $\Upsilon\geq 0$. Indeed, the first derivative of $f(\cdot)$ is given by
\begin{IEEEeqnarray}{lCl}
f'(\Upsilon) & = & - \frac{\frac{1}{\sigma} e^{-\frac{\Upsilon^2}{\sigma^2}}}{\bigl[Q\bigl(\frac{\Upsilon}{\sigma}\bigr)\bigr]^2\bigl[1-Q\bigl(\frac{\Upsilon}{\sigma}\bigr)\bigr]^2} \, g\biggl(\frac{\Upsilon}{\sigma}\biggr), \quad \Upsilon\geq 0 \IEEEeqnarraynumspace \label{eq:proof_note3_fprime}
\end{IEEEeqnarray}
where
\begin{equation}
g(u) \triangleq 2u Q(u)[1-Q(u)]-\frac{e^{-\frac{u^2}{2}}}{\sqrt{2\pi}}[1-2 Q(u)]
\end{equation}
for $u\geq 0$. For $u\geq 2$, we lower-bound the $Q$-function as \cite[Prop.~19.4.2]{lapidoth09}
\begin{equation}
Q(u) > \frac{3}{4}\frac{1}{\sqrt{2\pi}u} e^{-\frac{u^2}{2}}, \quad u \geq 2
\end{equation}
to obtain
\begin{IEEEeqnarray}{lCl}
g(u) & > & \frac{3}{2}\frac{e^{-\frac{u^2}{2}}}{\sqrt{2\pi}}[1-Q(u)] - \frac{e^{-\frac{u^2}{2}}}{\sqrt{2\pi}}[1-2 Q(u)] \nonumber\\
& = & \frac{e^{-\frac{u^2}{2}}}{\sqrt{8\pi}}[1+Q(u)] \nonumber\\
& > & 0.
\end{IEEEeqnarray}
For $0\leq u \leq 2$, it can be shown numerically that $g(u)\geq 0$; see Figure~\ref{fig:g_plot}. 

It thus follows that $g\bigl(\Upsilon/\sigma\bigr)\geq 0$, $\Upsilon/\sigma \geq 0$ and hence, by \eqref{eq:proof_note3_fprime}, $f'(\Upsilon)\leq 0$, $\Upsilon\geq 0$. Consequently,
\begin{equation}
\label{eq:proof_note3_fmax}
\max_{\Upsilon\geq 0} f(\Upsilon) = f(0) = 4
\end{equation}
which together with \eqref{eq:proof_note3_CUC} yields
\begin{equation}
\label{eq:doby_doby_doby}
\lim_{\const{P}\downarrow 0} \frac{C_{\textnormal{PP}}(\const{P})}{\const{P}} \leq \frac{1}{\pi\sigma^2}.
\end{equation}
Noting that the RHS of \eqref{eq:doby_doby_doby} is achieved for $p_+=p_-=1/2$ and a symmetric threshold quantizer (cf.\ \eqref{eq:CUEsym}), this proves Proposition~\ref{note:PPCUC}.

\section{Proofs: Fading Channels}
\label{sec:noncoherent_proofs}

\subsection{Proof of Theorem~\ref{prop:coherent}}
\label{sub:proof_prop7}
We will lower-bound the RHS of \eqref{eq:CUC_coherent} by restricting the supremum to radial quantizers
\begin{equation}
\label{eq:circsymthres}
\set{D} = \bigl\{\tilde{y}\in\Complex\colon |\tilde{y}|\geq\Upsilon\bigr\}, \quad \Upsilon> 0
\end{equation}
and thus demonstrate that
\begin{equation}
\label{eq:tobi_CUCFC_LB}
\dot{C}(0) \geq \frac{1}{\sigma^2}.
\end{equation}
Together with the upper bound \eqref{eq:CUC_unquant}, this will prove Theorem~\ref{prop:coherent}.

To prove \eqref{eq:tobi_CUCFC_LB}, note that, conditioned on $(H,X)=(h,x)$, the squared magnitude of $\sqrt{2/\sigma^2}\tilde{Y}$ has a noncentral chi-square distribution with $2$ degrees of freedom and noncentrality parameter $\frac{2}{\sigma^2} |h|^2 |x|^2$ \cite[p.~8]{simon02}. Consequently, a radial quantizer induces the channel \cite[Sec.~2-E]{simon02}
\begin{IEEEeqnarray}{lCl}
\IEEEeqnarraymulticol{3}{l}{\Prob\bigl(Y=1\bigm|H=h,X=x\bigr)}\nonumber\\\
\qquad\qquad & = & Q_1\Biggl(\sqrt{\frac{2}{\sigma^2}}|h| |x|, \sqrt{\frac{2}{\sigma^2}}\Upsilon\Biggr)
\end{IEEEeqnarray}
for $h\in\Complex$, $x\in\Complex$, and $\Upsilon>0$, where $Q_1(\cdot,\cdot)$ denotes the first-order Marcum $Q$-function \cite[Eq.~(2.20)]{simon02}. For $x=0$ this becomes
\begin{equation}
\Prob\bigl(Y=1\bigm|H=h,X=0\bigr) = e^{-\frac{\Upsilon^2}{\sigma^2}}
\end{equation}
for $h\in\Complex$ and $\Upsilon>0$. This yields
\begin{subequations}
\begin{IEEEeqnarray}{lCl}
\IEEEeqnarraymulticol{3}{l}{D\bigl(P_{Y|H,X=\xi}\bigm\| P_{Y|H,X=0}\bigm| P_H\bigr)}\nonumber\\
\,\, & = & \E{Q_1\Biggl(\sqrt{\frac{2}{\sigma^2}}|H||\xi|, \sqrt{\frac{2}{\sigma^2}}\Upsilon\Biggr)\log\frac{1}{e^{-\frac{\Upsilon^2}{\sigma^2}}}} \nonumber\\
& & {} + \E{\Biggl\{1-Q_1\Biggl(\sqrt{\frac{2}{\sigma^2}}|H||\xi|, \sqrt{\frac{2}{\sigma^2}}\Upsilon\Biggr)\Biggr\}\log\frac{1}{1-e^{-\frac{\Upsilon^2}{\sigma^2}}}} \nonumber\\
& &  {} - \E{H_b\Biggl(Q_1\Biggl(\sqrt{\frac{2}{\sigma^2}}|H||\xi|,\sqrt{\frac{2}{\sigma^2}}\Upsilon\Biggr)\Biggr)} \label{eq:proof_prop7_3a}\\
& \geq & \E{Q_1\Biggl(\sqrt{\frac{2}{\sigma^2}}|H||\xi|, \sqrt{\frac{2}{\sigma^2}}\Upsilon\Biggr) \frac{\Upsilon^2}{\sigma^2}} - \log 2 \label{eq:proof_prop7_3b}
\end{IEEEeqnarray}
\end{subequations}
where \eqref{eq:proof_prop7_3b} follows because the second term in
\eqref{eq:proof_prop7_3a} is nonnegative, and because the binary
entropy function is upper-bounded by $\log 2$.

By applying \eqref{eq:proof_prop7_3b} to \eqref{eq:CUC_coherent}, we obtain
\begin{IEEEeqnarray}{lCl}
\dot{C}(0) & \geq & \sup_{\substack{\xi\neq 0,\\\Upsilon>0}} \Biggl\{ \E{Q_1\Biggl(\sqrt{\frac{2}{\sigma^2}}|H||\xi|, \sqrt{\frac{2}{\sigma^2}}\Upsilon\Biggr) \frac{\Upsilon^2}{|\xi|^2\sigma^2}} \nonumber\\
& & \quad\qquad\qquad\qquad\qquad\qquad\qquad {} - \frac{1}{|\xi|^2}\log 2\Biggr\}.\IEEEeqnarraynumspace\label{eq:proof_prop7_3}
\end{IEEEeqnarray}
We lower-bound the supremum on the RHS of \eqref{eq:proof_prop7_3} by choosing $\Upsilon = \mu |h| |\xi|$ for some fixed $0<\mu<1$ and by taking $|\xi|$ to infinity. We then lower-bound the first-order Marcum \mbox{$Q$-function} using \cite[Sec.~C-2, Eq.~(C.24)]{simon02}
\begin{IEEEeqnarray}{lCl}
\IEEEeqnarraymulticol{3}{l}{Q_1(\alpha,\beta)} \nonumber\\
\quad & \geq & 1- \frac{1}{2}\Biggl[\exp\biggl(-\frac{(\alpha-\beta)^2}{2}\biggr)-\exp\biggl(-\frac{(\alpha+\beta)^2}{2}\biggr)\Biggr]\IEEEeqnarraynumspace
\end{IEEEeqnarray}
for $\alpha>\beta\geq 0$. This yields
\begin{IEEEeqnarray}{lCl}
\IEEEeqnarraymulticol{3}{l}{\dot{C}(0)}\nonumber\\
\,\, & \geq & \frac{\mu^2\E{|H|^2}}{\sigma^2} \nonumber\\
 & & {}  - \lim_{|\xi|\to\infty} \frac{1}{2|\xi|^2}\E{\exp\biggl(-\frac{|H|^2|\xi|^2}{\sigma^2}(1-\mu)^2\biggr)\frac{\mu^2|H|^2|\xi|^2}{\sigma^2}} \nonumber\\
& & {} +\lim_{|\xi|\to\infty} \frac{1}{2|\xi|^2}\E{\exp\biggl(-\frac{|H|^2|\xi|^2}{\sigma^2}(1+\mu)^2\biggr)\frac{\mu^2|H|^2|\xi|^2}{\sigma^2}} \label{eq:proof_prop7_before5}\nonumber\\
& \geq & \frac{\mu^2\E{|H|^2}}{\sigma^2} - \lim_{|\xi|\to\infty} \frac{\mu^2}{2|\xi|^2 e\,(1-\mu)^2} \nonumber\\
& = &  \frac{\mu^2\E{|H|^2}}{\sigma^2} \label{eq:proof_prop7_5}
\end{IEEEeqnarray}
where the second step follows because $0\leq xe^{-\alpha x}\leq 1/(e\alpha)$ for every $x\geq 0$ and $\alpha>0$. This establishes \eqref{eq:tobi_CUCFC_LB} because $H$ is of unit variance and $\mu$ can be arbitrarily close to $1$.

\subsection{Proof of Theorem~\ref{prop:noncoherent}}
\label{sub:proof_prop8}
By the Data Processing Inequality for Relative Entropy, the relative entropy on the RHS of \eqref{eq:CUC_noncoherent} is
upper-bounded by the relative entropy corresponding to the unquantized
channel, i.e., \cite[Eq.~(64)]{verdu02}
\begin{equation}
\frac{D\bigl(P_{Y|X=\xi}\bigm\| P_{Y|X=0}\bigr)}{|\xi|^2} \leq \frac{1}{\sigma^2} - \frac{\log\Bigl(1+\frac{|\xi|^2}{\sigma^2}\Bigr)}{|\xi|^2}. \label{eq:proof_prop8_1}
\end{equation}
Consequently, the capacity per unit-energy \eqref{eq:CUC_noncoherent} is
strictly smaller than $1/\sigma^2$ unless the supremum on the RHS of
\eqref{eq:CUC_noncoherent} is approached as $|\xi|$ tends to
infinity. It thus remains to show that
\begin{equation}
\label{eq:proof_prop8_claim}
\varlimsup_{|\xi|\to\infty} \sup_{\set{D}} \frac{D\bigl(P_{Y|X=\xi}\bigm\| P_{Y|X=0}\bigr)}{|\xi|^2} < \frac{1}{\sigma^2}.
\end{equation}
To this end, we first note that, for every $\xi\neq 0$, the supremum
in \eqref{eq:proof_prop8_claim} over all quantizers $\set{D}$ can be
replaced with the supremum over all \emph{radial} quantizers
\eqref{eq:circsymthres}. Indeed, for every quantization region
satisfying
\begin{equation*}
\Prob\bigl(Y=1\bigm|X=\xi\bigr) = \beta, \quad 0<\beta< 1
\end{equation*}
the relative entropy
\begin{IEEEeqnarray}{lCl}
\IEEEeqnarraymulticol{3}{l}{D\bigl(P_{Y|X=\xi}\bigm\| P_{Y|X=0}\bigr)}\nonumber\\
\quad & = & \beta\log\frac{1}{\Prob\bigl(Y=1\bigm| X=0\bigr)} \nonumber\\
& & {}  + (1-\beta)\log\frac{1}{1-\Prob\bigl(Y=1|X=0\bigr)} - H_b(\beta) \IEEEeqnarraynumspace\label{eq:proof_prop8_D_beta}
\end{IEEEeqnarray}
is a convex function of $\Prob\bigl(Y=1\bigm|X=0\bigr)$. Thus, for every $0<\beta<1$, the RHS of \eqref{eq:proof_prop8_D_beta} is maximized for the quantization region that minimizes (or maximizes) $\Prob\bigl(Y=1\bigm| X=0\bigr)$ while holding $\Prob\bigl(Y=1\bigm|X=\xi\bigr)=\beta$ fixed. By the Neyman-Pearson Lemma \cite{neymanpearson32}, such a quantization region has the form
\begin{equation}
\label{eq:proof_prop8_NP_D_1}
\set{D}^{\star} = \biggl\{\tilde{y}\in\Complex\colon \frac{f(\tilde{y}|0)}{f(\tilde{y}|\xi)}\leq \Lambda \biggr\}, \quad \Lambda>0
\end{equation}
(or the complement thereof), where $f(\tilde{y}|x)$ denotes the conditional density of $\tilde{Y}$, conditioned on $X=x$, and where $\Lambda$ is such that $\Prob\bigl(\tilde{Y}\in\set{D}^{\star}\bigm|X=\xi\bigr)=\beta$.
(Note that for every $0<\beta< 1$ there exists such a $\Lambda$ since, for the channel model \eqref{eq:complex_channel}, $\Prob\bigl(\tilde{Y}\in\set{D}^{\star}\bigm|X=\xi\bigr)$ is a continuous, strictly increasing function of $\Lambda>0$.) 
The likelihood ratio on the RHS of \eqref{eq:proof_prop8_NP_D_1} is given by
\begin{equation}
\frac{f(\tilde{y}|0)}{f(\tilde{y}|\xi)} = \biggl(1+\frac{|\xi|^2}{\sigma^2}\biggr) e^{-\frac{|\tilde{y}|^2}{\sigma^2}\frac{|\xi|^2}{\sigma^2+|\xi|^2}}, \quad \tilde{y}\in\Complex
\end{equation}
so \eqref{eq:proof_prop8_NP_D_1} is a radial quantizer with threshold
\begin{equation}
\Upsilon = \sigma\sqrt{\biggl(1+\frac{\sigma^2}{|\xi|^2}\biggr)\log\left(\frac{1+\frac{|\xi|^2}{\sigma^2}}{\Lambda}\right)}.
\end{equation}
Thus, for every $0<\beta<1$, the RHS of \eqref{eq:proof_prop8_D_beta} is maximized by a radial quantizer whose threshold is a function of $\beta$. This implies that, for every nonzero $\xi$, the relative entropy
$D(P_{Y|X=\xi} \| P_{Y|X=0})$ is maximized by a radial
quantizer. Such a quantizer induces the channel
\begin{equation}
\Prob\bigl(Y=1\bigm| X=x\bigr) = \exp\biggl(-\frac{\Upsilon^2}{|x|^2+\sigma^2}\biggr)
\end{equation}
for $x\in\Complex$ and $\Upsilon> 0$. Consequently,
\begin{IEEEeqnarray}{lCl}
\IEEEeqnarraymulticol{3}{l}{D\bigl(P_{Y|X=\xi}\bigm\| P_{Y|X=0}\bigr)} \nonumber\\
\quad & = & e^{-\frac{\Upsilon^2}{|\xi|^2+\sigma^2}} \log\frac{1}{e^{-\frac{\Upsilon^2}{\sigma^2}}} \nonumber\\
& & {} + \biggl[1-e^{-\frac{\Upsilon^2}{|\xi|^2+\sigma^2}}\biggr] \log\frac{1}{1-e^{-\frac{\Upsilon^2}{\sigma^2}}} - H_b\biggl(e^{-\frac{\Upsilon^2}{|\xi|^2+\sigma^2}}\biggr)\nonumber\\
& \leq &  \frac{\Upsilon^2}{\sigma^2} e^{-\frac{\Upsilon^2}{|\xi|^2+\sigma^2}} - \biggl[1-e^{-\frac{\Upsilon^2}{\sigma^2}}\biggr] \log\biggl(1-e^{-\frac{\Upsilon^2}{\sigma^2}}\biggr) \nonumber\\
& \leq & \frac{\Upsilon^2}{\sigma^2} e^{-\frac{\Upsilon^2}{|\xi|^2+\sigma^2}} + \frac{1}{e} \label{eq:proof_prop8_2}
\end{IEEEeqnarray}
where the second step follows because $H_b(\cdot)\geq 0$ and $\exp\bigl(-\Upsilon^2/(|\xi|^2+\sigma^2)\bigr) \geq \exp\bigl(-\Upsilon^2/\sigma^2\bigr)$; and the third step follows because $-x\log x\leq\frac{1}{e}$, $0<x<1$.

The first term on the RHS of \eqref{eq:proof_prop8_2} is maximized for \mbox{$\Upsilon^2=|\xi|^2+\sigma^2$}, which yields
\begin{equation}
\frac{\Upsilon^2}{\sigma^2} e^{-\frac{\Upsilon^2}{|\xi|^2+\sigma^2}} \leq \frac{|\xi|^2}{e\,\sigma^2}+\frac{1}{e}, \quad \Upsilon>0.
\end{equation}
The RHS of \eqref{eq:proof_prop8_2} is thus upper-bounded by
\begin{equation}
\label{eq:proof_prop8_3}
D\bigl(P_{Y|X=\xi}\bigm\| P_{Y|X=0}\bigr) \leq \frac{|\xi|^2}{e\, \sigma^2} + \frac{2}{e}.
\end{equation}
Dividing both sides of \eqref{eq:proof_prop8_3} by $|\xi|^2$, and computing the limit as $|\xi|$ tends to infinity, yields
\begin{equation}
\varlimsup_{|\xi|\to\infty}\sup_{\set{D}} \frac{D\bigl(P_{Y|X=\xi}\bigm\| P_{Y|X=0}\bigr)}{|\xi|^2} \leq \frac{1}{e\,\sigma^2} < \frac{1}{\sigma^2}.
\end{equation}
This proves Theorem~\ref{prop:noncoherent}.

\section{Summary and Conclusion}
\label{sec:conclusion}
It is well-known that quantizing the output of the discrete-time,
average-power-limited, Gaussian channel using a symmetric threshold
quantizer reduces the capacity per unit-energy by a factor of $2/\pi$,
a loss which translates to a power loss of approximately 2dB. We have shown
that this loss can be avoided by using asymmetric threshold quantizers
with corresponding asymmetric signal constellations. Moreover, the capacity per unit-energy can be achieved by a PPM scheme. For this scheme, the error probability can be analyzed
directly using the Union Bound and the standard upper bound  on
the $Q$-function~\eqref{eq:QUB}. There is no need to resort to conventional methods used
to prove coding theorems such as the method of types,
information-spectrum methods, or random coding exponents.

The above results demonstrate that the 2dB power loss incurred on the
Gaussian channel with symmetric one-bit output quantization is not due
to the hard decisions but due to the suboptimal quantizer. In
fact, if we employ an asymmetric threshold quantizer, and if we use
asymmetric signal constellations, then hard-decision decoding achieves
the capacity per unit-energy of the Gaussian channel.

The above results also demonstrate that a threshold quantizer is asymptotically optimal as the SNR tends to zero. This is not only true asymptotically: for every fixed SNR, we have shown that, among all one-bit quantizers, a threshold quantizer is optimal.

We have also shown that the capacity per unit-energy can only be
achieved by flash-signaling input distributions. Since such
signaling leads to poor spectral efficiencies, a significant loss in spectral efficiency is unavoidable. Thus, while one-bit output quantization does not reduce the capacity per unit-energy, it does reduce the spectral efficiency.

For Rayleigh-fading channels, we have shown that, in the coherent case, a one-bit quantizer does not reduce the capacity per unit-energy, provided that we allow the quantizer to depend on the fading level. This is no longer true in the noncoherent
case: here all one-bit output quantizers reduce the capacity per
unit-energy.

\appendices

\section{}
\label{app:amos}
\begin{lemma}
\label{lem:amos_cont}
Let $\set{D}$ be a Borel subset of the reals, and let the sequence of real
numbers $\{x_{k}\}$ converge to $\xi$. Let $Z$ be a zero-mean
Gaussian random variable of positive variance $\sigma^{2}$. Then
\begin{equation}
  \label{eq:amos_cont10}
  \lim_{k \to \infty} \Prob\bigl(x_{k} + Z \in \set{D}\bigr) = \Prob\bigl(\xi + Z \in \set{D}\bigr).
\end{equation}
\end{lemma}
\begin{proof}
 Let $f(\cdot)$ denote the density of a zero-mean,
 variance-$\sigma^{2}$ Gaussian random variable, so
 \begin{equation*}
    \Prob\bigl(x_{k} + Z \in \set{D}\bigr) = \int_{\set{D}} f(\tilde{y} - x_{k})\d\tilde{y}.
       \end{equation*}
  Since $f(\cdot)$ is continuous, and since the sequence $\{x_{k}\}$ converges to $\xi$, it follows that the sequence of
  densities \mbox{$\tilde{y} \mapsto f(\tilde{y} - x_{k})$} converges to $\tilde{y}
  \mapsto f(\tilde{y} - \xi)$. The result follows then by noting that, for every $k$,
  \begin{equation}
  \Prob\bigl(x_{k} + Z \in \Reals\bigr) = \Prob\bigl(\xi + Z \in \Reals\bigr) = 1
  \end{equation}
  and from Scheffe's Theorem \cite[Th.~16.12]{billingsley95}.
 \end{proof}

From~Lemma~\ref{lem:amos_cont} we conclude that $x \mapsto \Prob(Y=1| X = x)$ is continuous. Since it also bounded, it
follows that $\Prob(Y=1)$ is continuous in the input distribution under
the weak topology. Since the binary entropy function is a
continuous bounded function, this implies that $H(Y)$ is continuous in the input
distribution. By the same lemma, it follows that also the mapping $x \mapsto
H_b\bigl(\Prob(Y=1|X=x)\bigr)$ is continuous and bounded, so
$H(Y|X)$ is also continuous in the input distribution. We thus have the following lemma.
\begin{lemma}
  \label{lem:amos_cont2}
  For every fixed quantizer $\set{D}$, the functionals $H(Y)$,
  $H(Y|X)$, and $I(X;Y)$ are continuous in the input distribution
  under the weak topology.
\end{lemma}

For proving the existence of a capacity-achieving input distribution we need
a compactness result:
\begin{lemma}
\label{lem:amos_cont3}
  Let $\const{A} > 0$ be fixed. Every sequence of probability measures
  on the interval $[-\const{A}, \const{A}]$ of second moment not
  exceeding $\const{P}$ has a subsequence that converges weakly to a
  probability distribution on the interval $[-\const{A}, \const{A}]$
  of second moment not exceeding $\const{P}$.
\end{lemma}
\begin{proof}
By Prokhorov's Theorem, every sequence of probability measures on $[-\const{A},\const{A}]$ has a subsequence that converges weakly to some probability measure on $[-\const{A},
  \const{A}]$. The second moment of this limiting probability measure
  cannot exceed $\const{P}$ because the function $x \mapsto x^{2}$ is
  a continuous bounded function on the interval $[-\const{A},
  \const{A}]$. 
\end{proof}
Note that Lemma~\ref{lem:amos_cont3} continues to hold for sequences of probability measures on $\Reals$ of second moment not exceeding $\const{P}$, albeit with a slightly different proof. Thus, the amplitude constraint $\const{A}$ is not essential.

It follows from Lemmas~\ref{lem:amos_cont}--\ref{lem:amos_cont3} that the supremum in
\eqref{eq:amos_def_CDAP} defining $C_{\set{D},\const{A}}(\const{P})$
is achieved.

\section{}
\label{app:increasing}
We show that, for $\xi<\theta$, the function $\xi\mapsto W\bigl(\Upsilon_1,\Upsilon_2\bigm| \xi\bigr)$ is strictly increasing. To this end, we note that
\begin{equation}
W\bigl(\Upsilon_1,\Upsilon_2\bigm| \xi\bigr) = Q\biggl(\frac{\theta-\Delta-\xi}{\sigma}\biggr)-Q\biggl(\frac{\theta+\Delta-\xi}{\sigma}\biggr)
\end{equation}
and take the derivative with respect to $\xi$. (Recall that \mbox{$\theta=(\Upsilon_1+\Upsilon_2)/2$} and $\Delta=(\Upsilon_2-\Upsilon_1)/2$.) This yields
\begin{IEEEeqnarray}{lCl}
\IEEEeqnarraymulticol{3}{l}{\frac{\partial}{\partial\xi} W\bigl(\Upsilon_1,\Upsilon_2\bigm| \xi\bigr)} \nonumber\\
\quad & = & \frac{1}{\sqrt{2\pi\sigma^2}}e^{-\frac{(\theta-\Delta-\xi)^2}{2\sigma^2}} - \frac{1}{\sqrt{2\pi\sigma^2}} e^{-\frac{(\theta+\Delta-\xi)^2}{2\sigma^2}}\nonumber\\
& = & \frac{1}{\sqrt{2\pi\sigma^2}} e^{-\frac{(\theta-\xi)^2+\Delta^2}{2\sigma^2}}\biggl[e^{\Delta\frac{\theta-\xi}{\sigma^2}}-e^{-\Delta\frac{\theta-\xi}{\sigma^2}}\biggr] \nonumber\\
& > & 0, \qquad \xi<\theta
\end{IEEEeqnarray}
thus proving the claim.

\section{}
\label{app:proof_thm4}
To show that
\begin{equation}
\lim_{\xi\to 0} \frac{\Psi(\xi)}{\xi^2} = \frac{1}{2\sigma^2}\biggl(\frac{1}{2}+\frac{1}{\pi}\biggr)
\end{equation}
we write $\Psi(\xi)$ as
\begin{IEEEeqnarray}{lCl}
\Psi(\xi) & = & \frac{1}{\sqrt{2\pi\sigma^2}} \int_0^{\infty} e^{-\frac{(\tilde{y}-\xi)^2}{2\sigma^2}}\biggl(\frac{\tilde{y}\xi}{\sigma^2}-\frac{\xi^2}{2\sigma^2}\biggr)\d\tilde{y} \nonumber\\
& & {} + Q\biggl(\frac{\xi}{\sigma}\biggr) \log\Biggl(2Q\biggl(\frac{\xi}{\sigma}\biggr)\Biggr)\nonumber\\
& = & \frac{\xi^2}{2\sigma^2}Q\biggl(-\frac{\xi}{\sigma}\biggr)+\frac{\xi}{\sqrt{2\pi\sigma^2}}\biggl(e^{-\frac{\xi^2}{2\sigma^2}}-1\biggr) \nonumber\\
& & {} + \Biggl[Q\biggl(\frac{\xi}{\sigma}\biggr) \log\Biggl(2Q\biggl(\frac{\xi}{\sigma}\biggr)\Biggr) + \frac{\xi}{\sqrt{2\pi\sigma^2}}\Biggr] \IEEEeqnarraynumspace \label{eq:app_proof_thm4_1}
\end{IEEEeqnarray}
and compute the limiting ratio of each term on the RHS of \eqref{eq:app_proof_thm4_1} to $\xi^2$ as $\xi$ tends to zero. For the first two terms, we have
\begin{equation}
\label{eq:app_proof_thm4_2}
\lim_{\xi\to 0} \frac{\frac{\xi^2}{2\sigma^2}Q\Bigl(-\frac{\xi}{\sigma}\Bigr)}{\xi^2} = \frac{1}{4\sigma^2}
\end{equation}
and
\begin{equation}
\label{eq:app_proof_thm4_3}
\lim_{\xi\to 0} \frac{\frac{\xi}{\sqrt{2\pi\sigma^2}}\Bigl(e^{-\frac{\xi^2}{2\sigma^2}}-1\Bigr)}{\xi^2} = 0.
\end{equation}
To evaluate the last term on the RHS of \eqref{eq:app_proof_thm4_1}, we express $\xi\mapsto Q\bigl(\xi/\sigma\bigr)$ as a Taylor series around zero
\begin{equation}
Q\biggl(\frac{\xi}{\sigma}\biggr) = \frac{1}{2} - \frac{\xi}{\sqrt{2\pi\sigma^2}} + o\bigl(\xi^2\bigr).
\end{equation}
With this, we obtain
\begin{IEEEeqnarray}{lCl}
\IEEEeqnarraymulticol{3}{l}{\Biggl[Q\biggl(\frac{\xi}{\sigma}\biggr) \log\Biggl(2Q\biggl(\frac{\xi}{\sigma}\biggr)\Biggr) + \frac{\xi}{\sqrt{2\pi\sigma^2}}\Biggr]}\nonumber\\
\quad & = & \biggl(\frac{1}{2}-\frac{\xi}{\sqrt{2\pi\sigma^2}}+o\bigl(\xi^2\bigr)\biggr) \log\Biggl(1-\frac{\xi}{\sigma}\sqrt{\frac{2}{\pi}}+o\bigl(\xi^2\bigr)\Biggr) \nonumber\\
& & {} +\frac{\xi}{\sqrt{2\pi\sigma^2}} \nonumber\\
& = & \biggl(\frac{1}{2}-\frac{\xi}{\sqrt{2\pi\sigma^2}}+o\bigl(\xi^2\bigr)\biggr) \Biggl(-\frac{\xi}{\sigma}\sqrt{\frac{2}{\pi}}-\frac{\xi^2}{\sigma^2}\frac{1}{\pi}+o\bigl(\xi^2\bigr)\Biggr) \nonumber\\
& & {} +\frac{\xi}{\sqrt{2\pi\sigma^2}} \nonumber\\
& = & \frac{\xi^2}{2\sigma^2} \frac{1}{\pi} + o\bigl(\xi^2\bigr)
\end{IEEEeqnarray}
where the second step follows because
\begin{equation}
\log(1+x)=x-\frac{1}{2}x^2+o\bigl(x^2\bigr).
\end{equation}
Consequently,
\begin{equation}
\lim_{\xi\to 0} \frac{Q\Bigl(\frac{\xi}{\sigma}\Bigr) \log\biggl(2Q\Bigl(\frac{\xi}{\sigma}\Bigr)\biggr) + \frac{\xi}{\sqrt{2\pi\sigma^2}}}{\xi^2} = \frac{1}{2\sigma^2}\frac{1}{\pi}. \label{eq:app_proof_thm4_4}
\end{equation}
The claim follows by combining \eqref{eq:app_proof_thm4_2}--\eqref{eq:app_proof_thm4_4} with \eqref{eq:app_proof_thm4_1}.

\section{}
\label{app:proof_note3}

\subsection{Proof of \eqref{eq:proof_note3_claim1}}
\label{app:proof_note3_182}
To prove \eqref{eq:proof_note3_claim1}, namely
\begin{equation*}
\lim_{\const{A}\downarrow 0} \sup_{\Upsilon\geq 0}\frac{|\mathsf{R}_H(\const{A},\Upsilon,p_+)|}{\const{A}^2} = 0, \quad 0\leq p_+\leq 1
\end{equation*}
we fix some $\nu \geq 1$ and analyze the cases $0\leq\Upsilon\leq\nu$ and $\Upsilon>\nu$ separately. Since we are interested in the limit as $\const{A}$ tends to zero, there is no loss in generality in assuming that $\const{A}\leq 1$. 

If $0\leq \Upsilon\leq\nu$, then $\tilde{p}$ in \eqref{eq:proof_note3_H} is bounded by
\begin{equation}
Q\biggl(\frac{\nu+\const{A}}{\sigma}\biggr) \leq \tilde{p} \leq Q\biggl(-\frac{\const{A}}{\sigma}\biggr)
\end{equation}
which, by the assumption $\const{A}\leq 1$, implies that $\tilde{p}$ is bounded away from $0$ and $1$:
\begin{equation}
\label{eq:proof_note3_Upsleqnu}
Q\biggl(\frac{\nu+1}{\sigma}\biggr) \leq \tilde{p} \leq Q\biggl(-\frac{1}{\sigma}\biggr).
\end{equation}
Consequently, combining \eqref{eq:proof_note3_Q} with \eqref{eq:proof_note3_H} and using the Triangle Inequality yields for $0\leq\Upsilon\leq\nu$
\begin{IEEEeqnarray}{lCl}
\IEEEeqnarraymulticol{3}{l}{\bigl|\mathsf{R}_H(\const{A},\Upsilon,p_+)\bigr|}\nonumber\\
\quad & \leq & \biggl[ \frac{\const{A}}{2\sigma}\frac{|p_+-p_-|}{\sqrt{2\pi}}e^{-\frac{\Upsilon^2}{2\sigma^2}}+\bigl|\mathsf{R}_Q(\const{A},\Upsilon,p_+)\bigr|\biggr]^3 \frac{|1-2\tilde{p}|}{\tilde{p}^2(1-\tilde{p})^2} \nonumber\\
& \leq & \biggl[\frac{\const{A}}{2\sigma}\frac{1}{\sqrt{2\pi}}e^{-\frac{\Upsilon^2}{2\sigma^2}}+\bigl|\mathsf{R}_Q(\const{A},\Upsilon,p_+)\bigr|\biggr]^3 \frac{1}{\tilde{p}^2(1-\tilde{p})^2}  \nonumber\\
& \leq & \frac{ \const{A}^3\biggl[\frac{1}{2\sigma\sqrt{2\pi}} + \frac{\const{A}}{2\sigma^2\sqrt{2\pi e}}\biggr]^3}{\Bigl[Q\bigl(\frac{\nu+1}{\sigma}\bigr)\Bigl(1-Q\bigl(-\frac{1}{\sigma}\bigr)\Bigr)\Bigr]^2}. \label{eq:proof_note3_Usm_2}\IEEEeqnarraynumspace
\end{IEEEeqnarray}
Here the second step follows by upper-bounding $|1-2\tilde{p}|\leq 1$ and $|p_+-p_-|\leq 1$; and the third step follows from \eqref{eq:proof_note3_RQ} and \eqref{eq:proof_note3_Upsleqnu} and by upper-bounding $\exp\bigl(-\Upsilon^2/(2\sigma^2)\bigr)\leq 1$. Since the RHS of \eqref{eq:proof_note3_Usm_2} does not depend on $\Upsilon$, this yields
\begin{equation}
\label{eq:proof_note3_Usm}
\lim_{\const{A}\downarrow 0} \sup_{0\leq\Upsilon\leq\nu}\frac{|\mathsf{R}_H(\const{A},\Upsilon,p_+)|}{\const{A}^2} = 0, \quad 0\leq p_+\leq 1.
\end{equation}

For $\Upsilon>\nu$, we first upper-bound \eqref{eq:proof_note3_RQ_def} as
\begin{IEEEeqnarray}{lCl}
\bigl|\mathsf{R}_Q(\const{A},\Upsilon,p_+)\bigr| & \leq & \frac{\const{A}^2}{2\sigma^2}\frac{\Upsilon+\const{A}}{\sqrt{2\pi\sigma^2}} e^{-\frac{(\Upsilon-\const{A})^2}{2\sigma^2}} \nonumber\\
& \leq & \frac{\const{A}^2}{\sigma^2}\frac{\Upsilon}{\sqrt{2\pi\sigma^2}} e^{-\frac{(\Upsilon-1)^2}{2\sigma^2}} \label{eq:proof_note3_RQ2}
\end{IEEEeqnarray}
where the first step follows by upper-bounding $\tilde{x}\leq \Upsilon+\const{A}$ and $\exp\bigl(-\tilde{x}^2/(2\sigma^2)\bigr)\leq \exp\bigl(-(\Upsilon-\const{A})^2/(2\sigma^2)\bigr)$; and the second step follows because $\Upsilon>\nu$ and $\const{A}\leq 1$, so $\const{A}<\Upsilon$. Combining \eqref{eq:proof_note3_RQ2} with \eqref{eq:proof_note3_H} yields for $\Upsilon>\nu$
\begin{IEEEeqnarray}{lCl}
\IEEEeqnarraymulticol{3}{l}{\bigl|\mathsf{R}_H(\const{A},\Upsilon,p_+)\bigr|} \nonumber\\
\quad & \leq & \biggl[\frac{\const{A}}{2\sigma}\frac{|p_+-p_-|}{\sqrt{2\pi}}e^{-\frac{\Upsilon^2}{2\sigma^2}}+\bigl|\mathsf{R}_Q(\const{A},\Upsilon,p_+)\bigr|\biggr]^3 \frac{|1-2\tilde{p}|}{\tilde{p}^2(1-\tilde{p})^2} \nonumber\\
& \leq & \biggl[\frac{\const{A}}{2\sigma}\frac{1}{\sqrt{2\pi}}e^{-\frac{\Upsilon^2}{2\sigma^2}}+\frac{\const{A}^2}{\sigma^2}\frac{\Upsilon}{\sqrt{2\pi\sigma^2}} e^{-\frac{(\Upsilon-1)^2}{2\sigma^2}}\biggr]^3 \frac{1}{\tilde{p}^2(1-\tilde{p})^2} \nonumber\\
& \leq & \frac{\Upsilon^3}{(2\pi\sigma^2)^{\frac{3}{2}}} e^{-\frac{3(\Upsilon-1)^2}{2\sigma^2}}\biggl[\frac{\const{A}}{2}+\frac{\const{A}^2}{\sigma^2}\biggr]^3\frac{1}{\tilde{p}^2(1-\tilde{p})^2} \label{eq:proof_note3_1}
\end{IEEEeqnarray}
where the first step follows from the Triangle Inequality; the second step follows from \eqref{eq:proof_note3_RQ2} and because \mbox{$|p_+-p_-|\leq 1$} and $|1-2\tilde{p}|\leq 1$; and the last step follows because $\exp\bigl(-\Upsilon^2/(2\sigma^2)\bigr)\leq\Upsilon \exp\bigl(-(\Upsilon-1)^2/(2\sigma^2)\bigr)$ for $\Upsilon>1$.

We next note that, since $\Upsilon>\nu>\const{A}$ and $0\leq\const{A}\leq 1$, we have
\begin{equation}
\label{eq:proof_note3_ptilde_U}
\tilde{p} \leq Q\biggl(\frac{\Upsilon-\const{A}}{\sigma}\biggr) < \frac{1}{2}
\end{equation}
and
\begin{IEEEeqnarray}{lCl}
\tilde{p} & \geq & Q\biggl(\frac{\Upsilon+\const{A}}{\sigma}\biggr) \nonumber\\
& > & \biggl(1-\frac{\sigma^2}{(\Upsilon+\const{A})^2}\biggr) \frac{\sigma}{\sqrt{2\pi}(\Upsilon+\const{A})} e^{-\frac{(\Upsilon+\const{A})^2}{2\sigma^2}} \nonumber\\\
& > &  \biggl(1-\frac{\sigma^2}{\nu^2}\biggr) \frac{\sigma}{\sqrt{2\pi}(\Upsilon+1)} e^{-\frac{(\Upsilon+1)^2}{2\sigma^2}}, \quad \Upsilon>\nu \IEEEeqnarraynumspace\label{eq:proof_note3_ptilde_L}
\end{IEEEeqnarray}
where the second step follows from \cite[Prop.~19.4.2]{lapidoth09}. Consequently, using \eqref{eq:proof_note3_ptilde_U} and \eqref{eq:proof_note3_ptilde_L}, the RHS of \eqref{eq:proof_note3_1} can be upper-bounded by
\begin{IEEEeqnarray}{lCl}
\IEEEeqnarraymulticol{3}{l}{\bigl|\mathsf{R}_H(\const{A},\Upsilon,p_+)\bigr| }\nonumber\\
\quad & \leq & \frac{\Upsilon^3}{(2\pi\sigma^2)^{\frac{3}{2}}} e^{-\frac{3(\Upsilon-1)^2}{2\sigma^2}}\frac{4\biggl[\frac{\const{A}}{2}+\frac{\const{A}^2}{\sigma^2}\biggr]^3}{\Bigl(1-\frac{\sigma^2}{\nu^2}\Bigr)^2\frac{\sigma^2}{2\pi(\Upsilon+1)^2}e^{-\frac{(\Upsilon+1)^2}{\sigma^2}}} \nonumber\\
& = &  \frac{4\biggl[\frac{\const{A}}{2}+\frac{\const{A}^2}{\sigma^2}\biggr]^3}{\sqrt{2\pi}\sigma^5 \Bigl(1-\frac{\sigma^2}{\nu^2}\Bigr)^2} \Upsilon^3 (\Upsilon+1)^2 \times\nonumber\\
& & {} \times \exp\biggl(-\frac{3(\Upsilon-1)^2}{2\sigma^2}+\frac{(\Upsilon+1)^2}{\sigma^2}\biggr), \quad \Upsilon>\nu. \label{eq:proof_note3_2}\IEEEeqnarraynumspace
\end{IEEEeqnarray}
Since the function
\begin{equation*}
\Upsilon \mapsto \Upsilon^3 (\Upsilon+1)^2 \exp\biggl(-\frac{3(\Upsilon-1)^2}{2\sigma^2}+\frac{(\Upsilon+1)^2}{\sigma^2}\biggr)
\end{equation*}
is bounded in $\Upsilon>\nu$, this yields
\begin{equation}
\label{eq:proof_note3_Ula}
\lim_{\const{A}\downarrow 0} \sup_{\Upsilon>\nu}\frac{|\mathsf{R}_H(\const{A},\Upsilon,p_+)|}{\const{A}^2} = 0, \quad 0\leq p_+\leq 1.
\end{equation}
Combining \eqref{eq:proof_note3_Usm} and \eqref{eq:proof_note3_Ula} proves \eqref{eq:proof_note3_claim1}.

\subsection{Proof of \eqref{eq:proof_note3_claim2}}
\label{app:proof_note3_184}
To prove \eqref{eq:proof_note3_claim2}, namely
\begin{equation*}
\lim_{\const{A} \downarrow 0} \sup_{\Upsilon\geq 0} \frac{|\mathsf{K}(\const{A},\Upsilon,p_+)|}{\const{A}^2} = 0, \quad  0\leq p_+\leq 1
\end{equation*}
we fix some $\nu \geq 1$ and analyze the cases $0\leq\Upsilon\leq\nu$ and $\Upsilon>\nu$ separately. Without loss of generality, we assume that $\const{A}\leq 1$. If $0\leq\Upsilon\leq\nu$, then we have
\begin{equation}
\label{eq:app_note3_1}
Q\biggl(\frac{\nu}{\sigma}\biggr) \leq Q\biggl(\frac{\Upsilon}{\sigma}\biggr) \leq \frac{1}{2}
\end{equation}
which yields for every $0\leq p_+\leq 1$ and every $\const{A}\leq 1$
\begin{IEEEeqnarray}{lCl}
\IEEEeqnarraymulticol{3}{l}{\bigl|\mathsf{K}(\const{A},\Upsilon,p_+)\bigr|}\nonumber\\
\quad & = & \frac{\Bigl|\frac{\const{A}}{\sigma}\frac{2(p_+-p_-)}{\sqrt{2\pi}}e^{-\frac{\Upsilon^2}{2\sigma^2}} \mathsf{R}_Q(\const{A},\Upsilon,p_+) + \bigl|\mathsf{R}_Q(\const{A},\Upsilon,p_+)\bigr|^2\Bigr|}{2Q\bigl(\frac{\Upsilon}{\sigma}\bigr)\bigl[1-Q\bigl(\frac{\Upsilon}{\sigma}\bigr)\bigr]} \nonumber\\
& \leq & \frac{\frac{\const{A}}{\sigma}\frac{2|p_+-p_-|}{\sqrt{2\pi}} \bigl|\mathsf{R}_Q(\const{A},\Upsilon,p_+)\bigr| + \bigl|\mathsf{R}_Q(\const{A},\Upsilon,p_+)\bigr|^2}{Q\bigl(\frac{\nu}{\sigma}\bigr)}\nonumber\\
& \leq & \frac{1}{Q\bigl(\frac{\nu}{\sigma}\bigr)}\biggl[\frac{\const{A}^3}{\sigma^3}\frac{1}{2\pi\sqrt{e}} + \frac{\const{A}^4}{4\sigma^4 2 \pi e}\biggr], \quad 0\leq\Upsilon\leq \nu. \label{eq:app_note3_Usm_2}
\end{IEEEeqnarray}
Here the second step follows from \eqref{eq:app_note3_1}, from the upper bound $\exp\bigl(-\Upsilon^2/(2\sigma^2)\bigr)\leq 1$, $\Upsilon\in\Reals$, and from the Triangle Inequality; and the third step follows from \eqref{eq:proof_note3_RQ} and  because $|p_+-p_-|\leq 1$. Consequently,
\begin{equation}
\label{eq:app_note3_Usm}
\lim_{\const{A} \downarrow 0} \sup_{0\leq\Upsilon\leq \nu} \frac{|\mathsf{K}(\const{A},\Upsilon,p_+)|}{\const{A}^2} = 0, \quad  0\leq p_+\leq 1.
\end{equation}

If $\Upsilon> \nu$, then we have \cite[Prop.~19.4.2]{lapidoth09}
\begin{equation}
\label{eq:app_note3_2}
\frac{\sigma}{\sqrt{2\pi}\Upsilon} \biggl(1-\frac{\sigma^2}{\nu^2}\biggr) e^{-\frac{\Upsilon^2}{2\sigma^2}} < Q\biggl(\frac{\Upsilon}{\sigma}\biggr) < \frac{1}{2}
\end{equation}
and, by \eqref{eq:proof_note3_RQ2},
\begin{equation}
\bigl|\mathsf{R}_Q(\const{A},\Upsilon,p_+)\bigr| \leq  \frac{\const{A}^2}{\sigma^2}\frac{\Upsilon}{\sqrt{2\pi\sigma^2}} e^{-\frac{(\Upsilon-1)^2}{2\sigma^2}}, \quad \Upsilon>\nu. \label{eq:proof_note3_RQ2_2012}
\end{equation}
We thus obtain for $\Upsilon>\nu$
\begin{IEEEeqnarray}{lCl}
\IEEEeqnarraymulticol{3}{l}{\bigl|\mathsf{K}(\const{A},\Upsilon,p_+)\bigr|} \nonumber\\
\,\,\, & = & \frac{\Bigl|\frac{\const{A}}{\sigma}\frac{2(p_+-p_-)}{\sqrt{2\pi}}e^{-\frac{\Upsilon^2}{2\sigma^2}} \mathsf{R}_Q(\const{A},\Upsilon,p_+) + \bigl|\mathsf{R}_Q(\const{A},\Upsilon,p_+)\bigr|^2\Bigr|}{2Q\bigl(\frac{\Upsilon}{\sigma}\bigr)\bigl[1-Q\bigl(\frac{\Upsilon}{\sigma}\bigr)\bigr]} \nonumber\\
& \leq & \frac{\sqrt{2\pi} \Upsilon e^{\frac{\Upsilon^2}{2\sigma^2}}}{\sigma \bigl(1-\frac{\sigma^2}{\nu^2}\bigr)}\biggl[\frac{\const{A}}{\sigma}\frac{2|p_+-p_-|}{\sqrt{2\pi}}e^{-\frac{\Upsilon^2}{2\sigma^2}} \bigl|\mathsf{R}_Q(\const{A},\Upsilon,p_+)\bigr| \nonumber\\
& & \quad\qquad\qquad\qquad\qquad\qquad\qquad {} + \bigl|\mathsf{R}_Q(\const{A},\Upsilon,p_+)\bigr|^2\biggr] \nonumber\\
& \leq &  \frac{\sqrt{2\pi} \Upsilon e^{\frac{\Upsilon^2}{2\sigma^2}}}{\sigma \bigl(1-\frac{\sigma^2}{\nu^2}\bigr)}\biggl[\frac{\const{A}^3}{\sigma^3\pi}\frac{\Upsilon}{\sigma}e^{-\frac{\Upsilon^2}{2\sigma^2}-\frac{(\Upsilon-1)^2}{2\sigma^2}}  + \frac{\const{A}^4}{\sigma^4 2\pi}\frac{\Upsilon^2}{\sigma^2}e^{-\frac{(\Upsilon-1)^2}{\sigma^2}}\biggr] \nonumber\\
& \leq & \sqrt{\frac{2}{\pi}}\frac{1}{\bigl(1-\frac{\sigma^2}{\nu^2}\bigr)} \frac{\Upsilon^2}{\sigma^2} e^{\frac{\Upsilon^2}{2\sigma^2}-\frac{(\Upsilon-1)^2}{\sigma^2}}\biggl[1+\frac{\const{A}}{2\sigma}\frac{\Upsilon}{\sigma}\biggr]\frac{\const{A}^3}{\sigma^3} \nonumber\\
& \leq & \sqrt{\frac{2}{\pi}} \frac{1}{\bigl(1-\frac{\sigma^2}{\nu^2}\bigr)} \frac{\Upsilon^2}{\sigma^2} e^{\frac{\Upsilon^2}{2\sigma^2}-\frac{(\Upsilon-1)^2}{\sigma^2}}\biggl[1+\frac{\Upsilon^2}{2\sigma^2}\biggr]\frac{\const{A}^3}{\sigma^3} \label{eq:app_note3_almost}
\end{IEEEeqnarray}
where the second step follows from \eqref{eq:app_note3_2} and from the Triangle Inequality; the third step follows from \eqref{eq:proof_note3_RQ2_2012} and because $|p_+-p_-|\leq 1$; the fourth step follows by upper-bounding $\exp\bigl(-\Upsilon^2/(2\sigma^2)\bigr)\leq \exp\bigl((\Upsilon-1)^2/(2\sigma^2)\bigr)$; and the last step follows because $\Upsilon>\nu$ and $\const{A}\leq 1$, so $\const{A}\leq\Upsilon$.

Since the function
\begin{equation*}
\Upsilon \mapsto \frac{\Upsilon^2}{\sigma^2} e^{\frac{\Upsilon^2}{2\sigma^2}-\frac{(\Upsilon-1)^2}{\sigma^2}}\biggl[1+\frac{\Upsilon^2}{2\sigma^2}\biggr]
\end{equation*}
is bounded in $\Upsilon>\nu$, this yields 
\begin{equation}
\label{eq:app_note3_Ula}
\lim_{\const{A} \downarrow 0} \sup_{\Upsilon>\nu} \frac{|\mathsf{K}(\const{A},\Upsilon,p_+)|}{\const{A}^2} = 0, \quad  0\leq p_+\leq 1.
\end{equation}
Combining \eqref{eq:app_note3_Usm} and \eqref{eq:app_note3_Ula} proves \eqref{eq:proof_note3_claim2}.

\section*{Acknowledgment}
The authors wish to thank Paul~P.~Sotiriadis, who sparked their interest in the problem of quantization. They further wish to thank Tam\'as Linder, Alfonso~Martinez, and Sergio~Verd\'u for enlightening discussions and the Associate Editor Young-Han Kim and the anonymous referees for their valuable comments.


\begin{thebibliography}{10}
\providecommand{\url}[1]{#1}
\csname url@samestyle\endcsname
\providecommand{\newblock}{\relax}
\providecommand{\bibinfo}[2]{#2}
\providecommand{\BIBentrySTDinterwordspacing}{\spaceskip=0pt\relax}
\providecommand{\BIBentryALTinterwordstretchfactor}{4}
\providecommand{\BIBentryALTinterwordspacing}{\spaceskip=\fontdimen2\font plus
\BIBentryALTinterwordstretchfactor\fontdimen3\font minus
  \fontdimen4\font\relax}
\providecommand{\BIBforeignlanguage}[2]{{%
\expandafter\ifx\csname l@#1\endcsname\relax
\typeout{** WARNING: IEEEtran.bst: No hyphenation pattern has been}%
\typeout{** loaded for the language `#1'. Using the pattern for}%
\typeout{** the default language instead.}%
\else
\language=\csname l@#1\endcsname
\fi
#2}}
\providecommand{\BIBdecl}{\relax}
\BIBdecl

\bibitem{walden99}
R.~H. Walden, ``Analog-to-digital converter survey and analysis,'' \emph{IEEE
  J. Select. Areas Commun.}, vol.~17, no.~4, pp. 539--550, Apr. 1999.

\bibitem{viterbiomura79}
A.~J. Viterbi and J.~K. Omura, \emph{Principles of Digital Communication and
  Coding}.\hskip 1em plus 0.5em minus 0.4em\relax McGraw-Hill, 1979.

\bibitem{verdu02}
S.~Verd\'u, ``Spectral efficiency in the wideband regime,'' \emph{IEEE Trans.
  Inf. Theory}, vol.~48, no.~6, pp. 1319--1343, June 2002.

\bibitem{shannon48}
C.~E. Shannon, ``A mathematical theory of communication,'' \emph{Bell System
  Techn. J.}, vol.~27, pp. 379--423 and 623--656, July and Oct. 1948.

\bibitem{gallager68}
R.~G. Gallager, \emph{Information Theory and Reliable Communication}.\hskip 1em
  plus 0.5em minus 0.4em\relax John Wiley \& Sons, 1968.

\bibitem{kochlapidoth10_2}
T.~Koch and A.~Lapidoth, ``Increased capacity per unit-cost by oversampling,''
  in \emph{Proc. IEEE 26th Conv. of Electrical and Electronics Eng. in Israel},
  2010, pp. 684--688.

\bibitem{kochlapidoth10_1_arxiv}
\BIBentryALTinterwordspacing
\rule[0.5ex]{2em}{0.2pt}, ``Increased capacity per unit-cost by oversampling,'' Sept. 2010.
  [Online]. Available: \url{http://arxiv.org/abs/1008.5393}
\BIBentrySTDinterwordspacing

\bibitem{gilbert93}
E.~N. Gilbert, ``Increased information rate by oversampling,'' \emph{IEEE
  Trans. Inf. Theory}, vol.~39, pp. 1973--1976, Nov. 1993.

\bibitem{shamai94}
S.~Shamai~(Shitz), ``Information rates by oversampling the sign of a
  bandlimited process,'' \emph{IEEE Trans. Inf. Theory}, vol.~40, pp.
  1230--1236, July 1994.

\bibitem{coverthomas91}
T.~M. Cover and J.~A. Thomas, \emph{Elements of Information Theory},
  1st~ed.\hskip 1em plus 0.5em minus 0.4em\relax John Wiley \& Sons, 1991.

\bibitem{verdu90}
S.~Verd\'u, ``On channel capacity per unit cost,'' \emph{IEEE Trans. Inf.
  Theory}, vol.~36, pp. 1019--1030, Sept. 1990.

\bibitem{singhdabeermadhow09_2}
J.~Singh, O.~Dabeer, and U.~Madhow, ``On the limits of communication with
  low-precision analog-to-digital conversion at the receiver,'' \emph{IEEE
  Trans. Commun.}, vol.~57, no.~12, pp. 3629--3639, Dec. 2009.

\bibitem{grafluschgy00}
S.~Graf and H.~Luschgy, \emph{Foundations of Quantization for Probability
  Distributions}, ser. Lecture Notes in Mathematics.\hskip 1em plus 0.5em minus
  0.4em\relax Springer Verlag, 2000, vol. 1730.

\bibitem{lapidoth09}
A.~Lapidoth, \emph{A Foundation in Digital Communication}.\hskip 1em plus 0.5em
  minus 0.4em\relax Cambridge University Press, 2009.

\bibitem{wozencraftjacobs65}
J.~M. Wozencraft and I.~M. Jacobs, \emph{Principles of Communication
  Engineering}.\hskip 1em plus 0.5em minus 0.4em\relax John Wiley \& Sons,
  1965.

\bibitem{zhangwillemshuang11}
P.~Zhang, F.~M.~J. Willems, and L.~Huang, ``Investigations of noncoherent {OOK}
  based schemes with soft and hard decisions for {WSN}s,'' in \emph{Proc. 49th
  Allerton Conf. Comm., Contr. and Comp.}, Allerton H., Monticello, Il, Sept.
  28--30, 2011, pp. 1702--1709.

\bibitem{lapidothshamai02}
A.~Lapidoth and S.~Shamai~(Shitz), ``Fading channels: how perfect need
  {`}perfect side-information{'} be{?}'' \emph{IEEE Trans. Inf. Theory},
  vol.~48, no.~5, pp. 1118--1134, May 2002.

\bibitem{mezghaninossek07}
A.~Mezghani and J.~A. Nossek, ``On ultra-wideband {MIMO} systems with 1-bit
  quantized outputs: Performance analysis and input optimization,'' in
  \emph{Proc. IEEE Int. Symp. Inf. Theory}, Nice, France, June 24--29,
  2007, pp. 1286--1289.

\bibitem{mezghaninossek08}
\rule[0.5ex]{2em}{0.2pt}, ``Analysis of {R}ayleigh-fading channels with 1-bit quantized output,''
  in \emph{Proc. IEEE Int. Symp. Inf. Theory}, Toronto, Canada, July
  6--11, 2008, pp. 260--264.

\bibitem{mezghaninossek09}
\rule[0.5ex]{2em}{0.2pt}, ``Analysis of 1-bit output noncoherent fading channels in the low {SNR}
  regime,'' in \emph{Proc. IEEE Int. Symp. Inf. Theory}, Seoul, Korea,
  June 28 -- July 3, 2009, pp. 1080--1084.

\bibitem{kronefettweis10}
S.~Krone and G.~Fettweis, ``Fading channels with 1-bit output quantization:
  Optimal modulation, ergodic capacity and outage probability,'' in \emph{Proc.
  Inf. Theory Workshop (ITW)}, Dublin, Ireland, Aug. 30 -- Sept. 3, 2010,
  pp. 1--5.

\bibitem{kochlapidoth12_1}
T.~Koch and A.~Lapidoth, ``One-bit quantizers for fading channels,'' in
  \emph{Proc. IZS}, Zurich, Switzerland, Feb. 29 -- Mar. 2, 2012, pp. 36--39.

\bibitem{witsenhausen80}
H.~S. Witsenhausen, ``Some aspects of convexity useful in information theory,''
  \emph{IEEE Trans. Inf. Theory}, vol.~26, no.~3, pp. 265--271, May 1980.

\bibitem{rockafellar70}
R.~T. Rockafellar, \emph{Convex Analysis}.\hskip 1em plus 0.5em minus
  0.4em\relax Princeton University Press, 1970.

\bibitem{simon02}
M.~K. Simon, \emph{Probability Distributions Involving {G}aussian Random
  Variables: A Handbook for Engineers and Scientists}.\hskip 1em plus 0.5em
  minus 0.4em\relax Kluwer Academic Publishers, 2002.

\bibitem{neymanpearson32}
J.~Neyman and E.~Pearson, ``On the problem of the most efficient test of
  statistical hypotheses,'' \emph{Phil. Trans. R. Soc. Lond. A}, vol. 231, no.
  694--706, pp. 289--337, Jan. 1932.

\bibitem{billingsley95}
P.~Billingsley, \emph{Probability and Measure}, 3rd~ed., ser. Wiley Series in
  Probability and Mathematical Statistics: Probability and Mathematical
  Statistics.\hskip 1em plus 0.5em minus 0.4em\relax John Wiley \& Sons, 1995.

\end{thebibliography}


\IEEEtriggeratref{14}


\end{document}